\documentclass[a4paper,11pt]{article}

\usepackage{fullpage}
\usepackage{times}
\usepackage{soul}
\usepackage{url}
\usepackage[utf8]{inputenc}
\usepackage[small]{caption}
\usepackage{graphicx}
\usepackage{xcolor}
\usepackage{amsmath}
\usepackage{booktabs}
\usepackage{algorithm}
\usepackage{enumerate}
\usepackage{dsfont}
\usepackage{tikz}
\usepackage{natbib}
\usepackage{hyperref}
\hypersetup{
    unicode=false,          % non-Latin characters in Acrobat’s bookmarks
    colorlinks=true,        % false: boxed links; true: colored links
    linkcolor=red,          % color of internal links (change box color with linkbordercolor)
    citecolor=blue,    % color of links to bibliography
    filecolor=magenta,      % color of file links
    urlcolor=blue           % color of external links
}
 
\usepackage{breakurl}

\usepackage{subcaption} 
\usepackage{todonotes}
\usepackage{etoc}
\usepackage{xspace,xcolor}
\usepackage{amsthm}
\usepackage{amssymb}
\usepackage{algpseudocode}
\etocdepthtag.toc{main}

\usepackage{bbm}

\newtheorem{theorem}{Theorem}[section]

\newtheorem{proposition}[theorem]{Proposition}
\newtheorem{lemma}[theorem]{Lemma}

\theoremstyle{definition}
\newtheorem{definition}[theorem]{Definition}

\newcommand{\OPT}{\mathrm{OPT}}

\newcommand{\eps}{\varepsilon}

\usepackage[capitalize, nameinlink]{cleveref}

\usepackage{mathtools}

\usepackage{todonotes}

\allowdisplaybreaks

\usepackage{authblk}

\title{\bf Efficient Ensemble Selection from Binary and \\ Pairwise Feedback}

\author[1,*]{Tzeh Yuan Neoh}
\author[2,*]{Nicholas Teh}
\author[1]{Je Qin Chooi}
\author[2]{Paul W. Goldberg}
\author[1]{Milind Tambe}

\affil[1]{Harvard University, USA}
\affil[2]{University of Oxford, UK}

\date{\vspace{-10mm}}
\begin{document}

\maketitle

\begin{abstract}
    Organizations increasingly deploy multiple AI systems across task domains, but selecting a small, high-performing ensemble can require costly model calls, benchmark runs, and human evaluation. We study this selection problem as a distributional variant of multiwinner voting: tasks are drawn from an unknown domain distribution, each task induces feedback over candidate experts, and a committee's value on a task is determined by its best-performing member. We analyze both binary feedback, for tasks with correct/incorrect outcomes, and pairwise feedback, for tasks where candidate outputs are compared by preference. In the binary setting, the induced objective is coverage. We give exhaustive-elicitation baselines and matching worst-case query lower bounds, and we design a failure-conditioned greedy algorithm that preserves the standard $(1-1/e)$ guarantee while obtaining instance-dependent query savings. In the pairwise setting, we study $\theta$-winning committees. We show that full-information optimization admits a PTAS but no EPTAS under Gap-ETH, and that the objective is monotone but not submodular. This motivates a weighted ordinal coverage relaxation, which is submodular and supports a failure-conditioned greedy oracle under pairwise feedback. We then convert this oracle back into $\theta$-type guarantees through finite-family auditing or a minimax wrapper. We also provide small-scale LLM experiments illustrating the predicted query savings and the role of complementarity in committee selection.
\end{abstract}

\def\thefootnote{*}\footnotetext{These authors contributed equally to this work and are co-first authors.}

\section{Introduction}
Organizations are increasingly incorporating LLMs into everyday workflows across business functions \citep{mckinsey2025stateai}. 
However, choosing the right AI system is non-trivial: required capabilities and organizational priorities are multifaceted, and no single model dominates across all criteria~\citep{ni2025survey,wang2023decodingtrust}. 
Moreover, the relevant choice set is often much larger than the set of base models, since a single model can give rise to many candidate systems through different deployment choices~\citep{khattab2023dspy}. Selecting a system mainly by reputation, or without task-specific testing, can therefore mismatch the system to the use case, with failures appearing only after deployment~\citep{aws2025bedrockevaluation}. 
The problem is further complicated by the fact that many organizations now deploy multiple such systems across use cases, rather than relying on a single model~\citep{a16z2025enterprise}.
Motivated by these developments, we study how to choose an ensemble
$S\subseteq C$ from a set $C$ of $m$ experts for a domain $D$ of
tasks. We restrict attention to ensembles of size at most $k$. This
budget captures deployment-time costs: larger ensembles may require more
model calls, more routing or scoring, and more downstream validation. Our
query complexity results concern only the pre-deployment selection phase.
Human judgments are counted only when they are used to evaluate candidate
outputs during this phase; we do not assume human labeling after deployment.

In our model, the value of an ensemble on a task is the value of its best
member on that task. This captures objective tasks, such as proof
construction or error detection, where candidate outputs can be checked for
correctness. It also captures subjective tasks, such as writing or image
design, where the user may care mainly about the best output produced by the
ensemble.

To identify a suitable ensemble, we model an internal trial phase in which tasks are drawn from the domain distribution $D$ and candidate experts are evaluated on these tasks. A query is an evaluation action: in the binary model it reveals whether a candidate solves a sampled task, and in the pairwise model it reveals which of two candidate outputs is preferred on a sampled task. We aim to minimize the number of such queries needed to select a high-performing ensemble, since they correspond to selection-time costs such as model calls, benchmark execution, and human or automated assessment.

\medskip
\noindent
\textbf{Our approach and contributions.} 
We use multiwinner voting as a framework for ensemble selection. 
In classical multiwinner voting, a finite set of voters reports preferences
over candidates, and the rule selects a committee of size $k$~\citep{brandt2016handbook}.
Here, a voter is replaced by a task drawn from $D$. The task population may be large or infinite, and the algorithm observes only sampled tasks together
with the feedback it chooses to query. This gives a sampled-task version of
multiwinner voting with partial feedback. To our knowledge, this exact model
has not been studied before; see Section~\ref{sec:relatedworks} for nearby and concurrent work.

The distinction between objective and subjective tasks matches two standard
ballot types in computational social choice~\citep{brandt2016handbook}. Objective tasks resemble
approval ballots: an expert either solves the task or does not. Subjective
generation tasks resemble ordinal ballots: each task induces a ranking over
experts. This mirrors common AI evaluation practice. Benchmark evaluations
such as MMLU-Pro score models by correctness~\citep{wang2024mmlu}, while preference evaluations
such as Chatbot Arena use pairwise comparisons~\citep{chiang2024chatbotarena}.
This gives rise to the following research questions (RQ):

% The objective is not unique a priori; different deployment semantics could lead to different formalizations. We focus on the objectives induced by the best-expert ensemble semantics above. Under binary feedback, a task is covered if at least one selected expert solves it. Under pairwise feedback, a task is covered relative to a rival if at least one selected expert is ranked at least as highly as that rival. This leads to the following questions:
\begin{itemize}
    \item[\textbf{RQ1.}] How should we formalize the notion of a good ensemble for a task domain?
    \item[\textbf{RQ2.}] Given full information about a domain, how well can an optimal ensemble be approximated efficiently?
    \item[\textbf{RQ3.}] Given only sampled tasks from the domain, how does the number of samples affect the approximation guarantees we can obtain, independently of computational constraints?
\end{itemize}

In Section~\ref{sec:approval}, we study binary feedback. The objective is coverage: the probability that at least one selected expert solves a task drawn from $D$.
This is the approval Chamberlin-Courant  objective~\citep{SkowronFaliszewski2017} on sampled tasks. The finite fully observed problem is \text{MaxCover}, so greedy gives the standard $(1-1/e)$-approximation~\citep{NemhauserWolseyFisher1978}, and this ratio is tight unless $\mathrm{P}=\mathrm{NP}$~\citep{Feige1998}. Our main contribution in this setting is a query-efficient greedy algorithm: it focuses candidate queries on tasks missed by the current committee and therefore saves queries on favorable instances while matching the usual worst-case limits.

In Section~\ref{sec:pairwise}, we study pairwise feedback. Since majority cycles can rule out a deterministic committee that beats every other size-$k$ committee, we use $\theta$-winning committees, which compare a committee against each outside candidate. With full rankings, optimizing $\theta$ admits a PTAS but no EPTAS under Gap-ETH. The function $\theta$ is monotone but not submodular, so greedy does not directly apply. We therefore introduce a weighted ordinal coverage objective $\Phi_\lambda$, which is submodular for every rival distribution $\lambda$. We then convert this fixed-$\lambda$ oracle back into $\theta$-type guarantees by auditing finite families of committees or
by using a minimax wrapper.

In Section~\ref{sec:experiments}, experiments on multilingual QA and score-derived pairwise LLM evaluations, covering both the binary and ordinal settings, show that failure-conditioned methods match or improve budgeted baselines while using queries on the residual failures of the current committee, as predicted by the instance-dependent theory. They also show that selecting ensembles based on complementarity can
significantly outperform top-$k$ individual-model baselines, echoing the idea
that diversity can beat individual strength in team
selection~\citep{hong2004groups,licalzi2012power}. This is especially relevant
for arena-style AI leaderboards \citep{chiang2024chatbotarena}, which rank individual models rather than
complementary teams.
\section{Related Work}\label{sec:relatedworks}

\textbf{Committee selection from binary feedback.}
Our binary feedback model is closely related to approval-based committee voting, especially the approval Chamberlin-Courant objective: choose a size-$k$ committee that represents as many voters as possible
\citep{chamberlin1983representative,lackner2023multiwinner,SkowronFaliszewski2017}.
Equivalently, it is a maximum-coverage problem, for which the greedy algorithm gives the standard $(1-1/e)$ approximation and this factor is essentially tight under standard complexity assumptions
\citep{Feige1998,NemhauserWolseyFisher1978}.
The key difference in our
setting is that the ``voters'' are not a fixed, fully observed population.
They are tasks drawn from an unknown distribution, and observing whether an
expert solves a task requires a query.  Thus the main challenge is not only
computational, but also statistical and query-efficient: we must decide which
expert-task outcomes to observe in order to select a good committee.

Prior work has studied committee selection when ballots are incomplete or only
partly revealed \citep{halpern2026representation,imber2025approval,lindeboom2025diverse,lu2013multiwinner}.
These works usually start from a fixed election and ask what can be inferred
from missing or uncertain ballots.  
Recent work on query-based committee selection also studies how to ask structured questions about voters'
preferences under a limited budget \citep{zimet2026query}. Query-efficient elicitation has also been studied for randomized (socially acceptable) choices in AI deployment: \citet{choo2026learnqueries} consider the problem of learning whether there exists a lottery over options that is acceptable to all stakeholders, using only binary accept/reject responses to proposed lotteries. Our model is different from these fixed-election and acceptability-query settings in two ways. First, the objective is distributional: the committee is evaluated on future tasks sampled from the same population, rather than on a fixed electorate.
Second, our algorithms exploit the special structure of ensemble coverage.  In
particular, the failure-conditioned greedy algorithm queries new candidates
mainly on tasks not already covered by the current committee.  This preserves
the classical approximation guarantee while reducing the number of model calls
needed to find a strong ensemble.

\medskip
\noindent
\textbf{Pairwise preferences and Condorcet-style committees.}
Our pairwise feedback model connects to the literature on learning and
aggregating preferences from comparisons.  Pairwise comparisons are common in modern AI evaluation, including Bradley-Terry style models and arena-based leaderboards \citep{amelistatistical,bradley1952rank,chiang2024chatbotarena}.  Several recent
works argue that a single global ranking can hide important disagreement across tasks or users, and propose more robust or pluralistic evaluation procedures \citep{haghtalab2026pluralistic,khalaf2026robust,lanctot2025soft}.  These works
focus mainly on ranking or evaluating individual systems.  Our goal is different:
we select a small set of systems whose best member performs well on each task.

The closest social choice notion for our pairwise setting is the
$\theta$-winning set of \citet{elkind2011choosing}, which asks whether a committee can beat every outside candidate with sufficiently high probability.
Recent work has sharpened the size of such sets needed in worst-case majority tournaments \citep{charikar2025six,song2026few}.  We use this idea in a distributional
setting where comparisons are sampled from tasks and must be queried.  Our
results show both the promise and the limitation of this objective.  Maximizing
$\theta$ admits a PTAS under full information, but the objective is not
submodular, so the direct analogue of greedy coverage fails.  To obtain
query-efficient algorithms, we introduce a weighted ordinal coverage surrogate
that is monotone and submodular, optimize it with failure-conditioned queries,
and then convert the resulting coverage guarantee into a $\theta$ guarantee by
auditing only a finite family of possible outside challengers.

\medskip
\noindent
\textbf{LLM ensembles and routing.}
There is a growing systems literature on combining large language models.
Representative examples include output-level ensembling and ranking
\citep{jiang2023llmblender}, cost-aware model cascading
\citep{chen2024frugalgpt}, hybrid use of stronger and weaker models
\citep{ding2024hybrid}, learned routing between models
\citep{ong2025routellm}, and benchmarks for router training and evaluation
\citep{hu2024routerbench}.  These methods show that different models have
different strengths and that routing can reduce cost while maintaining quality.
Our contribution is complementary.  We do not assume a trained router, calibrated
scores, or full evaluations of all models on all tasks.  Instead, we give
query complexity guarantees for the upstream problem of selecting a small
ensemble from binary or pairwise feedback.  The experiments show why this matters:
the committees chosen by our algorithms exploit complementarity across models
and can outperform simply taking the individually strongest models, while using
far fewer evaluations than exhaustive search.

\medskip
\noindent
An extended discussion on additional related work is deferred to Appendix~\ref{app:additional-related-work}.

\section{Preliminaries}
For $z \in \mathbb{N}$, denote $[z] := \{1, \dots, z\}$. For a finite set $A$, let $\Delta(A)$ denote the probability simplex over $A$. 
Let $C= [m]$ be a set of $m$ \emph{experts} (also called \emph{candidates}).
An \emph{ensemble} (or \emph{committee}) is a subset $S\subseteq C$. 
For
$k\in\{0,\dots,m\}$, let $\mathcal S_k := \{S\subseteq C : |S|=k\}$ denote the family of size-$k$ committees. Tasks are drawn i.i.d. from an
underlying task distribution $\mathcal D$; we write $d\sim\mathcal D$ for a
sampled task.

\medskip
\noindent
\textbf{Binary feedback.}
In the binary feedback model, a task $d$ induces one bit for each expert:
$u(d,c)\in\{0,1\}$. The value $u(d,c)=1$ means that expert $c$ solves
task $d$, and $u(d,c)=0$ means that it does not. The binary feedback oracle
takes a sampled task-expert pair $(d,c)$ and returns $u(d,c)$.

For a committee $S\subseteq C$, define its task-level binary utility by $U(d,S) := \mathbf{1}\{\exists c\in S : u(d,c)=1\}$.
Equivalently, $U(d,S)=\max_{c\in S}u(d,c)$ for $S\neq\varnothing$, with the
convention $U(d,\varnothing)=0$. The \emph{domain-level coverage} of $S$ is $v(S) := \Pr_{d\sim\mathcal D}[U(d,S)=1]$, and the optimal size-$k$ binary feedback value is $\OPT_k := \max_{S\in\mathcal S_k} v(S)$.

\medskip
\noindent
\textbf{Pairwise feedback.}
In the pairwise feedback model, a task $d$ induces a strict ranking $\pi_d$
of the experts in $C$. Equivalently, $\mathcal D$ induces a distribution
$P_{\mathcal D}$ over the set $\mathfrak S_m$ of all permutations of $C$.
We denote $\pi\sim P_{\mathcal D}$ the latent ranking induced by a sampled
task. 
For $c \in C$, let $\mathrm{rank}_{\pi}(c) \in \{1,\dots,m\}$ denote
the position of $c$ in $\pi$, where rank $1$ is best. We write
$a \succ_{\pi} b$ as shorthand for
$\mathrm{rank}_{\pi}(a) < \mathrm{rank}_{\pi}(b)$.
The learner does not observe $\pi$ directly. The pairwise feedback oracle,
given a sampled task with latent ranking $\pi$ and two distinct experts
$a,b\in C$, returns $\mathrm{Query}(a,b;\pi)
    := \mathbf{1}\{\mathrm{rank}_{\pi}(a) < \mathrm{rank}_{\pi}(b)\}$.
Thus the oracle reports whether $a$ is ranked above $b$ on that task.

Committees are compared through their best-ranked members. For
$S\subseteq C$, define $r_{\pi}(S) := \min_{c\in S}\mathrm{rank}_{\pi}(c)$ if $S \neq \varnothing$ and $r_{\pi}(S) := m+1$ if $S = \varnothing$.
Thus $r_{\pi}(S)$ is the rank of the best member of $S$, with the empty
committee placed below every expert. For nonempty $A\subseteq C$, let
$\top_{\pi}(A)$ denote the unique expert in $A$ attaining $r_{\pi}(A)$.

For two committees $S,S'\subseteq C$, define the pairwise win rate of $S$
against $S'$ by $\mathrm{WIN}(S,S')
    :=
    \mathbb E_{\pi\sim P_{\mathcal D}}
    [\mathbf{1}\{r_{\pi}(S)<r_{\pi}(S')\}
        + \frac12\mathbf{1}\{r_{\pi}(S)=r_{\pi}(S')\}]$.
This gives half credit to ties, which can arise when the two committees overlap.

The formal optimization problems use exact size committees $\mathcal S_k$.
This is the same benchmark as an at-most-$k$ budget for the committee
performance notions above: adding an expert cannot decrease $U(d,S)$ in the
binary model and cannot increase $r_{\pi}(S)$ in the pairwise model. Hence a
smaller committee can be padded to size $k$ without worsening its performance.
Whenever $C\setminus S$ appears in a definition, $S$ is assumed to be a
proper committee; for $S\in\mathcal S_k$, this requires $k<m$.

\medskip
\noindent
\textbf{Query access.}
A sampled task is observed only as a handle to its latent feedback profile. In the binary model, the latent profile is the vector $(u(d,c))_{c \in C}$; in the pairwise model, it is the ranking $\pi_d$. The selection algorithm may store sampled tasks and query them later, choosing each query adaptively from the feedback observed so far. For a fixed sampled task, the latent profile is fixed, so repeating the same query returns the same answer and is redundant. In the pairwise model, the answer to $(a,b)$ also determines the answer to $(b,a)$. Query complexity counts oracle calls. We use ``committee evaluation'' for queries used to test the current committee on a task and ``candidate-evaluation'' for queries used to evaluate possible additions; when both appear, their costs are reported separately.

\section{Binary Feedback} \label{sec:approval}
In the binary feedback setting, the goal is to find $S\in\mathcal S_k$
with maximum coverage $v(S)$. For binary utilities, the task-wise win rate
of $S$ against $S'$ is $\mathrm{WIN}(S,S')
=
\Pr[U(d,S)>U(d,S')]
+\frac12\Pr[U(d,S)=U(d,S')]$ where $d\sim D$. 
Since $U(d,S),U(d,S')\in\{0,1\}$, $\mathrm{WIN}(S,S')
=
\frac12+\frac12 (v(S)-v(S'))$.
Thus a coverage-maximizing committee does not lose, under this comparison,
to any other committee of the same size.

With full information, maximizing $v$ is approval Chamberlin-Courant on
sampled tasks~\citep{SkowronFaliszewski2017}. Each task acts as a voter and approves exactly the experts that solve it. A committee covers a task if at least one of its members is
approved. For a finite fully observed sample, the problem is weighted
MaxCover. Hence greedy achieves the optimal polynomial-time
$(1-1/e)$-approximation, and improving this factor is NP-hard~\citep{Feige1998,NemhauserWolseyFisher1978}. The contribution of this section is a query complexity analysis for learning a high-coverage committee from sampled tasks.

\subsection{Baseline: Exhaustive Elicitation}
\label{subsec:approval-baseline}

A direct benchmark draws $d_1,\dots,d_T\sim D$ and queries every expert
on every sampled task, using exactly $Q=mT$ oracle queries. For each
committee $S\subseteq C$, define $\widehat v_T(S) := \frac1T\sum_{t=1}^T U(d_t,S)$,
and let $\widehat S_{\mathrm{ERM}}
    \in \arg\max_{S\in\mathcal S_k} \widehat v_T(S)$.
This final optimization may require searching over all
$\binom{m}{k}$ committees.

\begin{theorem}
\label{thm:approval-erm}
    Fix $\varepsilon,\delta\in(0,1)$. If $T \ge \frac{2}{\varepsilon^2}
    (\log {m\choose k}+\log (2/\delta))$,
    then, with probability at least $1-\delta$, $v(\widehat S_{\mathrm{ERM}})\ge \mathrm{OPT}_k-\varepsilon$.
\end{theorem}

This worst-case statistical dependence is unavoidable already for
singletons.

\begin{theorem}
\label{thm:approval-lower}
    Assume $m\ge 2$, $0<\varepsilon\le 1/8$, and $k=1$. Any adaptive algorithm that returns $\widehat c\in C$ with $v(\{\widehat c\})\ge \mathrm{OPT}_1-\varepsilon$ with probability at least $2/3$ on every instance has worst-case expected query complexity $\Omega(m/\varepsilon^2)$.
\end{theorem}

For a polynomial-time benchmark, we still fully elicit the $T$ sampled
tasks, but replace ERM by the standard greedy algorithm applied to
$\widehat v_T$. Since coverage is monotone submodular, this gives the usual $(1-1/e)$ guarantee.

\begin{theorem}
\label{thm:approval-full-greedy}
    Fix $\varepsilon,\delta\in(0,1)$. Let $\widehat S_{\mathrm{gr}}\in
    \mathcal S_k$ be the committee returned by empirical greedy on
    $\widehat v_T$, padded arbitrarily to size $k$ if greedy stops early. If $T \ge
    \frac{(2-1/e)^2}{2\varepsilon^2} (\log {m\choose k}+\log ({2}/{\delta})$, then, with probability at least $1-\delta$, $v(\widehat S_{\mathrm{gr}})
    \ge
    (1-1/e)\mathrm{OPT}_k-\varepsilon$.
    The procedure uses $Q=mT$ queries.
\end{theorem}

\subsection{Adaptive Query Savings from Missed Instances} \label{subsec:adaptive_query} 
We next ask whether the greedy benchmark can be made query-adaptive. Instead
of eliciting all $mT$ expert-task outcomes, we would like to query only the outcomes needed to identify large marginal gains.
For $S\subseteq C$, let $\rho(S):=\Pr_{d\sim D}[U(d,S)=0]=1-v(S)$ be its miss rate. For $c\in C\setminus S$, define $q(c\mid S)
:=
\Pr_{d\sim D}[u(d,c)=1\mid U(d,S)=0]$ when $\rho(S)>0$, and set $q(c\mid S)=0$ when $\rho(S)=0$. Let $\Delta(c\mid S):=v(S\cup\{c\})-v(S)$.
The key identity is
\begin{equation} \label{eq:approval-failure-marginal}
    \Delta(c\mid S)=\rho(S)q(c\mid S).
\end{equation}
Thus, at a fixed greedy step, maximizing marginal gain is the same as
maximizing rescue rate on tasks missed by the current committee.

\medskip
\noindent
\textbf{Failure-conditioned elimination.}
We use a fixed confidence elimination routine, \textsc{FailCond-Elim}, whose
full pseudocode is deferred to Appendix~\ref{app:failcond-elim}
(Algorithm~\ref{alg:failcond-elim}). The routine repeatedly samples failure
instances of the current committee, queries only the currently active
candidates, and eliminates any candidate whose empirical rescue rate is
certifiably worse than the current empirical best by more than the target
accuracy. Its cost is therefore gap-dependent: weak candidates are discarded
quickly when their rescue rate gaps are large.

For the next result, fix $S$ with $\rho(S)>0$ and a nonempty candidate set $A\subseteq C\setminus S$.
% Algorithm~\ref{alg:failcond-elim} exploits this identity. It is a
% fixed-confidence elimination routine run on failure instances of the
% current committee. Its cost depends on the rescue rate gaps $\Delta_q(c):=\max_{a\in A}q(a\mid S)-q(c\mid S)$, so weak candidates can be discarded quickly when these gaps are large.

\begin{theorem} \label{thm:failcond-elim}
    Fix $\eta\in(0,1]$ and $\delta\in(0,1)$. The routine uses $R=O({\log(e|A|/(\delta\eta))}/{\eta^2})$ accepted failure instances deterministically. With probability at least
    $1-\delta$, $q(\widehat c\mid S)\ge \max_{a\in A}q(a\mid S)-\eta$ and $Q_{\mathrm{cand}}
    =
    O(
    \sum_{c\in A}
    {\log(e|A|/(\delta\eta))}
         /{\max\{\Delta_q(c),\eta\}^2}
    )$.
\end{theorem}

Accepted failures can be generated by rejection sampling: draw $d\sim D$
and evaluate $S$ until either one queried member succeeds or all queried
members fail. Each accepted failure requires $1/\rho(S)$ unconditional
draws in expectation. We count these committee evaluation queries separately
below.

We now wrap the elimination routine in greedy selection. Start from
$S_0=\varnothing$. At step $i$, let $\rho_i:=1-v(S_i)$. If
$\rho_i\le \varepsilon/k$, fill the remaining slots arbitrarily and stop.
Otherwise, call $\textsc{FailCond-Elim}
(S_i,C\setminus S_i,\frac{\varepsilon}{k\rho_i},\frac{\delta}{k})$ and add the returned candidate. Appendix~\ref{sec:AFG} gives the pseudocode, including
an implementable version that replaces the unknown $\rho_i$ by confidence
bounds.

Let $\tau$ be the number of executed elimination calls. For each executed
step $i$ and each $c\in C\setminus S_i$, define the marginal gap $\Delta_i(c)
:=
\max_{c'\in C\setminus S_i}\Delta(c'\mid S_i)-\Delta(c\mid S_i)$.

\begin{theorem}
\label{thm:adaptive-fail-greedy}
    Fix $k\in\{1,\dots,m\}$ and $\varepsilon,\delta\in(0,1)$. Let
    $\widehat S$ be the output of the greedy selection above. With probability at least $1-\delta$, $v(\widehat S)\ge (1-1/e)\mathrm{OPT}_k-\varepsilon$.
    On the same event, $Q_{\mathrm{cand}}
    =
    O (
    \sum_{i=0}^{\tau-1}\sum_{c\in C\setminus S_i}
    {\rho_i^2\log(emk^2/(\delta\varepsilon))}/
         {\max\{\Delta_i(c),\varepsilon/k\}^2}
    )$.
\end{theorem}

Theorem~\ref{thm:adaptive-fail-greedy} counts only candidate evaluations after failure instances have
been found. The extra cost of finding those failures depends on how the
current committee is tested. If an unconditional draw $d\sim D$ is tested
against $S_i$ in order $\sigma_i$, let
$Q_{\mathrm{eval}}(d,S_i,\sigma_i)\le |S_i|=i$ be the number of
committee-evaluation queries used on that draw. The expected additional
committee-evaluation cost at step $i$ is $O (
\rho_i\,
\mathbb E[Q_{\mathrm{eval}}(d,S_i,\sigma_i)]
{k^2\log(mk^2/(\delta\varepsilon))}/{\varepsilon^2}
)$.
We keep this term separate because it depends on the testing order
$\sigma_i$.
Theorem~\ref{thm:adaptive-fail-greedy} is an instance-dependent refinement of the exhaustive greedy baseline in Theorem~\ref{thm:approval-full-greedy}, not a worst-case improvement over the lower bound in Theorem~\ref{thm:approval-lower}.
 The worst case still requires
$\Omega(m/\varepsilon^2)$ queries, but the adaptive bound is smaller when
miss rates shrink quickly or when marginal gaps are large.

The local gap dependence is also unavoidable. Once we condition on failures
of a fixed committee, the next greedy step is an $\eta$-best-arm
identification problem over rescue rates.
For this lower bound, fix $S\subseteq C$ with $\rho(S)>0$ and a nonempty
candidate set $A\subseteq C\setminus S$. For any conditional distribution
of $d$ given $U(d,S)=0$, let $q^*:=\max_{a\in A}q(a\mid S)$ and $\Delta_q(a):=q^*-q(a\mid S)$.

\begin{theorem} \label{thm:onestep_lb}
    Fix $\eta\in(0,1)$ and $\delta\in(0,1/4]$. Any algorithm that, for every
    such conditional distribution, returns $\widehat c\in A$ satisfying $q(\widehat c\mid S)\ge q^*-\eta$
    with probability at least $1-\delta$ has a hard instance on which $\mathbb E[Q_{\mathrm{cand}}]
    =
    \Omega (
    \sum_{a\in A:\Delta_q(a)>\eta}
    {\log(1/\delta)}/{\Delta_q(a)^2} )$.
\end{theorem}

By \eqref{eq:approval-failure-marginal}, a rescue rate gap
$\Delta_q(c)$ corresponds to a marginal gap $\rho(S)\Delta_q(c)$. 
Thus, 
Theorems~\ref{thm:failcond-elim} and~\ref{thm:onestep_lb} match up
to logarithmic factors in the local gap dependence, while candidates
already within $\eta$ of optimal need not be separated.

\section{Pairwise Feedback}
\label{sec:pairwise}
We now turn to pairwise feedback. Each task induces a strict ranking
$\pi$ of the experts, but the learner observes this ranking only through
pairwise comparisons. For a committee $S$, recall that $r_\pi(S)$ is the
rank of its best member.

A natural extension of the binary guarantee would ask for a deterministic
committee that beats every other size-$k$ committee by majority vote. This
can fail badly because of majority cycles. Fix $N>2k$, let $C=\mathbb Z_N$,
and let $P_D$ be uniform over the cyclic rankings $\pi_i:\quad i\succ i+1\succ\cdots\succ i-1$ for $i\in\mathbb Z_N$.
For any $S\in\mathcal S_k$, set $T=S-1:=\{s-1 \pmod N : s\in S\}$.
If $i\notin S$ and $s$ is the first member of $S$ in the cyclic order
from $i$, then $s-1\in T$ appears before $s$. Hence
$r_{\pi_i}(T)<r_{\pi_i}(S)$. Thus $T$ beats $S$ on at least $N-k$
rankings, so $\mathrm{WIN}(S,T)\le k/N$. Taking $N$ large makes this
arbitrarily small. We therefore compare a committee to each outside
candidate rather than to every opposing committee.

\begin{definition}[$\theta$-winning committees]
\label{def:theta-winning}
    For a proper committee $S\subsetneq C$, define $\theta(S):=\min_{x\in C\setminus S}\mathrm{WIN}(S,\{x\})
    = \min_{x\in C\setminus S}
    \Pr_{\pi\sim P_D} [r_\pi(S)<\mathrm{rank}_\pi(x)]$.
    For $\vartheta\in[0,1]$, we say that $S$ is $\vartheta$-winning if
    $\theta(S)\ge \vartheta$. For $1\le k<m$, let $\theta_k^*:=\max_{S\in\mathcal S_k}\theta(S)$.
\end{definition}

This candidate-wise guarantee also gives a comparison guarantee against any
committee $S'$. If $S$ is $\vartheta$-winning, then $S$ weakly beats
$S'$ with probability at least $1-|S'|(1-\vartheta)$. This is useful when
$|S'|(1-\vartheta)$ is small, for example when $S'$ is small or when
$S$ has a high candidate-wise guarantee. It does not contradict the cyclic
example above: in that construction the best candidate-wise guarantee is only about $1-1/k$, so the induced guarantee against another size-$k$ committee is vacuous.

\begin{lemma}
\label{lem:theta-lifts-to-committees}
    If $S\subsetneq C$ is nonempty and $\vartheta$-winning, then for every
    committee $S'\subseteq C$, $\Pr_{\pi\sim P_D}\left[r_\pi(S)\le r_\pi(S')\right]
    \ge
    1-|S'|(1-\vartheta)$.
    If $S\cap S'=\varnothing$, then $\mathrm{WIN}(S,S')\ge 1-|S'|(1-\vartheta)$.
\end{lemma}

\subsection{Full-information optimization of $\theta$}
\label{sec:theta-full-info}
We first consider the full-information problem for $\theta$. The main
message is simple: the problem has a Polynomial-Time Approximation Scheme (PTAS), but no Efficient PTAS (EPTAS) under Gap-ETH.

\begin{theorem}[Informal]
\label{thm:theta-full-info}
For every $\gamma\in(0,1)$, full-information maximization of $\theta$
admits a $(1-\gamma)$-approximation in time $m^{O(1/\gamma)}\mathrm{poly}(m,n)$.
Assuming Gap-ETH, the same guarantee cannot be achieved in time $f(1/\gamma)\mathrm{poly}(m,n)$ for any computable function $f$.
\end{theorem}

The PTAS is mainly a benchmark for the fully observed problem. Its polynomial exponent grows with $1/\gamma$, and the lower bound rules out an EPTAS under Gap-ETH. For query-limited selection, we therefore focus on efficient constant-factor guarantees.
% The PTAS is obtained via an existential result based on recent breakthroughs \citep{charikar2025six, song2026few}, and is unlikely to be practical as even $\gamma=1/2$ gives an $m^5$ dependence.  The EPTAS lower bound rules out removing the dependence on $1/\gamma$ from the polynomial exponent.  Thus, for practical purposes, the natural target is an efficient fixed-ratio approximation rather than an arbitrarily accurate scheme.

\subsection{Baseline: Exhaustive elicitation}\label{subsec: theta exhaustive elicitation}
We next give the ordinal analogue of exhaustive elicitation. As in the
binary setting, the algorithm draws $T$ independent tasks and fully elicits
the feedback induced by each task. Here, full elicitation means recovering a
ranking of the $m$ candidates for each task, using $O(m\log m)$ pairwise
comparisons rather than $m$ binary queries.

Fix $1\le k<m$ and $\varepsilon,\delta\in(0,1)$. Draw independent
rankings $\pi_1,\dots,\pi_T\sim P_D$, recover each ranking by pairwise
comparisons, and define $\widehat\theta_T(S)
:=
\min_{x\in C\setminus S}
\frac1T\sum_{t=1}^T
\mathbf 1\left\{r_{\pi_t}(S)<\mathrm{rank}_{\pi_t}(x)\right\}$.
Let $\widehat S_{\mathrm{ERM}}\in
\arg\max_{S\in\mathcal S_k}\widehat\theta_T(S)$.

\begin{theorem}
\label{thm:ordinal-erm}
    If $T\ge
    \frac{2}{\varepsilon^2}
    \log ({2(m-k){m\choose k}}/{\delta} )$,
    then, with probability at least $1-\delta$, $\theta(\widehat S_{\mathrm{ERM}})
    \ge \theta_k^*-\varepsilon$.
    The procedure uses $O(Tm\log m)$ pairwise comparisons.
\end{theorem}

The next result shows that the linear dependence on $m$ cannot be avoided,
even for singleton committees.

\begin{theorem}[Worst-case pairwise query lower bounds]
\label{thm:ordinal-lower-main}
Assume $m\ge 2$ and $0<\varepsilon\le 1/16$. Then,
\begin{enumerate}[(i)]
    \item For a fixed singleton committee $S=\{s\}$, any adaptive algorithm that estimates $\theta(S)$ to additive error $\varepsilon$ with probability at least $2/3$ on every ranking distribution has worst-case expected pairwise query complexity $\Omega(({m-1})/{\varepsilon^2})$.
    
    \item Any adaptive algorithm that returns $\widehat c\in C$ with $\theta(\{\widehat c\})\ge \theta_1^*-\varepsilon$ with probability at least $2/3$ on every ranking distribution has worst-case expected pairwise query complexity $\Omega({m}/{\varepsilon^2})$.
    \end{enumerate}
\end{theorem}

\subsection{A submodular relaxation of \texorpdfstring{$\theta$}{theta}}
\label{sec:theta-vs-phi}

The results above motivate a query-efficient objective that still connects
to $\theta$. Greedy is the natural candidate, but it cannot be applied
directly to $\theta$: unlike binary coverage, $\theta$ is monotone but
not submodular.

\begin{proposition}
\label{prop:theta-not-submodular}
    On proper committees, $\theta$ is monotone but not submodular.
\end{proposition}

We therefore use a submodular relaxation. For $x\in C$ and ranking $\pi$,
let $P_\pi(x):=\{c\in C:\mathrm{rank}_\pi(c)\le
\mathrm{rank}_\pi(x)\}$ be the prefix up to $x$. A committee $S$ covers $(\pi,x)$ if $S\cap P_\pi(x)\neq\varnothing$, meaning that some member of $S$ is ranked at least as highly as $x$. Define $g_x(S):=\Pr_{\pi\sim P_D}[S\cap P_\pi(x)\neq\varnothing]$.
If $x\notin S$, then $g_x(S)=\mathrm{WIN}(S,\{x\})$, while if $x\in S$,
then $g_x(S)=1$. Hence, for every proper committee $S$, $\theta(S)=\min_{x\in C} g_x(S)$.

For $\lambda\in\Delta(C)$, define $\Phi_\lambda(S) := \sum_{x\in C}\lambda_x g_x(S) = \Pr_{\pi\sim P_D,\,x\sim\lambda}[S\cap P_\pi(x)\neq\varnothing]$.

\begin{lemma} \label{lem:theta-phi}
    For every $\lambda\in\Delta(C)$, the function $\Phi_\lambda$ is
    normalized, monotone, and submodular. Moreover, for every proper committee
    $S$, $\theta(S)=\min_{\lambda\in\Delta(C)}\Phi_\lambda(S)$.
\end{lemma}

\textbf{Important distinction.}
Lemma~\ref{lem:theta-phi} does not make $\theta$ submodular. It writes $\theta$ as the
pointwise minimum of the submodular functions $\Phi_\lambda$, and such a
minimum need not be submodular. Thus the greedy algorithm below is a
query-efficient oracle for the fixed-$\lambda$ problem $\max_{S\in\mathcal S_k}\Phi_\lambda(S)$, not a greedy algorithm for maximizing $\theta$. Section~5.5 explains how to return from $\Phi_\lambda$ to $\theta$, either by auditing a finite
family of committees or by using a minimax wrapper.

\subsection{A pairwise query oracle for fixed $\lambda$}
\label{sec:weighted-ordinal-oracle}
We now give the query primitive used later. Fix $\lambda\in\Delta(C)$. For
a committee $S$, define $\rho_\lambda(S)
:=
1-\Phi_\lambda(S)
=
\Pr_{\pi\sim P_D,\,x\sim\lambda}[S\cap P_\pi(x)=\emptyset]$.
If $\rho_\lambda(S)>0$, define the rival-conditioned rescue rate $q_\lambda(c\mid S)
:=
\Pr_{\pi\sim P_D,\,x\sim\lambda}
[c\in P_\pi(x)\mid S\cap P_\pi(x)=\emptyset]$ for $c\notin S$. Set $q_\lambda(c\mid S)=0$ when
$\rho_\lambda(S)=0$.

\begin{lemma}%[Rival-conditioned marginal identity]
\label{lem:ordinal-marginal-identity}
    For every $S\subseteq C$, $c\notin S$, and $\lambda\in\Delta(C)$, $\Phi_\lambda(S\cup\{c\})-\Phi_\lambda(S) = \rho_\lambda(S)q_\lambda(c\mid S)$
\end{lemma}

The event $S \cap P_{\pi}(x)=\varnothing$ is the pairwise feedback analogue of a
missed task in the binary feedback setting: no current committee member is ranked
at least as highly as the sampled rival $x$. Conditional on such a failed instance-rival pair $(\pi,x)$, testing whether a new candidate $c$ rescues the pair requires only the comparison between $c$ and $x$, with the convention that $c=x$ rescues automatically.

\begin{theorem}%[Failure-conditioned greedy for fixed $\lambda$]
\label{thm:weighted-ordinal-oracle}
    Fix $k\in\{1,\dots,m\}$, $\lambda\in\Delta(C)$, and
    $\varepsilon,\delta\in(0,1)$. The adaptive pairwise query algorithm in
    Appendix~\ref{app:weighted-ordinal-oracle} outputs $\widehat S\in\mathcal S_k$ such that, with
    probability at least $1-\delta$, $\Phi_\lambda(\widehat S)
    \ge
    (1-1/e)\Phi_{\lambda,k}^*-\varepsilon$ and $\Phi_{\lambda,k}^*:=\max_{S\in\mathcal S_k}\Phi_\lambda(S)$.
    Its accepted failed-pair comparisons satisfy the gap-dependent bound in
    Appendix~\ref{app:weighted-ordinal-oracle}; rejection-sampling and committee testing costs are counted
    separately there.
\end{theorem}

This is a fixed-weight oracle, not yet a $\theta$ guarantee: a committee
may perform well for one $\lambda$ while still failing a rival that
$\lambda$ gives little weight.

% Theorem~\ref{thm:weighted-ordinal-oracle} should be read as a weighted-subproblem oracle. It is the right primitive for a minimax method over rivals, but by itself it does not certify that $\widehat S$ has large $\theta(\widehat S)$. A committee can have large $\Phi_\lambda$ for one choice of $\lambda$ while performing poorly against a rival that receives little or no mass under $\lambda$.

\subsection{Returning to $\theta$}
\label{sec:return-to-theta}

We use two routes to convert pairwise feedback into guarantees for
$\theta$. The first is direct auditing, either for one proposed committee
or over an explicit finite family of committees.

For a finite family $\mathcal F$ of nonempty proper committees, let $N_{\mathcal F}:=\sum_{S\in\mathcal F}(m-|S|)$.

\begin{theorem}%[Auditing and finite-family $\theta$-learning]
\label{thm:theta-audit-pool-main}
Fix $\varepsilon,\delta\in(0,1)$. A fixed nonempty proper committee
$S\subsetneq C$ can be audited to additive error $\varepsilon$ with
probability at least $1-\delta$ using $O(
\frac{m}{\varepsilon^2}
\log ({2(m-|S|)}/{\delta})
)$ pairwise comparisons.

More generally, empirical maximization over any nonempty finite family
$\mathcal F$ of nonempty proper committees returns
$\widehat S\in\mathcal F$ such that, with probability at least $1-\delta$, $\theta(\widehat S)\ge \max_{S\in\mathcal F}\theta(S)-\varepsilon$, using $O ( \frac{m\log m}{\varepsilon^2}
\log ({2N_{\mathcal F}}/{\delta}) )$ pairwise comparisons.
\end{theorem}

This gives deterministic $\theta$-learning for explicit pools, such as
committees produced by a heuristic. Taking $\mathcal F=\mathcal S_k$
recovers the exhaustive ordinal ERM baseline of Theorem~\ref{thm:ordinal-erm}, but this family
is exponential. Gap-adaptive refinements are in Appendix~\ref{app:theta-audit}.

Alternatively, we can use the fixed-$\lambda$ oracle inside multiplicative
weights over rivals.

\begin{theorem}%[Minimax wrapper: $\theta$-guarantee with randomization or bicriteria size]
\label{thm:minimax-theta-wrapper}
    Fix $1\le k<m$ and $\varepsilon,\delta\in(0,1)$. With
    $R=O(\varepsilon^{-2}\log m)$ calls to the fixed-$\lambda$ oracle, plus $O(
    R\cdot \frac{m}{\varepsilon^2}\log ({mR}/{\delta}))$ additional pairwise comparisons for rival audits, the minimax wrapper outputs
    committees $S_1,\dots,S_R\in\mathcal S_k$. If $p$ is the uniform
    distribution over these committees, then with probability at least
    $1-\delta$, for every $x\in C$, $\mathbb E_{S\sim p}[g_x(S)]
    \ge
    (1-1/e)\theta_k^*-\varepsilon$.
    Consequently, the union committee $S^+:=\bigcup_{t=1}^R S_t$ satisfies $|S^+|\le kR$ and, with the convention $\theta(C)=1$, $\theta(S^+)\ge (1-1/e)\theta_k^*-\varepsilon$.
\end{theorem}

Thus, the weighted relaxation gives a global $\theta$ guarantee either for
a lottery over size-$k$ committees or for a larger union committee. It does not give a deterministic size-$k$ approximation for $\theta$.

\section{Experiments}
\label{sec:experiments}
We include a small empirical study to illustrate the main algorithmic phenomena predicted by the theory. The experiments test whether adaptive committee selection can exploit complementarity among candidate models while using substantially fewer queried evaluations than naive search.

We consider two settings. The first uses binary correctness feedback from multilingual extractive QA: a query reveals whether a candidate model answers a training example correctly, and a committee is evaluated by held-out coverage. The second uses score-derived ordinal feedback from LiveBench:
task-level scores are converted into weak rankings, and a query reveals the induced comparison between two candidates on a training task. 

Figure~\ref{fig:experiment-livebench} shows two consistent patterns. First, adaptive methods reach the held-out oracle reference or substantially improve over budgeted sampled-ERM baselines at comparable query budgets. Second, Top-$k$ individual baselines can perform poorly because they select strong but redundant candidates, whereas the adaptive methods select committees with complementary strengths. This effect is most visible after masking the strongest individual candidates: in the binary feedback setting, Top-$k$ improves after removing the top few individual models from the candidate pool, suggesting that the strongest standalone models may solve largely overlapping sets of tasks and therefore provide limited complementarity. Full details on baselines, experimental setup and additional results can be found in Appendix~\ref{app:experiments}.

% We test the algorithms from Sections~\ref{sec:approval} and~\ref{sec:pairwise} on real LLMs and benchmarks. Figure~\ref{fig:experiment-livebench} highlights two main findings. First, adaptive querying can achieve comparable accuracy to exhaustive search with substantially fewer queries. Second, the selected ensembles exploit complementary strengths that are missed by simple ``pick the best individual models'' heuristics. This is reinforced in the binary-feedback setting, where the top-$k$ baseline improves after removing the top few individual models from the candidate pool, suggesting that the strongest standalone models may solve largely overlapping sets of tasks and therefore provide limited complementarity.

%Details on baselines, experiment setup and additional results can be found in Appendix~\ref{app:experiments}.

% Suggested arrangement for the 6 paper figures from exp 82.
% Layout (2 rows x 3 cols):
%   columns: [no mask | mask | x-axis = queries]
%   rows:    [binary (top) | pairwise/LiveBench (bottom)]
%
% Pairwise side uses LiveBench filtered-ordinal (m=37, n=215, 50/50)
% instead of MTEB. Mask level is top-5 (~14% of m=37) to mirror the
% effective leader-pool fraction of binary's mask-3 (~3% of m=89) and
% MTEB's mask-3 (~4% of m=70).
%
% Assumes you have \usepackage{graphicx,subcaption} in the preamble and
% have copied (or symlinked) the PDFs from
% experiments/82-paper-figures-livebench/plots/paper/ into your img/
% directory.

\begin{figure}[!htb]
  \centering
  % --- top row: binary feedback ---
  \includegraphics[height=\baselineskip]{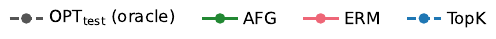}\\
  \begin{subfigure}[t]{0.32\linewidth}
    \includegraphics[width=\linewidth]{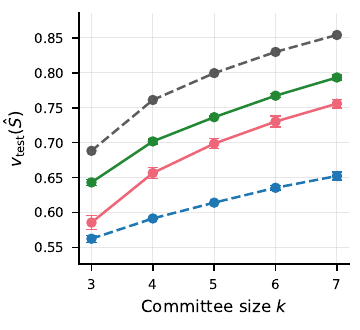}\vspace{-0.1in}
    \caption{Binary, $Q=10^{6}$.}
    \label{fig:lbb-binary-k-sweep}
  \end{subfigure}\hfill
  \begin{subfigure}[t]{0.32\linewidth}
    \includegraphics[width=\linewidth]{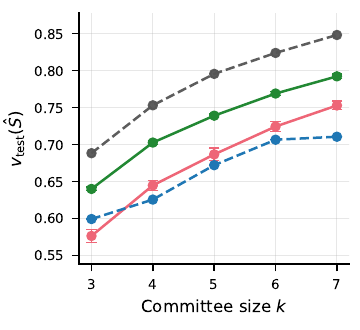}\vspace{-0.1in}
    \caption{Binary, top-5 masked.}
    \label{fig:lbb-binary-k-sweep-masked}
  \end{subfigure}\hfill
  \begin{subfigure}[t]{0.32\linewidth}
    \includegraphics[width=\linewidth]{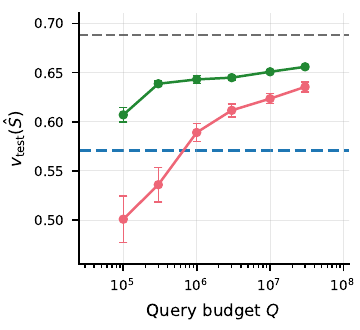}\vspace{-0.1in}
    \caption{Binary, $k=3$: $v_{\mathrm{test}}$ vs $Q$.}
    \label{fig:lbb-binary-q-frontier}
  \end{subfigure}

  % --- bottom row: LiveBench pairwise ---
  \includegraphics[height=\baselineskip]{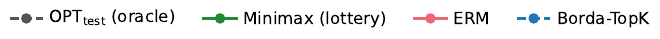}\\
  \begin{subfigure}[t]{0.32\linewidth}
    \includegraphics[width=\linewidth]{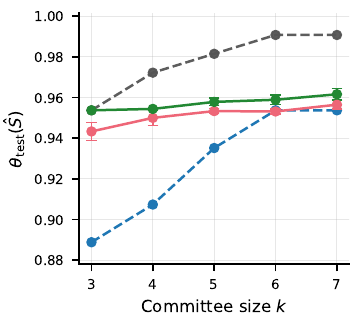}\vspace{-0.1in}
    \caption{Pairwise, $Q=10^{6}$.}
    \label{fig:lbb-livebench-k-sweep}
  \end{subfigure}\hfill
  \begin{subfigure}[t]{0.32\linewidth}
    \includegraphics[width=\linewidth]{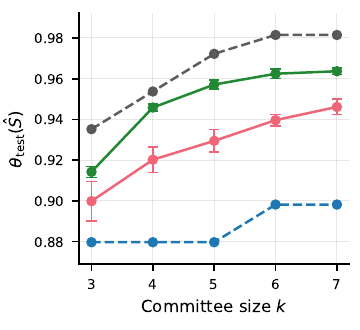}\vspace{-0.1in}
    \caption{Pairwise, top-5 masked.}
    \label{fig:lbb-livebench-k-sweep-masked}
  \end{subfigure}\hfill
  \begin{subfigure}[t]{0.32\linewidth}
    \includegraphics[width=\linewidth]{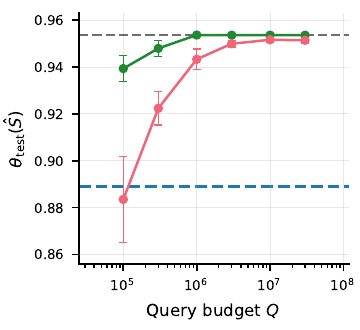}\vspace{-0.1in}
    \caption{\vspace{-0.05in} Pairwise, $k=3$: $\theta_\text{test}$ vs $Q$.}
    \label{fig:lbb-livebench-q-frontier}
  \end{subfigure}

\caption{
Top row: multilingual QA with binary correctness feedback, evaluated by held-out coverage. Bottom row: LiveBench score-derived ordinal feedback, evaluated by the held-out ordinal objective. Left: performance as committee size varies at $Q\approx 10^6$ queries. Middle: the same comparison after removing the five strongest individual candidates from the training pool. Right: performance as a function of query budget for $k=3$.
Top-$k$/Borda-Top-$k$ are full-information singleton baselines, and $\mathrm{OPT}_{\rm test}$ is a held-out oracle reference. Error bars show 95\% confidence intervals over algorithmic random seeds for a fixed train/test split.
}
  \label{fig:experiment-livebench}
\end{figure}

\section{Conclusion}
We introduced a distributional multiwinner voting framework for selecting small LLM ensembles from sampled task feedback. It highlights a deployment-relevant distinction often hidden by leaderboards: the best ensemble need not contain the strongest individual models, but those that best cover one another's failures. In the binary setting, this gives us a sampled-task coverage objective matching approval Chamberlin-Courant, with tight greedy-style guarantees. In the pairwise setting, majority cycles make committee-vs-committee dominance too strong, motivating $\theta$-winning committees and a weighted ordinal coverage relaxation that restores submodularity while preserving a link to preference-based ensemble quality.
The main message is that query efficiency comes from conditioning on failure. Rather than evaluating every candidate on every sampled task, our algorithms focus on residual cases where the current committee remains vulnerable. This preserves approximation guarantees while providing instance-dependent query savings. Small-scale LLM experiments illustrate this prediction: failure-conditioned methods match or improve non-adaptive baselines under comparable budgets, and selected ensembles exploit complementary strengths that Top-$k$ individual baselines often miss.

These results suggest treating LLM ensemble selection as a problem of coverage, complementarity, and evaluation cost, rather than truncating an individual-model leaderboard. They also suggest future directions. One is participatory budgeting, where candidate systems have heterogeneous costs, latencies, or validation burdens. Another is to move beyond subjective tasks and easily verifiable ground-truth tasks. For example, a scientist may ask several models to propose plausible hypotheses, but be misled if one model is right while others are plausibly wrong. This highlights that ensemble value is not always captured by its best member: it can also depend on how candidate answers are aggregated and acted upon. Extending our framework in this direction, together with recent work on LLM ensemble aggregation~\citep{ai2025majorityvotingllmaggregation}, would bring the theory closer to practice.

\bibliographystyle{plainnat}
\bibliography{bib,experiments}

\begin{thebibliography}{62}
\providecommand{\natexlab}[1]{#1}
\providecommand{\url}[1]{\texttt{#1}}
\expandafter\ifx\csname urlstyle\endcsname\relax
  \providecommand{\doi}[1]{doi: #1}\else
  \providecommand{\doi}{doi: \begingroup \urlstyle{rm}\Url}\fi

\bibitem[Ai et~al.(2025)Ai, Pan, Simchi-Levi, Tambe, and Xu]{ai2025majorityvotingllmaggregation}
Rui Ai, Yuqi Pan, David Simchi-Levi, Milind Tambe, and Haifeng Xu.
\newblock Beyond majority voting: Llm aggregation by leveraging higher-order information, 2025.
\newblock URL \url{https://arxiv.org/abs/2510.01499}.

\bibitem[{Amazon Web Services}(2025)]{aws2025bedrockevaluation}
{Amazon Web Services}.
\newblock Beyond the basics: A comprehensive foundation model selection framework for generative {AI}.
\newblock \url{https://aws.amazon.com/blogs/machine-learning/beyond-the-basics-a-comprehensive-foundation-model-selection-framework-for-generative-ai/}, 2025.

\bibitem[Ameli et~al.(2025)Ameli, Zhuang, Stoica, and Mahoney]{amelistatistical}
Siavash Ameli, Siyuan Zhuang, Ion Stoica, and Michael~W Mahoney.
\newblock A statistical framework for ranking {LLM}-based chatbots.
\newblock In \emph{The 13th International Conference on Learning Representations (ICLR)}, 2025.

\bibitem[{Andreessen Horowitz}(2025)]{a16z2025enterprise}
{Andreessen Horowitz}.
\newblock How 100 enterprise cios are building and buying gen ai in 2025.
\newblock \url{https://a16z.com/ai-enterprise-2025/}, 2025.

\bibitem[Artetxe et~al.(2020)Artetxe, Ruder, and Yogatama]{artetxe2020xquad}
Mikel Artetxe, Sebastian Ruder, and Dani Yogatama.
\newblock On the cross-lingual transferability of monolingual representations.
\newblock In \emph{Proceedings of the 58th Annual Meeting of the Association for Computational Linguistics}, pages 4623--4637, 2020.

\bibitem[Audibert and Bubeck(2010)]{audibert2010bestarm}
Jean-Yves Audibert and S{\'e}bastien Bubeck.
\newblock Best arm identification in multi-armed bandits.
\newblock In \emph{Proceedings of the 23rd Conference on Learning Theory (COLT)}, 2010.

\bibitem[{BhabhaAI}(2024)]{gajendra2024}
{BhabhaAI}.
\newblock {Gajendra-v0.1}: A hindi-hinglish-english instruct {LLM}.
\newblock \url{https://huggingface.co/BhabhaAI/Gajendra-v0.1}, 2024.

\bibitem[Bradley and Terry(1952)]{bradley1952rank}
Ralph~Allan Bradley and Milton~E Terry.
\newblock Rank analysis of incomplete block designs: I. the method of paired comparisons.
\newblock \emph{Biometrika}, 39\penalty0 (3/4):\penalty0 324--345, 1952.

\bibitem[Brandt et~al.(2016)Brandt, Conitzer, Endriss, Lang, and Procaccia]{brandt2016handbook}
Felix Brandt, Vincent Conitzer, Ulle Endriss, J{\'e}r{\^o}me Lang, and Ariel~D Procaccia.
\newblock \emph{Handbook of computational social choice}.
\newblock Cambridge University Press, 2016.

\bibitem[Caruana et~al.(2004)Caruana, Niculescu-Mizil, Crew, and Ksikes]{caruana2004ensemble}
Rich Caruana, Alexandru Niculescu-Mizil, Geoff Crew, and Alex Ksikes.
\newblock Ensemble selection from libraries of models.
\newblock In \emph{Proceedings of the 21st International Conference on Machine Learning (ICML)}, page~18, 2004.

\bibitem[Chamberlin and Courant(1983)]{chamberlin1983representative}
John~R. Chamberlin and Paul~N. Courant.
\newblock Representative deliberations and representative decisions: Proportional representation and the {Borda} rule.
\newblock \emph{American Political Science Review}, 77\penalty0 (3):\penalty0 718--733, 1983.

\bibitem[Charikar et~al.(2025)Charikar, Lassota, Ramakrishnan, Vetta, and Wang]{charikar2025six}
Moses Charikar, Alexandra Lassota, Prasanna Ramakrishnan, Adrian Vetta, and Kangning Wang.
\newblock Six candidates suffice to win a voter majority.
\newblock In \emph{Proceedings of the 57th Annual ACM Symposium on Theory of Computing (STOC)}, pages 1590--1601, 2025.

\bibitem[Chen et~al.(2024)Chen, Zaharia, and Zou]{chen2024frugalgpt}
Lingjiao Chen, Matei Zaharia, and James Zou.
\newblock {FrugalGPT}: How to use large language models while reducing cost and improving performance.
\newblock \emph{Transactions on Machine Learning Research}, 2024.

\bibitem[Chen et~al.(2014)Chen, Lin, King, Lyu, and Chen]{chen2014combinatorial}
Shouyuan Chen, Tian Lin, Irwin King, Michael~R. Lyu, and Wei Chen.
\newblock Combinatorial pure exploration of multi-armed bandits.
\newblock In \emph{Proceedings of the 28th International Conference on Neural Information Processing Systems (NeurIPS)}, pages 379--387, 2014.

\bibitem[Chiang et~al.(2024)Chiang, Zheng, Sheng, Angelopoulos, Li, Li, Zhang, Zhu, Jordan, Gonzalez, and Stoica]{chiang2024chatbotarena}
Wei-Lin Chiang, Lianmin Zheng, Ying Sheng, Anastasios~Nikolas Angelopoulos, Tianle Li, Dacheng Li, Hao Zhang, Banghua Zhu, Michael Jordan, Joseph~E. Gonzalez, and Ion Stoica.
\newblock Chatbot arena: An open platform for evaluating {LLM}s by human preference.
\newblock In \emph{Proceedings of the 41st International Conference on Machine Learning (ICML)}, pages 8359--8388, 2024.

\bibitem[Choo et~al.(2026)Choo, Goldberg, and Teh]{choo2026learnqueries}
Davin Choo, Paul~W. Goldberg, and Nicholas Teh.
\newblock Learning unanimously acceptable lotteries via queries.
\newblock In \emph{Proceedings of the 43rd International Conference on Machine Learning (ICML)}, 2026.
\newblock Extended version available as arXiv:2604.17505.

\bibitem[Clark et~al.(2020)Clark, Choi, Collins, Garrette, Kwiatkowski, Nikolaev, and Palomaki]{clark2020tydiqa}
Jonathan~H. Clark, Eunsol Choi, Michael Collins, Dan Garrette, Tom Kwiatkowski, Vitaly Nikolaev, and Jennimaria Palomaki.
\newblock {TyDi QA}: A benchmark for information-seeking question answering in typologically diverse languages.
\newblock \emph{Transactions of the Association for Computational Linguistics}, 8:\penalty0 454--470, 2020.

\bibitem[Ding et~al.(2024)Ding, Mallick, Wang, Sim, Mukherjee, R{\"u}hle, Lakshmanan, and Awadallah]{ding2024hybrid}
Dujian Ding, Ankur Mallick, Chi Wang, Robert Sim, Subhabrata Mukherjee, Victor R{\"u}hle, Laks V.~S. Lakshmanan, and Ahmed~Hassan Awadallah.
\newblock Hybrid {LLM}: Cost-efficient and quality-aware query routing.
\newblock In \emph{The 12th International Conference on Learning Representations (ICLR)}, 2024.

\bibitem[Elkind et~al.(2011)Elkind, Lang, and Saffidine]{elkind2011choosing}
Edith Elkind, J{\'e}r{\^o}me Lang, and Abdallah Saffidine.
\newblock Choosing collectively optimal sets of alternatives based on the condorcet criterion.
\newblock In \emph{Proceedings of the 22nd International Joint Conference on Artificial Intelligence (IJCAI)}, pages 186--191, 2011.

\bibitem[Even-Dar et~al.(2006)Even-Dar, Mannor, and Mansour]{evendar2006action}
Eyal Even-Dar, Shie Mannor, and Yishay Mansour.
\newblock Action elimination and stopping conditions for the multi-armed bandit and reinforcement learning problems.
\newblock \emph{Journal of Machine Learning Research}, 7:\penalty0 1079--1105, 2006.

\bibitem[Feige(1998)]{Feige1998}
Uriel Feige.
\newblock A threshold of $\ln n$ for approximating set cover.
\newblock \emph{Journal of the ACM}, 45\penalty0 (4):\penalty0 634--652, 1998.

\bibitem[Fujii et~al.(2024)Fujii, Nakamura, Loem, Iida, Ohi, Hattori, Shota, Mizuki, Yokota, and Okazaki]{fujii2024swallow}
Kazuki Fujii, Taishi Nakamura, Mengsay Loem, Hiroki Iida, Masanari Ohi, Kakeru Hattori, Hirai Shota, Sakae Mizuki, Rio Yokota, and Naoaki Okazaki.
\newblock Continual pre-training for cross-lingual {LLM} adaptation: Enhancing {Japanese} language capabilities.
\newblock In \emph{Conference on Language Modeling (COLM)}, 2024.

\bibitem[Gabillon et~al.(2013)Gabillon, Kveton, Wen, Eriksson, and Muthukrishnan]{gabillon2013adaptive}
Victor Gabillon, Branislav Kveton, Zheng Wen, Brian Eriksson, and S.~Muthukrishnan.
\newblock Adaptive submodular maximization in bandit setting.
\newblock In \emph{Proceedings of the 27th International Conference on Neural Information Processing Systems (NeurIPS)}, pages 2697--2705, 2013.

\bibitem[Golovin and Krause(2011)]{golovin2011adaptive}
Daniel Golovin and Andreas Krause.
\newblock Adaptive submodularity: Theory and applications in active learning and stochastic optimization.
\newblock \emph{Journal of Artificial Intelligence Research}, 42:\penalty0 427--486, 2011.

\bibitem[Gomes and Selman(2001)]{gomes2001algorithm}
Carla~P. Gomes and Bart Selman.
\newblock Algorithm portfolios.
\newblock \emph{Artificial Intelligence}, 126\penalty0 (1--2):\penalty0 43--62, 2001.

\bibitem[Haghtalab et~al.(2026)Haghtalab, Procaccia, Shao, Wang, and Yang]{haghtalab2026pluralistic}
Nika Haghtalab, Ariel~D. Procaccia, Han Shao, Serena~Lutong Wang, and Kunhe Yang.
\newblock Pluralistic leaderboards.
\newblock In \emph{Proceedings of the 43rd International Conference on Machine Learning}, 2026.

\bibitem[Halpern et~al.(2026)Halpern, Kehne, Procaccia, Tucker-Foltz, and W{\"u}thrich]{halpern2026representation}
Daniel Halpern, Gregory Kehne, Ariel~D Procaccia, Jamie Tucker-Foltz, and Manuel W{\"u}thrich.
\newblock Representation with incomplete votes.
\newblock \emph{Theory and Decision}, pages 1--40, 2026.

\bibitem[Hong and Page(2004)]{hong2004groups}
Lu~Hong and Scott~E Page.
\newblock Groups of diverse problem solvers can outperform groups of high-ability problem solvers.
\newblock \emph{Proceedings of the National Academy of Sciences}, 101\penalty0 (46):\penalty0 16385--16389, 2004.

\bibitem[Hu et~al.(2024)Hu, Bieker, Li, Jiang, Keigwin, Ranganath, Keutzer, and Upadhyay]{hu2024routerbench}
Qitian~Jason Hu, Jacob Bieker, Xiuyu Li, Nan Jiang, Benjamin Keigwin, Gaurav Ranganath, Kurt Keutzer, and Shriyash~Kaustubh Upadhyay.
\newblock {RouterBench}: A benchmark for multi-{LLM} routing system.
\newblock \emph{arXiv preprint arXiv:2403.12031}, 2024.

\bibitem[Huang et~al.(2024)Huang, Yu, Zhu, Sun, Cheng, Song, Chen, Alharthi, An, He, Liu, Chen, Li, Wang, Zhang, Sun, Wan, Li, and Xu]{huang2024acegpt}
Huang Huang, Fei Yu, Jianqing Zhu, Xuening Sun, Hao Cheng, Dingjie Song, Zhihong Chen, Mosen Alharthi, Bang An, Juncai He, Ziche Liu, Junying Chen, Jianquan Li, Benyou Wang, Lian Zhang, Ruoyu Sun, Xiang Wan, Haizhou Li, and Jinchao Xu.
\newblock {AceGPT}, localizing large language models in {Arabic}.
\newblock In \emph{Proceedings of the 2024 Conference of the North American Chapter of the Association for Computational Linguistics: Human Language Technologies (NAACL-HLT)}, pages 8139--8163, 2024.

\bibitem[Imber et~al.(2025)Imber, Israel, Brill, and Kimelfeld]{imber2025approval}
Aviram Imber, Jonas Israel, Markus Brill, and Benny Kimelfeld.
\newblock Approval-based committee voting under incomplete information.
\newblock \emph{Artificial Intelligence}, 347:\penalty0 104381, 2025.

\bibitem[Jiang et~al.(2023)Jiang, Ren, and Lin]{jiang2023llmblender}
Dongfu Jiang, Xiang Ren, and Bill~Yuchen Lin.
\newblock {LLM}-blender: Ensembling large language models with pairwise ranking and generative fusion.
\newblock In \emph{Proceedings of the 61st Annual Meeting of the Association for Computational Linguistics (ACL)}, pages 14165--14178, 2023.

\bibitem[Kerschke et~al.(2019)Kerschke, Hoos, Neumann, and Trautmann]{kerschke2019automated}
Pascal Kerschke, Holger~H. Hoos, Frank Neumann, and Heike Trautmann.
\newblock Automated algorithm selection: Survey and perspectives.
\newblock \emph{Evolutionary Computation}, 27\penalty0 (1):\penalty0 3--45, 2019.

\bibitem[Khalaf et~al.(2026)Khalaf, Wang, Halpern, Shapira, Calmon, and Procaccia]{khalaf2026robust}
Hadi Khalaf, Serena~L Wang, Daniel Halpern, Itai Shapira, Flavio du~Pin Calmon, and Ariel~D Procaccia.
\newblock Robust {AI} evaluation through maximal lotteries.
\newblock In \emph{Proceedings of the 43rd International Conference on Machine Learning (ICML)}, 2026.

\bibitem[Khattab et~al.(2023)Khattab, Singhvi, Maheshwari, Zhang, Santhanam, Vardhamanan, Haq, Sharma, Joshi, Moazam, Miller, Zaharia, and Potts]{khattab2023dspy}
Omar Khattab, Arnav Singhvi, Paridhi Maheshwari, Zhiyuan Zhang, Keshav Santhanam, Sri Vardhamanan, Saiful Haq, Ashutosh Sharma, Thomas~T. Joshi, Hanna Moazam, Heather Miller, Matei Zaharia, and Christopher Potts.
\newblock Dspy: Compiling declarative language model calls into self-improving pipelines.
\newblock \emph{arXiv preprint arXiv:2310.03714}, 2023.

\bibitem[Ko et~al.(2023)Ko, Yang, Ryu, Choi, Yang, Hyun, and Park]{ko2023polyglot}
Hyunwoong Ko, Kichang Yang, Minho Ryu, Taekyoon Choi, Seungmu Yang, Jiwung Hyun, and Sungho Park.
\newblock A technical report for {Polyglot-Ko}: Open-source large-scale {Korean} language models.
\newblock \url{https://www.eleuther.ai/papers-blog/polyglot-ko}, 2023.

\bibitem[Kotthoff(2014)]{kotthoff2014algorithm}
Lars Kotthoff.
\newblock Algorithm selection for combinatorial search problems: A survey.
\newblock \emph{{AI} Magazine}, 35:\penalty0 48--60, 2014.

\bibitem[Kuncheva(2004)]{kuncheva2004combining}
Ludmila~I. Kuncheva.
\newblock \emph{Combining Pattern Classifiers: Methods and Algorithms}.
\newblock Wiley, 2004.

\bibitem[Kurihara et~al.(2022)Kurihara, Kawahara, and Shibata]{kurihara2022jglue}
Kentaro Kurihara, Daisuke Kawahara, and Tomohide Shibata.
\newblock {JGLUE}: {Japanese} general language understanding evaluation.
\newblock In \emph{Proceedings of the 13th Language Resources and Evaluation Conference (LREC)}, pages 2957--2966, 2022.

\bibitem[Kwon et~al.(2023)Kwon, Li, Zhuang, Sheng, Zheng, Yu, Gonzalez, Zhang, and Stoica]{kwon2023vllm}
Woosuk Kwon, Zhuohan Li, Siyuan Zhuang, Ying Sheng, Lianmin Zheng, Cody~Hao Yu, Joseph~E. Gonzalez, Hao Zhang, and Ion Stoica.
\newblock Efficient memory management for large language model serving with {PagedAttention}.
\newblock In \emph{Proceedings of the 29th Symposium on Operating Systems Principles (SOSP)}, pages 611--626, 2023.

\bibitem[Lackner and Skowron(2023)]{lackner2023multiwinner}
Martin Lackner and Piotr Skowron.
\newblock \emph{Multi-Winner Voting with Approval Preferences}.
\newblock SpringerBriefs in Intelligent Systems. Springer, 2023.
\newblock \doi{10.1007/978-3-031-09016-5}.

\bibitem[Lanctot et~al.(2025)Lanctot, Larson, Kaisers, Berthet, Gemp, Diaz, Maura-Rivero, Bachrach, Koop, and Precup]{lanctot2025soft}
Marc Lanctot, Kate Larson, Michael Kaisers, Quentin Berthet, Ian Gemp, Manfred Diaz, Roberto-Rafael Maura-Rivero, Yoram Bachrach, Anna Koop, and Doina Precup.
\newblock Soft condorcet optimization for ranking of general agents.
\newblock In \emph{Proceedings of the 24th International Conference on Autonomous Agents and Multiagent Systems}, pages 1253--1262, 2025.

\bibitem[Lee(2023)]{koalpaca2023}
Junbum Lee.
\newblock {KoAlpaca}: {Korean} alpaca models.
\newblock \url{https://github.com/Beomi/KoAlpaca}, 2023.

\bibitem[LiCalzi and Surucu(2012)]{licalzi2012power}
Marco LiCalzi and Oktay Surucu.
\newblock The power of diversity over large solution spaces.
\newblock \emph{Management Science}, 58\penalty0 (7):\penalty0 1408--1421, 2012.

\bibitem[Lindeboom et~al.(2025)Lindeboom, Brehm, Grossi, and Murukannaiah]{lindeboom2025diverse}
Feline Lindeboom, Martijn Brehm, Davide Grossi, and Pradeep Murukannaiah.
\newblock Diverse committees with incomplete or inaccurate approval ballots.
\newblock \emph{arXiv preprint arXiv:2506.10843}, 2025.

\bibitem[Lu and Boutilier(2013)]{lu2013multiwinner}
Tyler Lu and Craig Boutilier.
\newblock Multi-winner social choice with incomplete preferences.
\newblock In \emph{Proceedings of the 23rd International Joint Conference on Artificial Intelligence (IJCAI)}, pages 263--270, 2013.

\bibitem[Manurangsi(2020)]{manurangsi2020maxkcoverage}
Pasin Manurangsi.
\newblock Tight running time lower bounds for strong inapproximability of maximum \(k\)-coverage, unique set cover and related problems (via \(t\)-wise agreement testing theorem).
\newblock In \emph{Proceedings of the 14th Annual ACM-SIAM Symposium on Discrete Algorithms (SODA)}, pages 62--81, 2020.

\bibitem[{McKinsey \& Company}(2025)]{mckinsey2025stateai}
{McKinsey \& Company}.
\newblock The state of {AI}: How organizations are rewiring to capture value.
\newblock \url{https://www.mckinsey.com/capabilities/quantumblack/our-insights/the-state-of-ai-how-organizations-are-rewiring-to-capture-value}, 2025.

\bibitem[Nemhauser et~al.(1978)Nemhauser, Wolsey, and Fisher]{NemhauserWolseyFisher1978}
George~L. Nemhauser, Laurence~A. Wolsey, and Marshall~L. Fisher.
\newblock An analysis of approximations for maximizing submodular set functions---i.
\newblock \emph{Mathematical Programming}, 14:\penalty0 265--294, 1978.

\bibitem[Ni et~al.(2025)Ni, Chen, Li, Chen, Li, Wang, Wang, Wang, Zhang, Fan, et~al.]{ni2025survey}
Shiwen Ni, Guhong Chen, Shuaimin Li, Xuanang Chen, Siyi Li, Bingli Wang, Qiyao Wang, Xingjian Wang, Yifan Zhang, Liyang Fan, et~al.
\newblock A survey on large language model benchmarks.
\newblock \emph{arXiv preprint arXiv:2508.15361}, 2025.

\bibitem[Ong et~al.(2025)Ong, Almahairi, Wu, Chiang, Wu, Gonzalez, Kadous, and Stoica]{ong2025routellm}
Isaac Ong, Amjad Almahairi, Vincent Wu, Wei-Lin Chiang, Tianhao Wu, Joseph~E. Gonzalez, M.~Waleed Kadous, and Ion Stoica.
\newblock {RouteLLM}: Learning to route {LLM}s with preference data.
\newblock In \emph{The 13th International Conference on Learning Representations (ICLR)}, 2025.
\newblock URL \url{https://openreview.net/forum?id=8sSqNntaMr}.

\bibitem[Rice(1976)]{rice1976algorithm}
John~R. Rice.
\newblock The algorithm selection problem.
\newblock \emph{Advances in Computers}, 15:\penalty0 65--118, 1976.

\bibitem[Skowron and Faliszewski(2017)]{SkowronFaliszewski2017}
Piotr Skowron and Piotr Faliszewski.
\newblock Chamberlin-courant rule with approval ballots: Approximating the maxcover problem with bounded frequencies in {FPT} time.
\newblock \emph{Journal of Artificial Intelligence Research}, 60:\penalty0 687--716, 2017.

\bibitem[Song et~al.(2026)Song, Nguyen, and Lin]{song2026few}
Haoyu Song, Th{\`a}nh Nguyen, and Young-San Lin.
\newblock A few good choices.
\newblock In \emph{Proceedings of the 2026 Annual ACM-SIAM Symposium on Discrete Algorithms (SODA)}, pages 4861--4874, 2026.

\bibitem[{UK AI Security Institute}(2024)]{aisi2024inspect}
{UK AI Security Institute}.
\newblock Inspect {AI}: A framework for large language model evaluations.
\newblock \url{https://github.com/UKGovernmentBEIS/inspect_ai}, 2024.

\bibitem[Wang et~al.(2023)Wang, Chen, Pei, Xie, Kang, Zhang, Xu, Xiong, Dutta, Schaeffer, et~al.]{wang2023decodingtrust}
Boxin Wang, Weixin Chen, Hengzhi Pei, Chulin Xie, Mintong Kang, Chenhui Zhang, Chejian Xu, Zidi Xiong, Ritik Dutta, Rylan Schaeffer, et~al.
\newblock Decodingtrust: A comprehensive assessment of trustworthiness in $\{$GPT$\}$ models.
\newblock In \emph{Proceedings of the 37th International Conference on Neural Information Processing Systems (NeurIPS)}, pages 31232--31339, 2023.

\bibitem[Wang et~al.(2024)Wang, Ma, Zhang, Ni, Chandra, Guo, Ren, Arulraj, He, Jiang, et~al.]{wang2024mmlu}
Yubo Wang, Xueguang Ma, Ge~Zhang, Yuansheng Ni, Abhranil Chandra, Shiguang Guo, Weiming Ren, Aaran Arulraj, Xuan He, Ziyan Jiang, et~al.
\newblock Mmlu-pro: A more robust and challenging multi-task language understanding benchmark.
\newblock \emph{Advances in Neural Information Processing Systems}, 37:\penalty0 95266--95290, 2024.

\bibitem[White et~al.(2025)White, Dooley, Roberts, Pal, Feuer, Jain, Shwartz-Ziv, Jain, Saifullah, Naidu, Hegde, LeCun, Goldstein, Neiswanger, and Goldblum]{white2025livebench}
Colin White, Samuel Dooley, Manley Roberts, Arka Pal, Ben Feuer, Siddhartha Jain, Ravid Shwartz-Ziv, Neel Jain, Khalid Saifullah, Siddartha Naidu, Chinmay Hegde, Yann LeCun, Tom Goldstein, Willie Neiswanger, and Micah Goldblum.
\newblock {LiveBench}: A challenging, contamination-limited {LLM} benchmark.
\newblock In \emph{International Conference on Learning Representations (ICLR)}, 2025.
\newblock Spotlight.

\bibitem[Xu et~al.(2008)Xu, Hutter, Hoos, and Leyton-Brown]{xu2008satzilla}
Lin Xu, Frank Hutter, Holger~H. Hoos, and Kevin Leyton-Brown.
\newblock {SATzilla}: Portfolio-based algorithm selection for {SAT}.
\newblock \emph{Journal of Artificial Intelligence Research}, 32:\penalty0 565--606, 2008.

\bibitem[Yuenyong et~al.(2024)Yuenyong, Viriyayudhakorn, Piyatumrong, and Jaroenkantasima]{yuenyong2024openthaigpt}
Sumeth Yuenyong, Kobkrit Viriyayudhakorn, Apivadee Piyatumrong, and Jillaphat Jaroenkantasima.
\newblock {OpenThaiGPT 1.5}: A {Thai}-centric open source large language model, 2024.
\newblock Cited for the OpenThaiGPT model series; the 1.0.0-7b-chat checkpoint used in our experiments predates this paper.

\bibitem[Zhou(2012)]{zhou2012ensemble}
Zhi-Hua Zhou.
\newblock \emph{Ensemble Methods: Foundations and Algorithms}.
\newblock Chapman and Hall/CRC, 2012.

\bibitem[Zimet et~al.(2026)Zimet, Alouf-Heffetz, and Talmon]{zimet2026query}
Itay~Asher Zimet, Shiri Alouf-Heffetz, and Nimrod Talmon.
\newblock Query-based committee selection.
\newblock \emph{arXiv preprint arXiv:2603.29729}, 2026.

\end{thebibliography}
\clearpage

\appendix
\etocdepthtag.toc{appendix}

\section*{Appendix}

\etocsettagdepth{main}{none}
\etocsettagdepth{appendix}{subsection}
\tableofcontents

\section{Additional Related Work}
\label{app:additional-related-work}

\textbf{Algorithm portfolios and per-instance selection.}
Outside LLMs, the closest older analogue is algorithm selection: choose a solver, or a portfolio of solvers, for instances drawn from a problem family \citep{gomes2001algorithm,kotthoff2014algorithm,kerschke2019automated,rice1976algorithm,xu2008satzilla}. This line of work explains why complementarity matters: no single solver has to be best on every instance. However, the usual goal is to learn a per-instance selector or solver schedule from a training set with rich instance features. We study a different pre-deployment question: using only limited binary or pairwise feedback, how many queries are needed to choose a small fixed committee that will cover future tasks well?

\medskip
\noindent
\textbf{Ensemble pruning and library selection.}
Classical ensemble learning also studies how to select a useful subset from a large library of trained models \citep{caruana2004ensemble,kuncheva2004combining,zhou2012ensemble}. These methods are usually empirical procedures for improving prediction accuracy on a validation set. Our setting keeps the subset-selection motivation but changes the feedback model and the guarantee: the algorithm may not see the full validation matrix, and the aim is a query bound for selecting a high-value committee rather than an empirical recipe for weighting or averaging models.

\medskip
\noindent
\textbf{Pure exploration with structured feedback.}
Each greedy step in our binary algorithm has a flavor of pure exploration in bandits: we need to identify a candidate with large marginal value, not maximize reward while the algorithm is running \citep{audibert2010bestarm,evendar2006action}. Combinatorial pure exploration studies related identification problems when the desired answer is a structured set of arms \citep{chen2014combinatorial}. Our problem has additional structure because samples are shared tasks and the value of a candidate depends on what the current committee already covers. This is why the query bounds depend on miss rates and marginal gaps, rather than only on the gaps between independent arms.

\medskip
\noindent
\textbf{Adaptive submodular and bandit submodular optimization.}
There is also related work on adaptive submodularity and submodular maximization with bandit feedback \citep{gabillon2013adaptive,golovin2011adaptive}. Those frameworks are useful background because coverage is submodular and observations arrive gradually. The role of adaptivity is different here: we are not choosing a sequence of deployment actions, but deciding which entries of a model-task or model-model evaluation table to reveal before selecting the final committee. This distinction is what allows failure-conditioned querying and the finite auditing step used for the pairwise setting.

\section{Supplementary Details for Section~\ref{sec:approval}}

\subsection{Additional Proofs in Section~\ref{subsec:approval-baseline}}

\paragraph{Proof of Theorem~\ref{thm:approval-erm}}
For each $S\in\mathcal S_k$ and each $t\in[T]$, define $Y_t(S):=U(d_t,S)\in\{0,1\}$.
Then $(Y_t(S))_{t=1}^T$ are i.i.d. with mean $\mathbb{E}[Y_t(S)] = \mathbb{E}_{d\sim D}[U(d,S)] = v(S)$, and $\widehat v_T(S) = \frac1T \sum_{t=1}^T Y_t(S)$.

Fix any size-$k$ committee $S$. By Hoeffding's inequality,
\[
\Pr\left(\left|\widehat v_T(S)-v(S)\right|>\frac{\eps}{2}\right)
\le 2\exp\left(-2T\left(\frac{\eps}{2}\right)^2\right)
= 2e^{-T\eps^2/2}.
\]
Applying a union bound over $\mathcal S_k$, with
$|\mathcal S_k|=\binom{m}{k}$, we get that 
\[
\Pr\left(\exists S\in\mathcal S_k:
\left|\widehat v_T(S)-v(S)\right|>\frac{\varepsilon}{2}\right)
\le
2\binom{m}{k}e^{-T\varepsilon^2/2}.
\]
Thus, if $T \ge \frac{2}{\eps^2} ( \ln \binom{m}{k} + \ln \frac{2}{\delta} )$, then with probability at least $1-\delta$ the event $E:=\{
\left|\widehat v_T(S)-v(S)\right|\le \frac{\varepsilon}{2}
\text{ for all } S\in\mathcal S_k \}$ holds.
Condition on $E$. Let $S^* \in \arg\max_{\substack{S\subseteq C\\|S|=k}} v(S)$ and $\widehat S_{\mathrm{ERM}}
\in \arg\max_{\substack{S\subseteq C\\|S|=k}} \widehat v_T(S)$.
Then
\[
v(\widehat S_{\mathrm{ERM}})
\ge \widehat v_T(\widehat S_{\mathrm{ERM}})-\frac{\eps}{2}
\ge \widehat v_T(S^*)-\frac{\eps}{2}
\ge v(S^*)-\eps
= \OPT_k-\eps.
\]
This proves the claim.

\paragraph{Proof of Theorem~\ref{thm:approval-lower}}

In our query model, sampled
instances are informative only through queried oracle answers. We therefore construct hard instances whose
utilities are hidden bit-vectors.

For each sampled instance draw a vector
\[
d=(d_1,\dots,d_m)\in\{0,1\}^m,
\]
and define
\[
u(d,c):=d_c \quad (c\in C).
\]
Thus querying expert $c$ on that instance reveals exactly the $c$th coordinate
of the hidden vector.

For each $i\in[m]$, let $D^{(i)}$ be the product distribution on $\{0,1\}^m$
such that
\[
d_i \sim \mathrm{Bernoulli}\left(\frac12+2\eps\right),
\quad
d_j \sim \mathrm{Bernoulli}\left(\frac12\right)
\quad\text{independently for all } j\neq i.
\]
Let $D^{(0)}$ be the product distribution with all coordinates i.i.d.
$\mathrm{Bernoulli}(1/2)$.

Under $D^{(i)}$,
\[
v(\{i\}) = \frac12+2\eps,
\quad
v(\{j\}) = \frac12 \quad \text{for all } j\neq i.
\]
Hence
\[
\OPT_1 = \frac12+2\eps,
\quad
\OPT_1-\eps = \frac12+\eps.
\]
Therefore any output $\widehat c$ satisfying
\[
v(\{\widehat c\}) \ge \OPT_1-\eps
\]
must equal $i$. Since the theorem assumes success probability at least $2/3$
for every instance distribution, we have
\[
\Pr_{D^{(i)}}[\widehat c = i] \ge \frac23
\quad\text{for every } i\in[m].
\]

Let $P_0$ denote the law of the full observable transcript under $D^{(0)}$,
including the algorithm's internal randomness, all queried pairs, all observed
oracle answers, and the final output. Let $P_i$ be the analogous transcript law
under $D^{(i)}$. Write
\[
p_i := P_0(\widehat c = i).
\]
Since $\sum_{i=1}^m p_i = 1$, at least $m-1$ indices satisfy $p_i \le 1/2$.
Fix such an index $i$, and define the event
\[
A_i := \{\widehat c = i\}.
\]
Then
\[
P_i(A_i)-P_0(A_i)\ge \frac23-\frac12 = \frac16,
\]
so $\|P_i-P_0\|_{\mathrm{TV}} \ge 1/6$. By Pinsker's inequality,
\[
\mathrm{KL}(P_0\|P_i) \ge 2\|P_i-P_0\|_{\mathrm{TV}}^2 \ge \frac1{18}.
\]

Let $Q$ be the total number of queries made before termination, and for each
$j\in C$ let $N_j$ be the number of queries made to expert $j$, so $Q=\sum_{j\in C}N_j$.
Since re-querying the same task-expert pair is redundant by the preliminaries,
we may assume without loss of generality that no task-expert pair is queried twice:
an algorithm that repeats such a query can instead reuse the earlier answer, with
no loss in success probability and no larger query complexity.

If $\mathbb E_{P_0}[Q]=\infty$, then the desired lower bound already holds under
$D^{(0)}$. Hence assume $\mathbb E_{P_0}[Q]<\infty$. For the change-of-measure calculation,
include the algorithm's private randomness in the transcript; equivalently, apply the adaptive
KL chain rule to the transcript stopped at $Q\wedge n$ and let $n\to\infty$.
Conditional on any past transcript, the algorithm's next action has the same conditional law under
$D^{(0)}$ and $D^{(i)}$; only the oracle answer distribution can differ. By the product construction and the no-repetition
assumption, a query to an expert $j\neq i$ has the same conditional law under
$D^{(0)}$ and $D^{(i)}$. A fresh query to expert $i$ has conditional law
$\mathrm{Bernoulli}(1/2)$ under $D^{(0)}$ and conditional law
$\mathrm{Bernoulli}(1/2+2\varepsilon)$ under $D^{(i)}$. Therefore the
adaptive chain rule for KL divergence, applied to the stopped transcript, gives
\[
\mathrm{KL}(P_0\|P_i)
=
\mathbb E_{P_0}[N_i]\,
\mathrm{KL}\left(
\mathrm{Bernoulli}\left(\frac12\right)
\,\middle\|\,
\mathrm{Bernoulli}\left(\frac12+2\varepsilon\right)
\right).
\]

For the Bernoulli divergence,
\[
\mathrm{KL}\left(
\mathrm{Bernoulli}\left(\frac12\right)
\,\middle\|\,
\mathrm{Bernoulli}\left(\frac12+2\eps\right)
\right)
=
\frac12 \ln\left(\frac{1}{1-16\eps^2}\right)
\le
\frac12\cdot \frac{16\eps^2}{1-16\eps^2}
\le
\frac{32}{3}\eps^2,
\]
where we used $-\ln(1-x)\le x/(1-x)$ and $1-16\eps^2\ge 3/4$ since
$\eps\le 1/8$. Combining the previous three displays yields
\[
\mathbb{E}_{P_0}[N_i]
\ge
\frac{1/18}{(32/3)\eps^2}
=
\frac{1}{192\,\eps^2}.
\]
This lower bound holds for at least $m-1$ indices $i$, so
\[
\mathbb{E}_{P_0}[Q]
=
\sum_{j=1}^m \mathbb{E}_{P_0}[N_j]
\ge
(m-1)\cdot \frac{1}{192\,\eps^2}
=
\Omega(m/\eps^2).
\]
Since $D^{(0)}$ is one admissible instance distribution, the algorithm's
worst-case expected query complexity is $\Omega(m/\eps^2)$.

\paragraph{Proof of Theorem~\ref{thm:approval-full-greedy}}
We begin with the structural property needed by the empirical greedy algorithm.

\begin{lemma} \label{lem:approval-submod}
    In the binary feedback model, $v$ and, for every realization $d_1,\dots,d_T$, $\widehat v_T$ are normalized, monotone, and submodular.
\end{lemma}
\begin{proof}
    Fix an instance $d$.

We first show that the set function $S \mapsto U(d,S)$ is monotone. If
$S \subseteq T$ and $U(d,S)=1$, then there exists some $c\in S$ with
$u(d,c)=1$. Since $S \subseteq T$, the same expert belongs to $T$, so
$U(d,T)=1$. Hence
\[
U(d,S)\le U(d,T)
\quad
\text{whenever } S\subseteq T.
\]

We next verify diminishing returns. Fix $S \subseteq R \subseteq C$ and $x \in C \setminus R$.
If $U(d,R)=1$, then
\[
U(d,R\cup\{x\})-U(d,R)=0,
\]
while
\[
U(d,S\cup\{x\})-U(d,S)\ge 0,
\]
so the diminishing-returns inequality holds. If instead $U(d,R)=0$, then no member of $R$
solves $d$, hence also $U(d,S)=0$. In this case
\[
U(d,S\cup\{x\})-U(d,S)=u(d,x)
\quad\text{and}\quad
U(d,R\cup\{x\})-U(d,R)=u(d,x),
\]
so the two marginals are equal. 
Therefore $S \mapsto U(d,S)$ is submodular.

Since $U(d,\varnothing)=0$, we have
$v(\varnothing)=\widehat v_T(\varnothing)=0$. Also,
\[
    v(S)=\mathbb E_{d\sim D}[U(d,S)]
    \quad\text{and}\quad
    \widehat v_T(S)=\frac1T\sum_{t=1}^T U(d_t,S).
\]
Expectations and averages preserve monotonicity and submodularity. Hence both
$v$ and $\widehat v_T$ are normalized, monotone, and submodular.
\end{proof}

We now apply this lemma to the fully elicited empirical objective.

Set $\alpha := \frac{\eps}{2-1/e}$.
Exactly as in the proof of Theorem~\ref{thm:approval-erm}, Hoeffding's inequality and a union
bound over the $\binom{m}{k}$ committees of size $k$ imply that if
\[
T \ge \frac{1}{2\alpha^2}
\left(
\ln \binom{m}{k} + \ln \frac{2}{\delta}
\right),
\]
then with probability at least $1-\delta$,
\[
\left|\widehat v_T(S)-v(S)\right| \le \alpha
\quad
\text{for all } S\subseteq C \text{ with } |S|=k.
\]
Condition on this event.

Let $S^* \in \arg\max_{\substack{S\subseteq C\\|S|=k}} v(S)$, so that $v(S^*)=\OPT_k$. By Lemma~\ref{lem:approval-submod}, the empirical objective $\widehat v_T$ is normalized, monotone,
and submodular. Therefore the standard greedy guarantee under a cardinality-$k$
constraint gives
\[
\widehat v_T(\widehat S_{\rm gr})
\ge
(1-1/e)\max_{|S|\le k}\widehat v_T(S)
=
(1-1/e)\max_{S\in\mathcal S_k}\widehat v_T(S)
\ge
(1-1/e)\widehat v_T(S^*),
\]
where the equality uses monotonicity.
Using the uniform concentration event twice,
\[
\begin{aligned}
v(\widehat S_{\mathrm{gr}})
&\ge \widehat v_T(\widehat S_{\mathrm{gr}})-\alpha \\
&\ge (1-1/e)\widehat v_T(S^*)-\alpha \\
&\ge (1-1/e)(v(S^*)-\alpha)-\alpha \\
&= (1-1/e)\,\OPT_k-(2-1/e)\alpha \\
&= (1-1/e)\,\OPT_k-\eps.
\end{aligned}
\]
The query count is exactly $Q=mT$, because exhaustive evaluation queries every
pair $(d_t,c)$ with $t\in[T]$ and $c\in C$.

\subsection{Additional Proofs in Section \ref{subsec:adaptive_query}} \label{app:apdx-adaptive}

\subsubsection{Deferred pseudocode for failure-conditioned elimination} \label{app:failcond-elim}

\begin{algorithm}[t]
\caption{\textsc{FailCond-Elim}$(S,A,\eta,\delta)$: $\eta$-optimal selection on failures of $S$}
\label{alg:failcond-elim}
\textbf{Input:} committee $S$ with $\rho(S)>0$; nonempty candidate set $A\subseteq C\setminus S$; accuracy $\eta\in(0,1]$; confidence $\delta\in(0,1)$.\\
\textbf{Output:} a candidate $\widehat c\in A$.
\begin{algorithmic}[1]
\State $A_1 \leftarrow A$; initialize $\widehat q_0(c)\leftarrow 0$ for all $c\in A$.
\For{$r=1,2,\dots$}
    \State Draw $d_r \sim D$ conditioned on $U(d_r,S)=0$ (via rejection sampling).
    \ForAll{$c\in A_r$}
        \State Query $X_{r,c}\leftarrow u(d_r,c)\in\{0,1\}$ and set $\widehat q_r(c)\leftarrow
\frac{(r-1)\widehat q_{r-1}(c)+X_{r,c}}{r}$.
    \EndFor
    \State $\mathrm{rad}_r \leftarrow \sqrt{\frac{1}{2r}\ln\left(4|A|r^2/\delta\right)}$.
    \State Let $c_r^* \in \arg\max_{c\in A_r}\widehat q_r(c)$.
    \State $A_{r+1}\leftarrow \Big\{c\in A_r:\ \widehat q_r(c)+\mathrm{rad}_r \ge
            \widehat q_r(c_r^*)-\mathrm{rad}_r-\eta \Big\}$.
    \If{$\mathrm{rad}_r \le \eta/4$ \ \textbf{or}\ $|A_{r+1}|=1$}
        \State \textbf{return} any $\widehat c \in \arg\max_{c\in A_{r+1}}\widehat q_r(c)$.
    \EndIf
\EndFor
\end{algorithmic}
\end{algorithm}

\paragraph{Failure-conditioned marginal identity}
The following identity is the reason that the greedy marginal step can be estimated using only tasks missed by the current committee.

\begin{lemma}
\label{lem:approval-failure-identity}
In the binary feedback model, for every $S\subseteq C$ and
$c\in C\setminus S$, $\Delta(c\mid S)=\rho(S)q(c\mid S)$.
In particular, for fixed $S$ with $\rho(S)>0$, maximizing
$\Delta(c\mid S)$ over $c\in C\setminus S$ is equivalent to maximizing
$q(c\mid S)$.
\end{lemma}

\begin{proof}
Fix $S\subseteq C$ and $c\in C\setminus S$. Since utilities are binary,
\[
U(d,S\cup\{c\})-U(d,S)
=
\begin{cases}
1, & \text{if } U(d,S)=0 \text{ and } u(d,c)=1,\\
0, & \text{otherwise.}
\end{cases}
\]
Therefore
\[
\Delta(c\mid S)
=
\Pr_{d\sim D}[U(d,S)=0 \text{ and } u(d,c)=1].
\]
If $\rho(S)=0$, then the probability in the previous display is zero, while $q(c\mid S)=0$ by the convention in the definition of $q$; hence both sides are zero. If $\rho(S)>0$, then
conditional probability gives
\[
\Delta(c\mid S)
=
\Pr[U(d,S)=0]\Pr[u(d,c)=1\mid U(d,S)=0]
=
\rho(S)q(c\mid S).
\]
The final claim follows because, for the fixed $S$ under consideration, $\rho(S)>0$ is a constant independent of $c$.
\end{proof}

\paragraph{Proof of Theorem~\ref{thm:failcond-elim}}

    Let $q^* := \max_{c\in A} q(c\mid S)$. Since $S$ is fixed and $\rho(S)>0$, rejection sampling produces accepted failure instances that are i.i.d. from the conditional law $D \mid U(d,S)=0$.

    For the analysis, fix an infinite i.i.d. sequence $d_1,d_2,\dots$ from this conditional law and define
    \[
        X_{r,c} := u(d_r,c), \quad r\ge 1,\ c\in A.
    \]
    For each fixed $c\in A$, the variables $(X_{r,c})_{r\ge1}$ are i.i.d.
    Bernoulli with mean $q(c\mid S)$. For the analysis, define for every
    $c\in A$ and $r\ge1$
    \[
        \widehat q_r(c):=\frac1r\sum_{s=1}^r X_{s,c}.
    \]
    The algorithm reveals $X_{r,c}$ only while $c$ remains active. Because the
    active sets are nested, whenever $c\in A_r$, the displayed quantity
    $\widehat q_r(c)$ is exactly the empirical mean maintained by the algorithm.
    
    For any fixed $c\in A$ and $r\ge 1$, Hoeffding's inequality gives us
    \[
    \Pr\left(\left|\widehat q_r(c)-q(c\mid S)\right|>\mathrm{rad}_r\right)
    \le 2e^{-2r\mathrm{rad}_r^2}
    =
    \frac{\delta}{2|A|r^2}.
    \]
    Therefore, by a union bound over all $c\in A$ and all $r\ge 1$,
    \[
    \Pr\left(\exists c\in A,\exists r\ge 1:
    \left|\widehat q_r(c)-q(c\mid S)\right|>\mathrm{rad}_r\right)
    \le
    \sum_{c\in A}\sum_{r\ge 1}\frac{\delta}{2|A|r^2}
    < \delta.
    \]
    Consequently, with probability at least $1-\delta$, the event
    \[
    E:=\left\{\forall c\in A,\forall r\ge 1:
    \left|\widehat q_r(c)-q(c\mid S)\right|\le \mathrm{rad}_r\right\}
    \]
    holds. Condition on $E$.
    Let $c^*\in\arg\max_{c\in A} q(c\mid S)$, so that $q(c^*\mid S)=q^*$.
    
    \textbf{Step 1: $c^*$ is never eliminated.}
    If at round $r$ the elimination rule removed $c^*$, then
    \[
    \widehat q_r(c^*)+\mathrm{rad}_r
    <
    \widehat q_r(c_r^*)-\mathrm{rad}_r-\eta,
    \]
    where $c_r^*\in\arg\max_{c\in A_r}\widehat q_r(c)$. Under $E$,
    \[
    \widehat q_r(c^*)+\mathrm{rad}_r \ge q^*
    \quad\text{and}\quad
    \widehat q_r(c_r^*)-\mathrm{rad}_r-\eta \le q^*-\eta,
    \]
    a contradiction. Thus $c^*$ remains active throughout.

    \textbf{Step 2: correctness of the returned candidate.}
    There are two stopping cases.
    
    If $|A_{r+1}|=1$, then by Step~1 the sole surviving candidate must be $c^*$, so the algorithm
    returns an optimal candidate.
    
    Otherwise the algorithm stops because $\mathrm{rad}_r\le \eta/4$. Let $\widehat c$ be the returned
    candidate, chosen to maximize $\widehat q_r$ over $A_{r+1}$. Since $c^*\in A_{r+1}$ and
    $\widehat q_r(\widehat c)\ge \widehat q_r(c^*)$,
    \[
    q(\widehat c\mid S)
    \ge \widehat q_r(\widehat c)-\mathrm{rad}_r
    \ge \widehat q_r(c^*)-\mathrm{rad}_r
    \ge q^*-2\mathrm{rad}_r
    \ge q^*-\eta.
    \]
    
    \textbf{Step 3: number of accepted failure instances.}
    Let
    \[
        r_0 := \left\lceil K \frac{\log(e|A|/(\delta\eta))}{\eta^2}\right\rceil,
    \]
    where $K$ is a sufficiently large universal constant. Since
    $\log(e|A|/(\delta\eta))\ge 1$, we have
    \[
        \log\log(e|A|/(\delta\eta)) \le \log(e|A|/(\delta\eta)).
    \]
    Also, because $|A|/\delta \ge 1$ and $\eta \le 1$,
    \[
        \log(1/\eta) \le \log(e|A|/(\delta\eta)).
    \]
    Therefore
    \[
        \log r_0
        \le \log(2K) + 2\log(1/\eta)
           + \log\log(e|A|/(\delta\eta))
        \le C \log(e|A|/(\delta\eta))
    \]
    for a universal constant $C$. Hence
    \[
        \log\left(\frac{4|A|r_0^2}{\delta}\right)
        \le C' \log(e|A|/(\delta\eta))
    \]
    for another universal constant $C'$. Choosing $K$ large enough gives
    \[
        \mathrm{rad}_{r_0}
        =
        \sqrt{\frac{1}{2r_0}\log\left(\frac{4|A|r_0^2}{\delta}\right)}
        \le \frac{\eta}{4}.
    \]
    Thus Algorithm 1 must stop by round $r_0$, and consequently
    \[
        R = O\left(\frac{\log(e|A|/(\delta\eta))}{\eta^2}\right)
    \]
    deterministically.
    
    \textbf{Step 4: gap-dependent candidate-query bound.}
    At round $r$, the algorithm queries exactly the candidates in $A_r$, so
$Q_{\rm cand}=\sum_{r=1}^R |A_r|$. Fix $c\in A$.
    
    If $\Delta_q(c) \ge 2\eta$ and $\mathrm{rad}_r < \Delta_q(c)/8$, then under $E$,
    \begin{align*}
    \widehat q_r(c) + \mathrm{rad}_r
    \le q(c \mid S) + 2\mathrm{rad}_r 
    = q^* - \Delta_q(c) + 2\mathrm{rad}_r 
    &< q^* - \eta - 2\mathrm{rad}_r \\
    &\le \widehat q_r(c^*) - \mathrm{rad}_r - \eta \\
    &\le \widehat q_r(c_r^*) - \mathrm{rad}_r - \eta,
    \end{align*}
    so $c \notin A_{r+1}$. 
    By the same calculation as in Step~3, with $\eta$ replaced by $\Delta_q(c)$ and constants adjusted from $1/4$ to $1/8$, there is a round
    \[
        r =
        O\left(
            \frac{\ln(e|A|/(\delta\Delta_q(c)))}{\Delta_q(c)^2}
        \right)
    \]
    for which $\mathrm{rad}_r<\Delta_q(c)/8$. Hence, if $\Delta_q(c)\ge 2\eta$, candidate $c$ is queried at most
    \[
        O\left(
            \frac{\ln(e|A|/(\delta\eta))}{\Delta_q(c)^2}
        \right)
    \]
    times, since
    \[
        \ln(e|A|/(\delta\Delta_q(c)))
        \le
        \ln(e|A|/(\delta\eta)).
    \]
    If $\Delta_q(c)<2\eta$, then even in the worst case $c$ survives only until Algorithm~\ref{alg:failcond-elim} stops, which by Step~3 is
    \[
        O\left(
            \frac{\ln(e|A|/(\delta\eta))}{\eta^2}
        \right)
        =
        O\left(
            \frac{\ln(e|A|/(\delta\eta))}
                 {\max\{\Delta_q(c),\eta\}^2}
        \right),
    \]
    because $\eta \le \max\{\Delta_q(c),\eta\}<2\eta$. Combining the two cases shows that, on the event $E$, candidate $c$ is queried at most
    \[
        O\left(
            \frac{\ln(e|A|/(\delta\eta))}
                 {\max\{\Delta_q(c),\eta\}^2}
        \right)
    \]
    times, and summing over $c\in A$ proves the bound on $Q_{\mathrm{cand}}$.

\subsection{\textsc{Adaptive-Fail-Greedy} and Theorem \ref{thm:adaptive-fail-greedy}}\label{sec:AFG}

\begin{algorithm}[t]
\caption{\textsc{Adaptive-Fail-Greedy}$(k,\varepsilon,\delta)$}
\label{alg:adaptive-fail-greedy}
\begin{algorithmic}[1]
\State $S_0\gets\varnothing$
\For{$i=0,1,\dots,k-1$}
    \State $\rho_i\gets 1-v(S_i)$
    \If{$\rho_i\le \varepsilon/k$}
        \State Add any $k-i$ remaining experts to reach size $k$
        \State \Return the resulting committee
    \EndIf
    \State $\eta_i\gets \varepsilon/(k\rho_i)$ and $\delta_i\gets\delta/k$
    \State $\widehat c_{i+1}
        \gets \textsc{FailCond-Elim}(S_i,C\setminus S_i,\eta_i,\delta_i)$
    \State $S_{i+1}\gets S_i\cup\{\widehat c_{i+1}\}$
\EndFor
\State \Return $S_k$
\end{algorithmic}
\end{algorithm}

Algorithm~\ref{alg:adaptive-fail-greedy} is written in conceptual form: at step $i$, it uses the true miss rate $\rho_i = 1 - v(S_i)$ both to decide whether to stop early and to set the accuracy parameter $\eta_i$. We next show how to make this step implementable. The idea is to use the same rejection-sampling transcript both to certify that the miss rate is not already small and to warm-start the failure-conditioned elimination routine.

For a fixed committee $S \subseteq C$ with $\rho(S)>0$, let
$d_1,d_2,\dots$ be i.i.d. from $D$, and let $B_t(S) := \mathbf{1}\{U(d_t,S)=0\}$.
For each $r\ge 1$, define the $r$-th failure time
\[
N_r(S) := \inf\left\{n\ge 1:\sum_{t=1}^n B_t(S)=r\right\}.
\]
Thus $N_r(S)$ is the number of unconditional draws needed to obtain
$r$ accepted failures of $S$. Since the tasks are i.i.d., conditional on
any realization of the acceptance indicators $(B_t(S))_{t\ge 1}$, the
accepted tasks, in their observed order, are independent with common law
$D \mid U(d,S)=0$. This conditional product law does not depend on the
realized indicator sequence, so the same remains true after conditioning on
any event determined by these indicators, such as the value of $N_r(S)$ or
the event $N_r(S)\le M$. Hence Algorithm~\ref{alg:failcond-elim} may use an already stored
initial segment of accepted failures and then continue with fresh accepted
failures generated by rejection sampling. The concentration proof of
Theorem~\ref{thm:failcond-elim} applies to the entire resulting accepted-failure sequence. By the
query-access convention in the preliminaries, sampled instances may be stored
and later queried on adaptively chosen experts, so this warm start is fully
implementable.

\begin{lemma}
\label{lem:miss-rate-certification}
Fix $S \subseteq C$ with $\rho(S)>0$, and let $r\ge 1$. Then
\[
\Pr\left[N_r(S) < \frac{r}{2\rho(S)}\right] \le e^{-r/6},
\quad
\Pr\left[N_r(S) > \frac{2r}{\rho(S)}\right] \le e^{-r/4}.
\]
Consequently,
\[
\Pr\left[\rho(S) \le \frac{2r}{N_r(S)} \le 4\rho(S)\right]
\ge 1 - 2e^{-r/6}.
\]
\end{lemma}

\begin{proof}
Write $\rho := \rho(S)$.

For the lower tail, let $n_- := \lfloor r/(2\rho)\rfloor$. If $n_-=0$, then
the event $N_r(S) < r/(2\rho)$ is impossible. Otherwise define
\[
X_- := \sum_{t=1}^{n_-} B_t(S).
\]
Then $X_- \sim \mathrm{Binomial}(n_-,\rho)$ with mean
$\mu_- := \mathbb{E}[X_-] = n_-\rho \le r/2$. If $N_r(S) < r/(2\rho)$, then
$X_- \ge r$. Since $r = (1+\alpha)\mu_-$ for
$\alpha := r/\mu_- - 1 \ge 1$, the multiplicative Chernoff bound gives
\[
\Pr\left[N_r(S) < \frac{r}{2\rho}\right]
\le \Pr[X_- \ge r]
\le \exp\left(-\frac{\alpha\mu_-}{3}\right)
= \exp\left(-\frac{r-\mu_-}{3}\right)
\le e^{-r/6}.
\]

For the upper tail, let $n_+ := \lceil 2r/\rho\rceil$ and define $X_+ := \sum_{t=1}^{n_+} B_t(S)$.
Then $X_+ \sim \mathrm{Binomial}(n_+,\rho)$ with mean
$\mu_+ := \mathbb E[X_+] = n_+\rho \ge 2r$. If $N_r(S)>2r/\rho$, then,
since $N_r(S)$ is integer and $n_+=\lceil 2r/\rho\rceil$, after $n_+$
draws there have been at most $r$ accepted failures, so $X_+\le r$. Since
$r\le \mu_+/2$, the lower-tail Chernoff bound yields
\[
\Pr\left[N_r(S)>\frac{2r}{\rho}\right]
\le \Pr[X_+\le r]
\le \Pr[X_+\le \mu_+/2]
\le e^{-\mu_+/8}
\le e^{-r/4}.
\]

Finally, if $N_r(S) \in [\,r/(2\rho),\,2r/\rho\,]$, then $\rho \le \frac{2r}{N_r(S)} \le 4\rho$.
A union bound over the two tail events proves our result.
\end{proof}

\begin{algorithm}[t]
\caption{\textsc{CertifyMiss}$(S,r_0,M_0)$}
\label{alg:certify-miss}
\begin{algorithmic}[1]
\State $F \leftarrow \varnothing$, $N \leftarrow 0$.
\While{$N < M_0$ and $|F| < r_0$}
    \State Draw $d \sim D$ and set $N \leftarrow N+1$.
    \State Evaluate $U(d,S)$ by querying experts in $S$ until either one queried expert returns $1$ or all queried experts return $0$.
    \If{$U(d,S)=0$}
        \State append $d$ to $F$.
    \EndIf
\EndWhile
\If{$|F| < r_0$}
    \State return \textsf{STOP}.
\Else
    \State return $(F,N)$.
\EndIf
\end{algorithmic}
\end{algorithm}

\begin{algorithm}[t]
\caption{\textsc{Adaptive-Fail-Greedy-Implementable}$(k,\varepsilon,\delta)$}
\label{alg:adaptive-fail-greedy-impl}
\begin{algorithmic}[1]
\State $r_0 \leftarrow \left\lceil 6\ln(4k/\delta)\right\rceil$.
\State $M_0 \leftarrow \left\lfloor 2r_0k/\varepsilon \right\rfloor$.
\State $S_0 \leftarrow \varnothing$.
\For{$i=0,1,\dots,k-1$}
    \State Run \textsc{CertifyMiss}$(S_i,r_0,M_0)$.
    \If{\textsc{CertifyMiss} returns \textsf{STOP}}
        \State Add any $k-i$ remaining experts and return the resulting committee.
    \EndIf
    \State Let $(F_i,N_i)$ be the returned pair.
    \State Set
    \[
    \widetilde\rho_i \leftarrow \min\left\{1,\frac{2r_0}{N_i}\right\},
    \quad
    \eta_i \leftarrow \frac{\varepsilon}{k\widetilde\rho_i},
    \quad
    \delta_i \leftarrow \frac{\delta}{2k}.
    \]
    \State Run Algorithm~\ref{alg:failcond-elim} on input $(S_i, C\setminus S_i, \eta_i, \delta_i)$, but replace the first $r_0$ accepted failures on line~3 by the stored failures in $F_i$, processed in the order they were observed. Let $\widehat c_{i+1}$ be the returned candidate.
    \State $S_{i+1} \leftarrow S_i \cup \{\widehat c_{i+1}\}$.
\EndFor
\State return $S_k$.
\end{algorithmic}
\end{algorithm}

Lemma~\ref{lem:miss-rate-certification} gives the certification guarantee used by Algorithm~\ref{alg:adaptive-fail-greedy-impl}. On the event considered in the proof below, every executed step $i$ satisfies $\rho_i \le \widetilde\rho_i \le 4\rho_i$.
Thus $[\widetilde\rho_i/4,\widetilde\rho_i]$ is a valid constant-factor confidence interval for the unknown miss rate $\rho_i$.

\begin{theorem}
\label{thm:implementable-failgreedy}
Fix $k \in \{1,\dots,m\}$ and $\varepsilon,\delta \in (0,1)$. Let $\widehat S$
be the output of Algorithm~\ref{alg:adaptive-fail-greedy-impl}, and let
$Q_{\mathrm{cand}}$ denote the total number of candidate queries made inside its
warm-started calls to Algorithm~\ref{alg:failcond-elim}. For each loop iteration $i$ reached by Algorithm~\ref{alg:adaptive-fail-greedy-impl}, let $\rho_i := 1-v(S_i)$, and, for $c\in C\setminus S_i$, let $\Delta_i(c)
:= \max_{c'\in C\setminus S_i}\Delta(c'\mid S_i)-\Delta(c\mid S_i)$.
Then, with probability at least $1-\delta$, $v(\widehat S) \ge (1-1/e)\mathrm{OPT}_k - \varepsilon$.
Moreover,
\[
Q_{\mathrm{cand}}
=
O\left(
\sum_i \sum_{c \in C\setminus S_i}
\frac{\rho_i^2 \ln(emk^2/(\delta\varepsilon))}
{\max\{\Delta_i(c),\varepsilon/k\}^2}
\right),
\]
where the outer sum ranges over the steps that actually call Algorithm~\ref{alg:failcond-elim}.
\end{theorem}

\begin{proof}
Fix a step $i$ and condition on the history up to the start of that step, so
$S_i$ is fixed. If $\rho_i=0$, then \textsc{CertifyMiss} necessarily returns
\textsf{STOP}, so suppose $\rho_i>0$. Let $\mathcal{E}_i^{\mathrm{cert}}$ be
the event from Lemma~\ref{lem:miss-rate-certification} that
\[
\rho_i \le \frac{2r_0}{N_{r_0}(S_i)} \le 4\rho_i.
\]
By the choice of $r_0$,
\[
\Pr\left[(\mathcal{E}_i^{\mathrm{cert}})^c \,\middle|\, \text{history up to step }i\right]
\le 2e^{-r_0/6}
\le \frac{\delta}{2k}.
\]

On $\mathcal{E}_i^{\mathrm{cert}}$, the certification phase behaves as follows.

First, suppose Algorithm~\ref{alg:adaptive-fail-greedy-impl} stops at step $i$.
Then $N_{r_0}(S_i) > M_0$. If $\rho_i > \varepsilon/k$, then
\[
N_{r_0}(S_i)
\le \frac{2r_0}{\rho_i}
< \frac{2r_0k}{\varepsilon}.
\]
Since $N_{r_0}(S_i)$ is integer and
$M_0 = \lfloor 2r_0k/\varepsilon \rfloor$, this implies
$N_{r_0}(S_i) \le M_0$, a contradiction. Hence every early stop satisfies
\[
\rho_i \le \varepsilon/k.
\]

Second, suppose step $i$ proceeds. Then $N_i = N_{r_0}(S_i) \le M_0$, and so
\[
\frac{2r_0}{N_i}
\ge
\frac{2r_0}{M_0}
\ge
\frac{\varepsilon}{k}.
\]
Therefore $\widetilde\rho_i \ge \varepsilon/k$, which implies
$\eta_i = \varepsilon/(k\widetilde\rho_i) \in (0,1]$. Also, on
$\mathcal{E}_i^{\mathrm{cert}}$,
\[
\rho_i
\le
\widetilde\rho_i
=
\min\left\{1,\frac{2r_0}{N_i}\right\}
\le
4\rho_i.
\]
Hence
\[
\frac{\varepsilon}{4k\rho_i}
\le
\eta_i
=
\frac{\varepsilon}{k\widetilde\rho_i}
\le
\frac{\varepsilon}{k\rho_i}.
\]

Now let $\mathcal E_i^{\rm elim}$ be the event that the warm-started call
to Algorithm~\ref{alg:failcond-elim} at step $i$, if made, satisfies the accuracy and
candidate-query guarantees of Theorem~\ref{thm:failcond-elim}. Conditional on the history up to
step $i$ and on the full certification transcript at that step, the
quantities $N_i,\widetilde\rho_i,\eta_i,\delta_i$ are fixed. Moreover, by
the warm-start observation above, on the event that step $i$ calls
Algorithm~\ref{alg:failcond-elim}, the stored failures $F_i$, followed by the fresh accepted
failures generated during the warm-started call, form an i.i.d. sequence from
$D\mid U(d,S_i)=0$. Since $|S_i|=i<k\le m$, the set $C\setminus S_i$
is nonempty, and Theorem~\ref{thm:failcond-elim} applies. Thus, after integrating over the
certification transcript,
\[
\Pr\left(
(\mathcal E_i^{\rm elim})^c
\,\middle|\,
\text{history up to step }i,\,
\mathcal E_i^{\rm cert},\,
\text{step }i\text{ calls Algorithm~\ref{alg:failcond-elim}}
\right)
\le \delta_i=\frac{\delta}{2k}.
\]
By the adaptive union bound,
\[
\Pr\left[
\begin{array}{l}
\text{some reached step violates }\mathcal E_i^{\rm cert},\text{ or}\\
\text{some reached step that calls Algorithm~\ref{alg:failcond-elim} violates }\mathcal E_i^{\rm elim}
\end{array}
\right]
\le
\sum_{i=0}^{k-1}\frac{\delta}{2k}
+
\sum_{i=0}^{k-1}\frac{\delta}{2k}
=
\delta.
\]
Work on the complementary event.

If the algorithm stops early at step $i$, then $\rho_i \le \varepsilon/k$, and
therefore
\[
\OPT_k - v(S_i)
\le
1-v(S_i)
=
\rho_i
\le
\varepsilon/k
\le
\varepsilon.
\]
After padding $S_i$ with arbitrary remaining experts to size $k$, monotonicity
gives
\[
v(\widehat S) \ge v(S_i) \ge \OPT_k - \varepsilon,
\]
which is stronger than the claimed bound.

It remains to consider the case in which no early stop occurs. Then the algorithm
executes all $k$ additions, so $\widehat S = S_k$. Fix
$S^* \in \arg\max_{|S|=k} v(S)$ and define
\[
g_i := \OPT_k - v(S_i), \quad i=0,1,\dots,k.
\]
By monotonicity and submodularity,
\[
\OPT_k
=
v(S^*)
\le
v(S_i \cup S^*)
\le
v(S_i) + \sum_{c\in S^* \setminus S_i} \Delta(c \mid S_i)
\le
v(S_i) + k \max_{c\in C\setminus S_i}\Delta(c \mid S_i),
\]
hence
\[
\max_{c\in C\setminus S_i}\Delta(c \mid S_i) \ge \frac{g_i}{k}.
\]

Since step $i$ is executed, $\mathcal{E}_i^{\mathrm{elim}}$ and Theorem~\ref{thm:failcond-elim}
give us
\[
q(\widehat c_{i+1} \mid S_i)
\ge
\max_{c\in C\setminus S_i} q(c \mid S_i) - \eta_i.
\]
Applying Lemma~\ref{lem:approval-failure-identity} and using $\widetilde\rho_i \ge \rho_i$, we obtain
\[
\Delta(\widehat c_{i+1} \mid S_i)
=
\rho_i q(\widehat c_{i+1} \mid S_i)
\ge
\max_{c\in C\setminus S_i}\Delta(c \mid S_i) - \rho_i\eta_i
\ge
\max_{c\in C\setminus S_i}\Delta(c \mid S_i) - \frac{\varepsilon}{k}
\ge
\frac{g_i}{k} - \frac{\varepsilon}{k}.
\]
Therefore
\[
g_{i+1}
=
g_i - \Delta(\widehat c_{i+1} \mid S_i)
\le
\left(1-\frac{1}{k}\right)g_i + \frac{\varepsilon}{k}.
\]
Unrolling the recurrence gives
\[
g_k
\le
\left(1-\frac{1}{k}\right)^k \OPT_k
+
\sum_{j=0}^{k-1}\left(1-\frac{1}{k}\right)^j \frac{\varepsilon}{k}
\le
e^{-1}\OPT_k + \varepsilon,
\]
and hence
\[
v(\widehat S) = \OPT_k - g_k \ge (1-1/e)\OPT_k - \varepsilon.
\]

For the candidate-query bound, fix a step $i$ that calls Algorithm~\ref{alg:failcond-elim}, and
let $Q_{\mathrm{cand}}^{(i)}$ be the number of candidate queries made in
this warm-started call. Define
\[
\Delta_q^{(i)}(c)
:=
\max_{c'\in C\setminus S_i} q(c'\mid S_i)-q(c\mid S_i).
\]
Theorem~\ref{thm:failcond-elim} applied to the warm-started call with $A = C\setminus S_i$ and
$\delta_i = \delta/(2k)$ gives us
\[
Q_{\mathrm{cand}}^{(i)}
=
O\left(
\sum_{c \in C\setminus S_i}
\frac{\ln(e|C\setminus S_i|/(\delta_i\eta_i))}
{\max\{\Delta_q^{(i)}(c),\eta_i\}^2}
\right).
\]
Because $|C\setminus S_i| \le m$, $\delta_i = \delta/(2k)$, and
$\eta_i \ge \varepsilon/k$, the logarithmic factor satisfies
\[
\ln(e|C\setminus S_i|/(\delta_i\eta_i))
\le
\ln(2emk^2/(\delta\varepsilon))
=
\mathcal{O}(\ln(emk^2/(\delta\varepsilon))).
\]
Also, Lemma~\ref{lem:approval-failure-identity} implies $\Delta_i(c) = \rho_i \Delta_q^{(i)}(c)$,
and since $\widetilde\rho_i \le 4\rho_i$,
\[
\eta_i = \frac{\varepsilon}{k\widetilde\rho_i} \ge \frac{\varepsilon}{4k\rho_i}.
\]
Therefore
\[
\max\{\Delta_q^{(i)}(c),\eta_i\}
\ge
\frac{1}{\rho_i}\max\left\{\Delta_i(c),\frac{\varepsilon}{4k}\right\},
\]
so
\[
\frac{1}{\max\{\Delta_q^{(i)}(c),\eta_i\}^2}
\le
\frac{\rho_i^2}{\max\{\Delta_i(c),\varepsilon/(4k)\}^2}
\le
16 \cdot \frac{\rho_i^2}{\max\{\Delta_i(c),\varepsilon/k\}^2}.
\]
Substituting this into the previous display and summing over executed steps proves
\[
Q_{\mathrm{cand}}
=
O\left(
\sum_i \sum_{c \in C\setminus S_i}
\frac{\rho_i^2 \ln(emk^2/(\delta\varepsilon))}
{\max\{\Delta_i(c),\varepsilon/k\}^2}
\right). \qedhere
\]
\end{proof}

\paragraph{Committee evaluation cost.}
Theorem~\ref{thm:implementable-failgreedy} counts only candidate queries made after a failure instance has been accepted. The implementable algorithm also spends queries to test whether an unconditional draw is a failure of the current committee.

At any loop iteration $i$ reached by Algorithm~\ref{alg:adaptive-fail-greedy-impl}, the certification phase
examines at most
\[
M_0 = O\left(\frac{k\log(k/\delta)}{\varepsilon}\right)
\]
unconditional draws. More sharply, if $\rho_i>0$, then
\[
\mathbb E\left[\min\{N_{r_0}(S_i),M_0\}\right]
\le
\min\left\{\frac{r_0}{\rho_i},M_0\right\}
=
O\left(\frac{\log(k/\delta)}{\max\{\rho_i,\varepsilon/k\}}\right).
\]
When $\rho_i=0$, the same $O$-bound holds because the left-hand side is at
most $M_0$. If an unconditional draw $d\sim D$ is tested against $S_i$
using order $\sigma_i$, let $Q_{\mathrm{eval}}(d,S_i,\sigma_i)$ denote
the number of binary task-expert queries used to decide whether
$U(d,S_i)=0$. Conditioning on the history at the start of the iteration and
applying Wald's identity to the bounded stopping time
$\min\{N_{r_0}(S_i),M_0\}$, the certification phase contributes
\[
O\left(
\mathbb E[Q_{\mathrm{eval}}(d,S_i,\sigma_i)]
\cdot
\frac{\log(k/\delta)}{\max\{\rho_i,\varepsilon/k\}}
\right)
\]
queries in expectation at step $i$.

At a reached step $i$ that calls Algorithm~\ref{alg:failcond-elim}, condition on the history and
the certification transcript. The warm-started call then needs, in expectation,
at most
\[
O\left(
\frac{\log(emk^2/(\delta\varepsilon))}{\rho_i\eta_i^2}
\right)
=
O\left(
\frac{\widetilde\rho_i^{\,2}k^2\log(emk^2/(\delta\varepsilon))}
{\rho_i\varepsilon^2}
\right)
\]
fresh unconditional draws to generate accepted failures. On the
high-probability event used in Theorem~\ref{thm:implementable-failgreedy}, $\widetilde\rho_i\le 4\rho_i$,
so this is
\[
O\left(
\frac{\rho_i k^2\log(emk^2/(\delta\varepsilon))}{\varepsilon^2}
\right).
\]
Multiplying by $\mathbb E[Q_{\mathrm{eval}}(d,S_i,\sigma_i)]$ gives the same
committee-evaluation term as in the discussion after Theorem~\ref{thm:adaptive-fail-greedy}, up to
universal constants, plus the additive certification overhead above.

\paragraph{Proof of Theorem~\ref{thm:adaptive-fail-greedy}}

    For each loop iteration $i$ reached by Algorithm~\ref{alg:adaptive-fail-greedy}, condition on the
    history up to the moment Algorithm~\ref{alg:adaptive-fail-greedy} checks the early-stop condition at step
    $i$. On the event that Algorithm~\ref{alg:adaptive-fail-greedy} does not stop there and therefore calls
    Algorithm~\ref{alg:failcond-elim} with inputs $(S_i,C\setminus S_i,\eta_i,\delta_i)$,
    Theorem~\ref{thm:failcond-elim} implies that this call simultaneously satisfies its
    $\eta_i$-accuracy guarantee and its candidate-query bound with conditional
    probability at least $1-\delta_i$. Hence the probability that the call made
    at step $i$, if any, violates either guarantee is at most
    $\delta_i=\delta/k$. By the adaptive union bound over the at most $k$
    reached iterations, with probability at least $1-\delta$, every call to
    Algorithm~\ref{alg:failcond-elim} made by Algorithm~\ref{alg:adaptive-fail-greedy} satisfies both guarantees. Work on this event.
    
    If Algorithm~\ref{alg:adaptive-fail-greedy} stops early at step $i$, then $\rho_i \le \eps/k$. Since $OPT_k \le 1$, we have
    \[
    \OPT_k - v(S_i) \le 1 - v(S_i) = \rho_i \le \eps/k \le \eps.
    \]
    After padding $S_i$ with arbitrary remaining experts to size $k$, monotonicity gives
    \[
    v(\widehat S)\ge v(S_i)\ge \OPT_k-\eps,
    \]
    which is stronger than the claimed value bound.
    
    It remains to consider the case in which no early stop occurs. Then Algorithm~\ref{alg:adaptive-fail-greedy} executes all $k$ greedy additions, so $\widehat S=S_k$. Fix $S^*\in\arg\max_{|S|=k} v(S)$ and define
    \[
    g_i:=\OPT_k-v(S_i), \quad i=0,\dots,k.
    \]
    By monotonicity and submodularity,
    \[
    \OPT_k = v(S^*)\le v(S_i\cup S^*)
    \le v(S_i)+\sum_{c\in S^*\setminus S_i}\Delta(c\mid S_i)
    \le v(S_i)+k\max_{c\in C\setminus S_i}\Delta(c\mid S_i),
    \]
    and thus,
    \[
    \max_{c\in C\setminus S_i}\Delta(c\mid S_i)\ge \frac{g_i}{k}.
    \]
    Since step $i$ is executed, we have $\rho_i>\eps/k$, hence $\eta_i=\eps/(k\rho_i)\in(0,1)$. Theorem~\ref{thm:failcond-elim} therefore gives us
    \[
    q(\widehat c_{i+1}\mid S_i)\ge \max_{c\in C\setminus S_i} q(c\mid S_i)-\eta_i.
    \]
    Applying Lemma~\ref{lem:approval-failure-identity} and using $\rho_i\eta_i=\eps/k$, we obtain
    \[
    \Delta(\widehat c_{i+1}\mid S_i)
    = \rho_i q(\widehat c_{i+1}\mid S_i)
    \ge \max_{c\in C\setminus S_i}\Delta(c\mid S_i)-\frac{\eps}{k}
    \ge \frac{g_i}{k}-\frac{\eps}{k}.
    \]
    Therefore
    \[
    g_{i+1}=g_i-\Delta(\widehat c_{i+1}\mid S_i)
    \le \left(1-\frac1k\right)g_i+\frac{\eps}{k}.
    \]
    Unrolling the recurrence gives us
    \[
    g_k\le \left(1-\frac1k\right)^k \OPT_k
    +\sum_{i=0}^{k-1}\left(1-\frac1k\right)^{k-1-i}\frac{\eps}{k}
    \le e^{-1}\OPT_k+\eps,
    \]
    and hence
    \[
    v(\widehat S)=\OPT_k-g_k\ge (1-1/e)\,\OPT_k-\eps.
    \]
    
    For the candidate-query bound, fix any executed step $i$, and let
$Q_{\mathrm{cand}}^{(i)}$ be the number of candidate queries made in the
step-$i$ call to Algorithm~\ref{alg:failcond-elim}. Theorem~\ref{thm:failcond-elim} with $A=C\setminus S_i$,
$\delta_i=\delta/k$, and $\eta_i=\varepsilon/(k\rho_i)$ gives
    \[
    Q_{\mathrm{cand}}^{(i)}
    = \mathcal{O} \left(
    \sum_{c\in C\setminus S_i}
    \frac{\ln(e|C\setminus S_i|/(\delta_i\eta_i))}
    {\max\{\Delta_q^{(i)}(c),\eta_i\}^2}
    \right),
    \]
    where
    \[
    \Delta_q^{(i)}(c):=
    \max_{c'\in C\setminus S_i} q(c'\mid S_i)-q(c\mid S_i).
    \]
    Because $|C\setminus S_i|\le m$, $\delta_i=\delta/k$, and $\rho_i\le 1$,
    \[
    \ln(e|C\setminus S_i|/(\delta_i\eta_i))
    \le \ln(e m k^2/(\delta\eps)).
    \]
    Moreover, Lemma~\ref{lem:approval-failure-identity} implies
    \[
    \Delta_i(c)=\rho_i\Delta_q^{(i)}(c)
    \quad\text{and}\quad
    \eta_i=\frac{\eps}{k\rho_i},
    \]
    so
    \[
    \max\{\Delta_q^{(i)}(c),\eta_i\}
    =
    \frac1{\rho_i}\max\{\Delta_i(c),\eps/k\}.
    \]
    Substituting yields
    \[
    Q_{\mathrm{cand}}^{(i)}
    = \mathcal{O}\left(
    \sum_{c\in C\setminus S_i}
    \frac{\rho_i^2 \ln(e m k^2/(\delta\eps))}
    {\max\{\Delta_i(c),\eps/k\}^2}
    \right).
    \]
    Summing over the executed steps proves the displayed bound for $Q_{\mathrm{cand}}$.

\subsection{Proof of Theorem~\ref{thm:onestep_lb}}

By the query-access convention in the preliminaries, a sampled task is observed only as a handle to its latent binary profile. Hence a failure sample is informative only through oracle answers $u(d,a)$ to queried task-candidate pairs $(d,a)$. It is enough to construct a hard conditional law for $d \mid U(d,S)=0$. We realize it by an unconditional law with $U(d,S)=0$ almost surely, so $\rho(S)=1$.

Fix $c^*\in A$. We prove the slightly stronger statement that the hard law can realize any prescribed gaps
$\gamma_a\in(0,1/4]$, $a\in A\setminus\{c^*\}$. Under the base law, samples are i.i.d.; for each sample $d$, set
\[
u(d,b)=0 \quad \text{for all } b\in C\setminus A,
\]
and let the coordinates $(u(d,a))_{a\in A}$ be independent with
\[
u(d,c^*)\sim \mathrm{Bernoulli}\left(\frac12\right), \quad
u(d,a)\sim \mathrm{Bernoulli}\left(\frac12-\gamma_a\right)
\quad \text{for all } a\in A\setminus\{c^*\}.
\]
Since $A\subseteq C\setminus S$, we have $S\subseteq C\setminus A$, and therefore $U(d,S)=0$ almost surely. Thus this is a valid conditional failure law. Under the base law,
\[
q(c^*\mid S)=\frac12, \quad
q(a\mid S)=\frac12-\gamma_a
\quad \text{for all } a\in A\setminus\{c^*\}.
\]
Consequently,
\[
q^*=\frac12, \quad
\Delta_q(c^*)=0, \quad
\Delta_q(a)=\gamma_a
\quad \text{for all } a\in A\setminus\{c^*\}.
\]

Now fix any $c\in A\setminus\{c^*\}$ with $\gamma_c>\eta$, and define the
alternative conditional law by changing only the distribution of $u(d,c)$ to
\[
u(d,c)\sim \mathrm{Bernoulli}\left(\frac12+\gamma_c\right),
\]
leaving the law of every other coordinate $u(d,a)$ unchanged. 
Under this alternative law, candidate $c$ has rescue rate $1/2+\gamma_c$, whereas every other
candidate in $A$ has rescue rate at most
\[
\frac12 < \frac12+\gamma_c-\eta.
\]
Thus $c$ is the only candidate in $A$ satisfying
\[
q(c\mid S) \ge \max_{a\in A} q(a\mid S)-\eta
\]
under the alternative law.

Let $P_0$ denote the distribution of the algorithm's full transcript under the base law, including
its internal randomness, all queried task-candidate pairs, the observed oracle answers, and the final
output. Let $P_c$ denote the analogous transcript distribution under the alternative conditional law
for candidate $c$. Let
\[
E_c := \{\widehat c = c\}.
\]
Under the base law,
\[
q(c\mid S)=q^*-\gamma_c<q^*-\eta,
\]
so $E_c$ is contained in the failure event of the assumed guarantee. Hence $P_0(E_c)\le \delta$.
Under the alternative law, $c$ is the only candidate satisfying the required guarantee, so
$P_c(E_c)\ge 1-\delta$.
    
    Since $P_0(E_c)\le \delta$ and $P_c(E_c^c)\le \delta$, the Bretagnolle-Huber inequality gives
    \[
    2\delta \ge P_0(E_c)+P_c(E_c^c)
    \ge \frac12 e^{-\mathrm{KL}(P_0\|P_c)}.
    \]
    Hence
    \[
    \mathrm{KL}(P_0\|P_c) \ge \ln\left(\frac{1}{4\delta}\right).
    \]
    Also,
    \[
    \|P_0-P_c\|_{\mathrm{TV}} \ge P_c(E_c)-P_0(E_c) \ge 1-2\delta,
    \]
    so Pinsker's inequality gives
    \[
    \mathrm{KL}(P_0\|P_c) \ge 2(1-2\delta)^2 \ge \frac12.
    \]
    Therefore, for all $\delta \in (0,1/4]$,
    \[
    \mathrm{KL}(P_0\|P_c)
    \ge
    \max\left\{\ln\left(\frac{1}{4\delta}\right), \frac12\right\}
    \ge
    \frac{1}{2\ln 4 + 1}\ln\left(\frac{1}{\delta}\right).
    \]
    
    Let $N_c$ be the number of oracle queries whose expert coordinate is $c$. Repeated queries to the same task-candidate pair, if made, are included in $N_c$, although they add no new information. Apply the adaptive chain rule for KL divergence to the stopped transcript. If $\mathbb{E}_{P_0}[N_c]=\infty$, the desired lower bound for this candidate is immediate, so suppose $\mathbb{E}_{P_0}[N_c]<\infty$.

    Conditional on the past transcript before each query, the algorithm's next action has the same conditional law under $P_0$ and $P_c$; only the conditional law of the next oracle answer can differ. In the construction above, the coordinates $(u(d,a))_{a\in A}$ are independent on each sampled failure instance, and the alternative law changes only the marginal distribution of candidate $c$. Hence queries whose expert coordinate is not $c$ contribute $0$ to the KL divergence. A first query to candidate \(c\) on a sampled task contributes
\[
\mathrm{KL}\left(
\mathrm{Bernoulli}\left(\frac12-\gamma_c\right)
\,\middle\|\,
\mathrm{Bernoulli}\left(\frac12+\gamma_c\right)
\right),
\]
while a repeated query to the same task-candidate pair contributes \(0\), since the answer is already determined by the previous transcript. Therefore,
    \[
    \mathrm{KL}(P_0\|P_c)
    \le
    \mathbb{E}_{P_0}[N_c]\,
    \mathrm{KL}\left(
    \mathrm{Bernoulli}\left(\frac12-\gamma_c\right)
    \,\middle\|\,
    \mathrm{Bernoulli}\left(\frac12+\gamma_c\right)
    \right).
    \]
    
    For $\gamma:=\gamma_c\le 1/4$, set
    \[
    p=\frac12-\gamma,
    \quad
    q=\frac12+\gamma.
    \]
    Using $\ln x\le x-1$,
    \[
    \begin{aligned}
    \mathrm{KL}(\mathrm{Bernoulli}(p)\|
    \mathrm{Bernoulli}(q))
    &=
    p\ln\frac{p}{q}+(1-p)\ln\frac{1-p}{1-q} \\
    &\le
    p\frac{p-q}{q}+(1-p)\frac{q-p}{1-q} \\
    &=
    \frac{(p-q)^2}{q(1-q)} \\
    &=
    \frac{(2\gamma)^2}{(\frac12+\gamma)(\frac12-\gamma)} \\
    &\le \frac{64}{3}\gamma^2,
    \end{aligned}
    \]
    where the last inequality uses $\gamma\le 1/4$. Combining the previous displays gives
    \[
    \mathbb E_{P_0}[N_c]
    \ge
    \frac{3}{64(2\ln 4+1)}
    \cdot
    \frac{\ln(1/\delta)}{\gamma_c^2}
    \]
    for every $c\in A\setminus\{c^*\}$ with $\gamma_c>\eta$.
    
    Finally, for $a\in A$, let $N_a$ be the number of queries made to candidate
    $a$. Since these queries are included in $Q_{\mathrm{cand}}$,
    \[
    Q_{\mathrm{cand}} \ge \sum_{a\in A} N_a.
    \]
    The lower bound above holds under the same base transcript law $P_0$ for every
    $c\in A\setminus\{c^*\}$ with $\gamma_c>\eta$. Since, under the base
    conditional law, $\Delta_q(c)=\gamma_c$ for $c\ne c^*$ and
    $\Delta_q(c^*)=0$, we obtain
    \[
    \mathbb E_{P_0}[Q_{\mathrm{cand}}]
    \ge
    \sum_{c\in A:\,\Delta_q(c)>\eta} \mathbb E_{P_0}[N_c]
    =
    \Omega\left(
    \sum_{c\in A:\,\Delta_q(c)>\eta}
    \frac{\ln(1/\delta)}{\Delta_q(c)^2}
    \right).
    \]
    Thus the base conditional law is the claimed law.

\section{Supplementary Details for Section~\ref{sec:pairwise}}
\subsection{Additional Proofs in Section \ref{sec:pairwise}}

\paragraph{Proof of Lemma \ref{lem:theta-lifts-to-committees}}

Fix a ranking $\pi$. On the event $r_\pi(S)>r_\pi(S')$, the committee $S'$ must be nonempty, since $r_\pi(\varnothing)=m+1$ and $S\neq\varnothing$ implies $r_\pi(S)\le m$.
Let $x\in S'$ attain $r_\pi(S')$. This $x$ cannot belong to $S$, because otherwise
\[
r_\pi(S)\le \mathrm{rank}_\pi(x)=r_\pi(S'),
\]
contradicting $r_\pi(S)>r_\pi(S')$. Hence, on this event, there exists $x\in S'\setminus S$ such that
\[
\mathrm{rank}_\pi(x)=r_\pi(S')<r_\pi(S).
\]
Therefore, pointwise for every $\pi$,
\[
\{r_\pi(S)>r_\pi(S')\}\subseteq
\bigcup_{x\in S'\setminus S}\{\mathrm{rank}_\pi(x)<r_\pi(S)\}.
\]

Since $S$ is $\vartheta$-winning, Definition~\ref{def:theta-winning} gives, for every
$x\in C\setminus S$,
\[
\Pr_{\pi\sim P_D}\left[r_\pi(S)<\mathrm{rank}_\pi(x)\right]
\ge \vartheta.
\]
Because $S$ is nonempty, $x\notin S$, and rankings are strict, the events
$r_\pi(S)<\mathrm{rank}_\pi(x)$ and
$\mathrm{rank}_\pi(x)<r_\pi(S)$ are complementary. Thus
\[
\Pr_{\pi\sim P_D}\left[\mathrm{rank}_\pi(x)<r_\pi(S)\right]
\le 1-\vartheta.
\]
By the union bound,
\[
\Pr_{\pi\sim P_D}\left[r_\pi(S)>r_\pi(S')\right]
\le
\sum_{x\in S'\setminus S}
\Pr_{\pi\sim P_D}\left[\mathrm{rank}_\pi(x)<r_\pi(S)\right]
\le |S'\setminus S|(1-\vartheta)
\le |S'|(1-\vartheta).
\]
Taking complements gives
\[
\Pr_{\pi\sim P_D}\left[r_\pi(S)\le r_\pi(S')\right]
\ge 1-|S'|(1-\vartheta).
\]

Finally, suppose $S\cap S'=\varnothing$. Then $r_\pi(S)\neq r_\pi(S')$ for
every $\pi$: if $S'\neq\varnothing$, this follows from strict rankings, and
if $S'=\varnothing$, then $r_\pi(S')=m+1$ while $S\neq\varnothing$ implies
$r_\pi(S)\le m$. Hence
\[
r_\pi(S)\le r_\pi(S')
\quad\Longleftrightarrow\quad
r_\pi(S)<r_\pi(S').
\]
Using the definition of $\mathrm{WIN}$ from the preliminaries, the tie term is
zero, so
\[
\mathrm{WIN}(S,S')
=
\Pr_{\pi\sim P_D}\left[r_\pi(S)<r_\pi(S')\right]
\ge 1-|S'|(1-\vartheta).
\]

\subsection{Additional Proofs in Section \ref{sec:theta-full-info}}

We first record the full-information optimization landscape for the deterministic $\theta$-objective. Throughout this subsection, the learner is given a finite profile
$P=(\pi_1,\dots,\pi_n)$ of full rankings. For a nonempty proper committee
$S\subsetneq C$ and $x\in C\setminus S$, define
\[
\mathrm{WIN}_P(S,\{x\})
:=
\frac{1}{n}\sum_{t=1}^n
\mathbf{1}\{r_{\pi_t}(S)<\mathrm{rank}_{\pi_t}(x)\}.
\]
We say that $S$ covers $x$ on ranking $\pi_t$ if
$r_{\pi_t}(S)<\mathrm{rank}_{\pi_t}(x)$. Let
\[
\theta_P(S):=\min_{x\in C\setminus S}\mathrm{WIN}_P(S,\{x\})
\]
for every nonempty proper committee $S\subsetneq C$. For a target size
$1\le k<m$, write
\[
\theta^*_{k,P}:=\max_{S\in\mathcal{S}_k}\theta_P(S).
\]

We will use the following finite-profile monotonicity fact: if
$S\subseteq T\subsetneq C$ are nonempty, then $\theta_P(T)\ge \theta_P(S)$.
Indeed, for every $x\in C\setminus T$ and every $t$,
$r_{\pi_t}(T)\le r_{\pi_t}(S)$, so
\[
\begin{aligned}
\theta_P(T)
&=\min_{x\in C\setminus T}\mathrm{WIN}_P(T,\{x\})  \\
&\ge \min_{x\in C\setminus T}\mathrm{WIN}_P(S,\{x\}) \\
&\ge \min_{x\in C\setminus S}\mathrm{WIN}_P(S,\{x\})
=\theta_P(S).
\end{aligned}
\]

\begin{theorem}[Full-information approximability of $\theta$]\label{app:theta-basic}
For every $\gamma \in (0,1)$, there is an algorithm which,
given a finite profile $P=(\pi_1,\dots,\pi_n)$ of full rankings
and a budget $1\le k<m$, outputs $S\in\mathcal S_k$ satisfying $\theta_P(S)\ge (1-\gamma)\theta^*_{k,P}$ in time $m^{O(1/\gamma)}\mathrm{poly}(m,n)$.

Assuming Gap-ETH, there is no EPTAS for this problem: no algorithm can achieve the same guarantee in time $f(1/\gamma)\mathrm{poly}(m,n)$ for an arbitrary computable function $f$.
\end{theorem}

\begin{proof}[Proof of Theorem~\ref{thm:theta-full-info}/\ref{app:theta-basic}]
Let $c_0:=9.8217$ and
\[
    K_\gamma := \left\lceil \frac{c_0}{\gamma} \right\rceil .
\]
We use the structural theorem of Charikar et al.~\cite[Theorem~2]{charikar2025six}:
for every
integer $K \le m$, every finite profile admits a committee $U$ with $|U|\le K$
such that no outside candidate is ranked above every member of $U$ by more than
a $c_0/K$ fraction of the voters. For any $x\in C\setminus U$, strictness of the
rankings implies that the event
\[
        r_{\pi_t}(U) < \mathrm{rank}_{\pi_t}(x)
\]
is the complement of the event that $x$ is ranked above every member of $U$.
Thus the structural theorem gives
\[
        \theta_P(U) \ge 1-\frac{c_0}{K}.
\]

First suppose $k\ge K_\gamma$. Since $k<m$, we have $K_\gamma<m$. Enumerate all
nonempty committees $U\subseteq C$ with $|U|\le K_\gamma$, compute $\theta_P(U)$
for each, and choose one maximizing $\theta_P$. Applying the structural theorem
with $K=K_\gamma$, the enumerated family contains some $U$ with
\[
        \theta_P(U) \ge 1-\frac{c_0}{K_\gamma} \ge 1-\gamma .
\]
Pad $U$ arbitrarily to a committee $S\in\mathcal S_k$. By monotonicity,
\[
        \theta_P(S) \ge 1-\gamma .
\]
Since $\theta^*_{k,P}\le 1$, this implies
\[
        \theta_P(S) \ge (1-\gamma)\theta^*_{k,P}.
\]

Now suppose $k<K_\gamma$. Enumerate all nonempty committees
$U\subseteq C$ with $|U|\le k$, compute $\theta_P(U)$ for each,
and choose one maximizing $\theta_P$; call it $U^*$. Since every
committee of size at most $k$ can be padded to size exactly $k$ without
decreasing its winning value,
\[
    \max_{1\le |U|\le k}\theta_P(U)
    =
    \max_{S\in\mathcal S_k}\theta_P(S)
    =
    \theta^*_{k,P}.
\]
Pad $U^*$ arbitrarily to a committee $S\in\mathcal S_k$. By
monotonicity,
\[
    \theta_P(S)\ge \theta_P(U^*)=\theta^*_{k,P},
\]
so the algorithm returns an optimal size-$k$ committee in this case.

For any fixed committee $U$, the value $\theta_P(U)$ can be computed
in polynomial time by scanning the $n$ rankings and checking, for each
outside expert $x$, how many rankings place the best member of $U$
above $x$. The enumeration size is $m^{O(K_\gamma)}=m^{O(1/\gamma)}$.

We reduce from the Gap-ETH-hard promise version of Max $\kappa$-Coverage due
to \citet{manurangsi2020maxkcoverage}. Let $N$ denote the encoding size of
the Max $\kappa$-Coverage instance. There is a constant $\eta\in(0,1)$
such that, assuming Gap-ETH, for every computable function $f$ there is no
algorithm running in time $f(\kappa)\mathrm{poly}(N)$ that distinguishes
the following two cases: given a universe
$\mathcal U=\{e_1,\dots,e_h\}$, a family of sets
$\mathcal F=\{F_c:c\in B\}$, and an integer $\kappa$,
\[
\textsc{Yes:}\quad
\exists Y\subseteq B,\ |Y|\le \kappa,\quad
\bigcup_{c\in Y}F_c=\mathcal U,
\]
from
\[
\textsc{No:}\quad
\forall Y\subseteq B,\ |Y|\le \kappa,\quad
\left|\bigcup_{c\in Y}F_c\right|\le (1-\eta)h .
\]
For each element $e_i$, write
\[
B_i:=\{c\in B:e_i\in F_c\}.
\]

Given such an instance, construct a finite profile $P$ as follows.
Let
\[
q:=\kappa+3,
\quad
D:=\{a_1,\dots,a_q\},
\]
and introduce one additional candidate $b$.  The candidate set is the
disjoint union
\[
C:=B\dot\cup D\dot\cup \{b\}.
\]
We set the committee size in the $\theta$-instance to be
\[
k:=\kappa+1 .
\]
The profile $P$ has two types of rankings.  We use block notation:
$X\succ Y$ means that every member of $X$ is ranked above every
member of $Y$, with an arbitrary but fixed order inside each block.

First, for every $i\in[h]$ and every $j\in[q]$, include one type-I
ranking
\[
B_i \succ a_j \succ D\setminus\{a_j\}
      \succ B\setminus B_i \succ b .
\]
Second, for every $j\in[q]$, include
\[
M:=hq
\]
copies of the type-II ranking
\[
a_j \succ b \succ B \succ D\setminus\{a_j\}.
\]
Thus the total number of rankings is
\[
n=hq+qM=hq(q+1).
\]
The construction is polynomial in the size of the Max $\kappa$-Coverage
instance.

We first prove completeness.  Suppose there is
$Y\subseteq B$, $|Y|\le \kappa$, covering all elements of
$\mathcal U$.  Consider $S_0:=Y\cup\{b\}$, and pad it arbitrarily
to a committee $S\in\mathcal S_k$.  Padding cannot decrease
$\theta_P$, so it is enough to lower bound $\theta_P(S_0)$.

Let $c\in B\setminus Y$.  In every type-II ranking, the candidate
$b\in S_0$ is ranked above $c$.  Hence $c$ is beaten by $S_0$
in all $qM$ type-II rankings, giving
\[
\mathrm{WIN}_P(S_0,\{c\})\ge \frac{qM}{n}
=\frac{q}{q+1}.
\]
Now let $a_j\in D\setminus S_0$.  In every type-I ranking associated
with element $e_i$, some member of $Y\cap B_i$ is ranked above
$a_j$, because $Y$ covers $e_i$.  Thus all $hq$ type-I rankings
cover $a_j$.  In addition, among the type-II rankings, $b$ is ranked
above $a_j$ in all blocks except the $M$ copies whose first candidate
is $a_j$.  Therefore
\[
\mathrm{WIN}_P(S_0,\{a_j\})
\ge
\frac{hq+(q-1)M}{n}
=
\frac{hq+(q-1)hq}{hq(q+1)}
=
\frac{q}{q+1}.
\]
All outside candidates of $S_0$ are of these two forms, so
$\theta_P(S_0)\ge q/(q+1)$. Since $S_0\subseteq S\subsetneq C$, the
finite-profile monotonicity fact gives
\[
        \theta^*_{k,P} \ge \theta_P(S) \ge \theta_P(S_0) \ge \frac{q}{q+1}.
\]

We now prove soundness.  Suppose the Max $\kappa$-Coverage instance is
a No instance, and let $S\in\mathcal S_k$ be arbitrary.

First consider the case $b\notin S$.  Candidate $b$ is outside the
committee.  It is covered in all $hq$ type-I rankings, since $b$ is
ranked last there.  In type-II rankings, $b$ is covered only in those
blocks whose top dummy $a_j$ belongs to $S$.  Since $|S|=k$, this
contributes at most $kM$ type-II rankings.  Therefore
\[
\mathrm{WIN}_P(S,\{b\})
\le
\frac{hq+kM}{n}
=
\frac{1+k}{q+1}.
\]
Because $k=\kappa+1=q-2$, we get
\[
\mathrm{WIN}_P(S,\{b\})
\le
\frac{q-1}{q+1}
\le
\frac{q}{q+1}-\frac{\eta}{q(q+1)},
\]
where the last inequality uses $\eta\le 1$.

It remains to consider the case $b\in S$.  Let
\[
Y:=S\cap B.
\]
Since $b\in S$ and $|S|=\kappa+1$, we have $|Y|\le \kappa$.  By
the No-instance promise, $Y$ covers at most $(1-\eta)h$ elements, so
there are at least $\eta h$ indices $i$ with $Y\cap B_i=\varnothing$.
Also, since $|S\cap D|\le \kappa$ and $|D|=q=\kappa+3$, there is a
dummy candidate $a_j\in D\setminus S$.

Fix such an $a_j$.  In the $M=hq$ type-II rankings whose top
candidate is $a_j$, no member of $S$ is ranked above $a_j$, so
$a_j$ is uncovered in all these rankings.  Moreover, for every missed
element $e_i$, the type-I ranking
\[
B_i \succ a_j \succ D\setminus\{a_j\}
      \succ B\setminus B_i \succ b
\]
also leaves $a_j$ uncovered: no member of $Y$ lies in $B_i$, all
selected dummy candidates are in $D\setminus\{a_j\}$, and $b$ is
ranked last.  Hence $a_j$ is uncovered in at least
\[
M+\eta h=hq+\eta h
\]
rankings.  Therefore
\[
\mathrm{WIN}_P(S,\{a_j\})
\le
1-\frac{hq+\eta h}{hq(q+1)}
=
\frac{q}{q+1}-\frac{\eta}{q(q+1)}.
\]
Combining the two cases, every $S\in\mathcal S_k$ in a No instance
satisfies
\[
\theta_P(S)
\le
\frac{q}{q+1}-\frac{\eta}{q(q+1)}.
\]

Now suppose, toward a contradiction, that an EPTAS exists.  Run it on
the constructed profile with accuracy parameter
\[
\gamma:=\frac{\eta}{2q^2}.
\]
Its running time is
\[
f(1/\gamma)\mathrm{poly}(m,n)
=
f(2q^2/\eta)\mathrm{poly}(m,n)
=
f'(\kappa)\mathrm{poly}(N),
\]
because $q=\kappa+3$, $\eta$ is constant, and the constructed values of
$m$ and $n$ are polynomially bounded in $N$.

In a Yes instance, the EPTAS returns a committee $\widehat S\in\mathcal
S_k$ with
\[
\theta_P(\widehat S)
\ge
(1-\gamma)\theta_{k,P}^*
\ge
(1-\gamma)\frac{q}{q+1}
=
\frac{q}{q+1}
-
\frac{\eta}{2q(q+1)}
>
\frac{q}{q+1}-\frac{\eta}{q(q+1)}.
\]
In a No instance, every committee $S\in\mathcal S_k$ has
\[
\theta_P(S)\le
\frac{q}{q+1}-\frac{\eta}{q(q+1)}.
\]
Since $\theta_P(\widehat S)$ can be computed exactly in polynomial
time from the full rankings, this distinguishes the Yes and No cases of
Gap Max $\kappa$-Coverage in $f'(\kappa)\mathrm{poly}(N)$ time,
contradicting the Gap-ETH hardness stated above.  Hence no EPTAS for
the $\theta$-optimization problem exists unless Gap-ETH fails.
\end{proof}
\subsection{Additional Proofs from Sections~\ref{subsec: theta exhaustive elicitation} and~\ref{sec:theta-vs-phi}}

\paragraph{Proof of Theorem~\ref{thm:ordinal-erm}}
For each $S \in \mathcal S_k$ and $x \in C \setminus S$, define
\[
Y_t(S,x) := \mathbf 1\{r_{\pi_t}(S) < \mathrm{rank}_{\pi_t}(x)\}.
\]
Then $Y_1(S,x), \dots, Y_T(S,x)$ are i.i.d. Bernoulli random variables with mean
\[
\mathbb E[Y_t(S,x)] = \mathrm{WIN}(S,\{x\}),
\]
where the equality uses $x \notin S$ and the strict-ranking convention from the preliminaries, so the tie term in $\mathrm{WIN}$ is zero.
Let
\[
\widehat{w}_T(S,x)
:=
\frac{1}{T}\sum_{t=1}^T Y_t(S,x).
\]
For fixed $S$ and $x$, Hoeffding's inequality gives
\[
\Pr\left[
\left|\widehat{w}_T(S,x)-\mathrm{WIN}(S,\{x\})\right|
>
\frac{\varepsilon}{2}
\right]
\le
2\exp\left(-\frac{T\varepsilon^2}{2}\right).
\]
There are exactly $(m-k)\binom{m}{k}$ pairs $(S,x)$ with
$S \in \mathcal{S}_k$ and $x \in C\setminus S$. Hence, by a union bound, if
\[
T \ge
\frac{2}{\varepsilon^2}
\log \frac{2(m-k)\binom{m}{k}}{\delta},
\]
then with probability at least $1-\delta$ the event
\[
\mathcal{E}
:=
\left\{
\left|\widehat{w}_T(S,x)-\mathrm{WIN}(S,\{x\})\right|
\le
\frac{\varepsilon}{2}
\text{ for all } S \in \mathcal{S}_k,\ x \in C\setminus S
\right\}
\]
holds.

Condition on $\mathcal{E}$. For every $S \in \mathcal{S}_k$,
\[
\left|\widehat{\theta}_T(S)-\theta(S)\right|
=
\left|
\min_{x \in C\setminus S}\widehat{w}_T(S,x)
-
\min_{x \in C\setminus S}\mathrm{WIN}(S,\{x\})
\right|
\le
\frac{\varepsilon}{2}.
\]
Let
\[
S^* \in \arg\max_{S \in \mathcal{S}_k}\theta(S).
\]
Since $\widehat{S}_{\mathrm{ERM}}$ maximizes $\widehat{\theta}_T$ over $\mathcal{S}_k$,
\[
\theta(\widehat{S}_{\mathrm{ERM}})
\ge
\widehat{\theta}_T(\widehat{S}_{\mathrm{ERM}})-\frac{\varepsilon}{2}
\ge
\widehat{\theta}_T(S^*)-\frac{\varepsilon}{2}
\ge
\theta(S^*)-\varepsilon
=
\theta_k^*-\varepsilon.
\]

It remains to account for the query cost. On each sampled task, the pairwise comparison oracle gives
a comparison oracle for the unknown strict ranking $\pi_t$. A comparison-sorting algorithm therefore
recovers $\pi_t$ using $O(m\log m)$ pairwise comparisons. Repeating this for all $T$ sampled
rankings uses $O(Tm\log m)$ comparisons, which gives the stated bound after substituting the lower
bound on $T$.

\paragraph{Proof of Theorem~\ref{thm:ordinal-lower-main}, part 1}
Let $S=\{s\}$, and write $R:=C\setminus\{s\}$. On each sampled instance, draw bits $(B_x)_{x\in R}$. The ranking places all rivals with $B_x=0$ above $s$, then $s$, then all rivals with $B_x=1$, with uniformly random tie-breaking inside the two rival blocks. Thus $s$ beats $x$ exactly when $B_x=1$.

Under the null distribution $P_0$, all bits are independent $\mathrm{Bernoulli}(1/2)$, so
\[
    \theta(\{s\})=\frac12.
\]
For each $i\in R$, define $P_i$ by changing only
\[
    B_i\sim \mathrm{Bernoulli}\left(\frac12-4\varepsilon\right),
\]
leaving all other bits unbiased. Under $P_i$,
\[
    \theta(\{s\})=\frac12-4\varepsilon.
\]
Let $\widehat\theta$ be the estimate output by the algorithm. Hence an $\varepsilon$-accurate audit must distinguish $P_0$ from every $P_i$ with constant probability. Indeed, the event
\[
\mathcal G_i := \left\{\widehat\theta \ge \frac12 - 2\varepsilon\right\}
\]
has probability at least $2/3$ under $P_0$ and at most $1/3$ under $P_i$.
Let $\mathbb P_0$ and $\mathbb P_i$ denote the full transcript laws of the algorithm under $P_0$ and $P_i$, including its internal randomness and final estimate. Then
\[
\|\mathbb P_0-\mathbb P_i\|_{\mathrm{TV}}
\ge \mathbb P_0(\mathcal G_i)-\mathbb P_i(\mathcal G_i)
\ge \frac13.
\]
Pinsker's inequality gives
\[
\mathrm{KL}(\mathbb P_0\|\mathbb P_i) \ge c
\]
for a universal constant $c>0$. 
Let $M_i$ be the number of sampled instances on which the algorithm makes at least one query involving expert $i$. To upper bound the information about $P_0$ versus $P_i$, reveal $B_i$ on a sampled instance at the moment of the first query involving $i$ on that instance. This can only increase the transcript KL. After $B_i$ is revealed, the remaining bits and the random tie-breaking have the same conditional law under $P_0$ and $P_i$. Thus queries not involving $i$, and all later queries on an instance after this reveal, contribute zero conditional KL. By the adaptive chain rule for KL,
\[
\mathrm{KL}(\mathbb{P}_0\|\mathbb{P}_i)
\le
\mathbb{E}_0[M_i]\,
\mathrm{KL}\left(
\mathrm{Bernoulli}\left(\frac12\right)
\,\middle\|\,
\mathrm{Bernoulli}\left(\frac12 - 4\varepsilon\right)
\right)
\le C\varepsilon^2 \mathbb{E}_0[M_i],
\]
where the last inequality uses $\varepsilon \le 1/16$.
Therefore
\[
    \mathbb E_0[M_i]=\Omega(1/\varepsilon^2)
\]
for every $i\in R$.

Let $N_i$ be the total number of pairwise queries involving rival $i$, and let $Q$ be the total number of pairwise queries. Then $N_i\ge M_i$, and each pairwise query involves at most two rivals, so
\[
    \sum_{i\in R} N_i \le 2Q.
\]
Thus
\[
    \mathbb E_0[Q]
    \ge
    \frac12\sum_{i\in R}\mathbb E_0[N_i]
    =
    \Omega\left(\frac{m-1}{\varepsilon^2}\right).
\]
Since $P_0$ is an admissible ranking distribution, the claimed worst-case lower bound follows.

\paragraph{Proof of Theorem~\ref{thm:ordinal-lower-main}, part 2}
For each $i\in C$, define a ranking distribution $P_i$ as follows. On each sampled instance, draw independent bits $B_1,\dots,B_m$, with
\[
    B_i\sim \mathrm{Bernoulli}\left(\frac12+4\varepsilon\right),
    \quad
    B_j\sim \mathrm{Bernoulli}\left(\frac12\right)
    \quad(j\neq i).
\]
The ranking first lists all candidates with $B_j=1$ in a uniformly random order, followed by all candidates with $B_j=0$ in a uniformly random order. Let $P_0$ denote the same construction with all bits unbiased.

Under $P_i$, for every $j\neq i$,
\[
\begin{aligned}
    \mathbb P[i\succ_\pi j]
    &=
    \mathbb P[B_i=1,B_j=0]
    +
    \frac12\mathbb P[B_i=B_j]  \\
    &=
    \frac12+2\varepsilon.
\end{aligned}
\]
Similarly,
\[
    \mathbb P[j\succ_\pi i]=\frac12-2\varepsilon,
\]
and any two nonspecial candidates beat each other with probability $1/2$. 
Hence
\[
\theta(\{i\}) = \frac12 + 2\varepsilon, \quad \theta(\{j\}) \le \frac12 - 2\varepsilon \quad (j \ne i).
\]
Thus $\theta_1^*=\theta(\{i\})$, and every $j\ne i$ satisfies
\[
\theta(\{j\}) < \theta_1^*-\varepsilon.
\]
Therefore any algorithm satisfying the theorem's guarantee must output $i$ with probability at least
$2/3$ under $P_i$.

Let $\mathbb P_i$ and $\mathbb P_0$ be the laws of the algorithm's full transcript under $P_i$ and $P_0$, including its internal randomness and final output. Let
\[
    A_i:=\{\widehat c=i\}.
\]
Since $\sum_i \mathbb P_0(A_i)=1$, at least $m-1$ indices satisfy $\mathbb P_0(A_i)\le 1/2$. Fix such an $i$. Then $\mathbb P_i(A_i)\ge 2/3$, so
\[
    \|\mathbb P_i-\mathbb P_0\|_{\mathrm{TV}}\ge \frac16.
\]
Pinsker's inequality gives
\[
    \mathrm{KL}(\mathbb P_0\Vert \mathbb P_i)\ge \frac1{18}.
\]

Let $M_i$ be the number of sampled instances on which the algorithm makes at least one query involving expert $i$. The enhanced-oracle argument from part 1, with the sign of the bias reversed, gives
\[
\mathrm{KL}(\mathbb{P}_0\|\mathbb{P}_i)
\le
\mathbb{E}_0[M_i]\,
\mathrm{KL}\left(
\mathrm{Bernoulli}\left(\frac12\right)
\,\middle\|\,
\mathrm{Bernoulli}\left(\frac12 + 4\varepsilon\right)
\right)
\le C\varepsilon^2 \mathbb{E}_0[M_i],
\]
where the last inequality uses $\varepsilon \le 1/16$.
Thus $\mathbb E_0[M_i] = \Omega(1/\varepsilon^2)$ for every index $i$ with $\mathbb P_0(A_i)\le 1/2$, and there are at least
$m-1$ such indices. Let $N_i$ be the total number of pairwise queries
involving expert $i$. Since $N_i \ge M_i$ pathwise and each query involves
at most two experts, $\sum_i N_i \le 2Q$, where $Q$ is the total number of pairwise queries. Therefore, under $P_0$,
\[
\mathbb E_0[Q]
\ge \frac12 \sum_i \mathbb E_0[N_i]
\ge \frac12 \sum_{i:\,\mathbb P_0(A_i)\le 1/2} \mathbb E_0[N_i]
= \Omega\left(\frac{m}{\varepsilon^2}\right).
\]
Since $P_0$ is an admissible ranking distribution, the worst-case expected query complexity is $\Omega\left(\frac{m}{\varepsilon^2}\right)$.

\paragraph{Proof of Proposition~\ref{prop:theta-not-submodular}}
Monotonicity follows because adding experts can only improve the best member of the committee and can only remove outside-rival constraints.

For non-submodularity, let $C=\{a,b,c,d\}$. Put probability $2/3$ on
\[
    b\succ a\succ d\succ c
\]
and probability $1/3$ on
\[
    c\succ a\succ d\succ b.
\]
Then
\[
    \theta(\{a\})=\theta(\{a,c\})=\frac13,
    \quad
    \theta(\{a,b\})=\frac23,
    \quad
    \theta(\{a,b,c\})=1.
\]
Thus adding $c$ to $\{a\}$ gives zero marginal gain, while adding $c$ to the superset $\{a,b\}$ gives marginal gain $1/3$. This violates diminishing returns.

\paragraph{Proof of Lemma~\ref{lem:theta-phi}}
Fix $(\pi,x)$. The map
\[
    S\mapsto \mathbf 1\{S\cap P_\pi(x)\neq \varnothing\}
\]
is a coverage function, hence it is normalized, monotone, and submodular. Since $\Phi_\lambda$ is a nonnegative weighted average of these functions, it has the same three properties.

For the representation of $\theta$, observe that for every proper committee $S$,
\[
    \min_{\lambda\in\Delta(C)}\Phi_\lambda(S)
    =
    \min_{\lambda\in\Delta(C)}
    \sum_{x\in C}\lambda_x g_x(S)
    =
    \min_{x\in C}g_x(S),
\]
because the minimum of a linear function over the simplex is attained at an extreme point. If $x\in S$, then $g_x(S)=1$. If $x\notin S$, then $g_x(S)=\mathrm{WIN}(S,\{x\})$. Since $S\subsetneq C$, there is at least one outside candidate, and therefore
\[
    \min_{x\in C}g_x(S)
    =
    \min_{x\in C\setminus S}\mathrm{WIN}(S,\{x\})
    =
    \theta(S).
\]

\subsection{Failure-conditioned weighted ordinal greedy}
\label{app:weighted-ordinal-oracle}

\begin{algorithm}[H]
\caption{\textsc{Ordinal-Fail-Cond-Elim}$(S,A,\lambda,\eta,\delta)$}
\label{alg:ordinal-fail-cond-elim}
\begin{algorithmic}[1]
\Require Committee $S$ with $\rho_\lambda(S)>0$; nonempty candidate set
$A\subseteq C\setminus S$; rival weights $\lambda\in\Delta(C)$; accuracy
$\eta\in(0,1]$; confidence $\delta\in(0,1)$.
\State $A_1\gets A$.
\For{$r=1,2,\dots$}
    \State Draw $(\pi_r,x_r)\sim P_D\times\lambda$ conditioned on $S\cap P_{\pi_r}(x_r)=\varnothing$.
    \Comment{Implemented by rejection sampling.}
    \ForAll{$c\in A_r$}
        \If{$c=x_r$}
            \State $X_{r,c}\gets 1$.
        \Else
            \State $X_{r,c}\gets \mathrm{Query}(c,x_r;\pi_r)$.
        \EndIf
        \State $\widehat q_r(c)\gets r^{-1}\sum_{s=1}^r X_{s,c}$.
    \EndFor
    \State $\mathrm{rad}_r\gets \sqrt{\frac{1}{2r}\log\frac{4|A|r^2}{\delta}}$.
    \State $c^*_r\in\arg\max_{c\in A_r}\widehat q_r(c)$.
    \State
    \[
        A_{r+1}
        \gets
        \left\{
            c\in A_r:
            \widehat q_r(c)+\mathrm{rad}_r
            \ge
            \widehat q_r(c^*_r)-\mathrm{rad}_r-\eta
        \right\}.
    \]
    \If{$\mathrm{rad}_r\le \eta/4$ or $|A_{r+1}|=1$}
        \State \Return any $\widehat c\in\arg\max_{c\in A_{r+1}}\widehat q_r(c)$.
    \EndIf
\EndFor
\end{algorithmic}
\end{algorithm}

\begin{algorithm}[H]
\caption{\textsc{Weighted-Ordinal-Fail-Greedy}$(k,\lambda,\varepsilon,\delta)$}
\label{alg:weighted-ordinal-fail-greedy}
\begin{algorithmic}[1]
\State $S_0\gets \varnothing$.
\For{$i=0,1,\dots,k-1$}
    \State Draw
            \[
                L=\left\lceil \frac{128k^2}{\varepsilon^2}\log\frac{4k}{\delta}\right\rceil
            \]
            tentative pairs $(\pi,x)\sim P_D\times\lambda$, and test whether
        $S_i\cap P_\pi(x)=\varnothing$ by scanning $s\in S_i$: if $s=x$, declare
        the pair covered without an oracle call; otherwise query
        $\mathrm{Query}(s,x;\pi)$, declaring the pair covered as soon as the
        answer is $1$. If no scan declares the pair covered, record a failure.
    \State Let $\widehat\rho_i$ be the empirical failure frequency and set
    \[
        \overline\rho_i
        :=
        \min\left\{
            1,\,
            \widehat\rho_i+\frac{\varepsilon}{16k}
        \right\}.
    \]
    \If{$\overline\rho_i\le \varepsilon/(4k)$}
        \State Extend $S_i$ arbitrarily to some $S\in\mathcal S_k$, and return $S$.
    \EndIf
    \State $\eta_i\gets \varepsilon/(4k\overline\rho_i)$, $\delta_i\gets \delta/(2k)$.
    \State
    \[
        \widehat c_{i+1}
        \gets
        \textsc{Ordinal-Fail-Cond-Elim}
        (S_i,C\setminus S_i,\lambda,\eta_i,\delta_i).
    \]
    \State $S_{i+1}\gets S_i\cup\{\widehat c_{i+1}\}$.
\EndFor
\State \Return $S_k$.
\end{algorithmic}
\end{algorithm}

\paragraph{Proof of Lemma~\ref{lem:ordinal-marginal-identity}}
The pair $(\pi,x)$, with $\pi\sim P_D$ and
$x\sim\lambda$, is newly covered by adding $c$ exactly when
$S\cap P_\pi(x)=\varnothing$ and $c\in P_\pi(x)$. Therefore
\[
\Phi_\lambda(S\cup\{c\})-\Phi_\lambda(S)
=
\Pr_{\pi\sim P_D,\,x\sim\lambda}
\left[S\cap P_\pi(x)=\varnothing,\ c\in P_\pi(x)\right].
\]
If $\rho_\lambda(S)>0$, this probability equals
$\rho_\lambda(S)q_\lambda(c\mid S)$ by the definition of $q_\lambda$. If
$\rho_\lambda(S)=0$, the probability is $0$, which is also
$\rho_\lambda(S)q_\lambda(c\mid S)$ under the convention in Section~5.4.

\paragraph{Proof of Theorem~\ref{thm:weighted-ordinal-oracle} and the detailed form}
Fix $k\in\{1,\dots,m\}$,
$\lambda\in\Delta(C)$, and $\varepsilon,\delta\in(0,1)$, and let
\[
\Phi^*_{\lambda,k}:=\max_{S\in\mathcal S_k}\Phi_\lambda(S).
\]
Algorithm~\ref{alg:weighted-ordinal-fail-greedy} outputs a committee $\widehat S\in\mathcal S_k$ satisfying, with
probability at least $1-\delta$,
\[
\Phi_\lambda(\widehat S)\ge (1-1/e)\Phi^*_{\lambda,k}-\varepsilon .
\]
Let $S_i$ denote the committee at the start of iteration $i$ of Algorithm~\ref{alg:weighted-ordinal-fail-greedy}, whenever that iteration is reached, and write
\[
\rho_i^\lambda := \rho_\lambda(S_i)=1-\Phi_\lambda(S_i).
\]
Let $\tau\le k$ be the number of iterations that pass the stopping test and call Algorithm~\ref{alg:ordinal-fail-cond-elim}. For $i<\tau$ and $c\in C\setminus S_i$, define the true marginal gap
\[
\Delta_i^\lambda(c)
:=
\max_{a\in C\setminus S_i}
\bigl(\Phi_\lambda(S_i\cup\{a\})-\Phi_\lambda(S_i)\bigr)
-
\bigl(\Phi_\lambda(S_i\cup\{c\})-\Phi_\lambda(S_i)\bigr).
\]
If $Q_{\rm cand}$ denotes the number of candidate-versus-rival comparisons made on accepted failed pairs, then on the same event,
\[
Q_{\rm cand}
=
O\left(
\sum_{i=0}^{\tau-1}
\sum_{c\in C\setminus S_i}
\frac{(\rho_i^\lambda+\varepsilon/k)^2
\log(emk^2/(\delta\varepsilon))}
{\max\{\Delta_i^\lambda(c),\varepsilon/k\}^2}
\right).
\]
The rejection-sampling and committee testing overheads are described after the proof.

For each reached iteration $i$, conditional on the history up to the start of that iteration, the failure indicators of the $L$ tentative pairs used to form $\widehat\rho_i$ are i.i.d. Bernoulli with mean $\rho_i^\lambda$. Since
\[
\bar\rho_i=\min\left\{1,\widehat\rho_i+\frac{\varepsilon}{16k}\right\},
\]
the choice
\[
L=\left\lceil \frac{128k^2}{\varepsilon^2}\log\frac{4k}{\delta}\right\rceil
\]
and Hoeffding's inequality give
\[
\Pr\left(
\left.
\left|\widehat\rho_i-\rho_i^\lambda\right|>\frac{\varepsilon}{16k}
\,\right|\,\text{history}
\right)
\le \frac{\delta}{2k}.
\]
Thus, by an adaptive union bound over the at most $k$ reached iterations, with probability at least $1-\delta/2$,
\[
\forall\text{ reached }i,\quad
\rho_i^\lambda\le \bar\rho_i\le \rho_i^\lambda+\frac{\varepsilon}{8k}.
\]

Now consider a call to Algorithm~\ref{alg:ordinal-fail-cond-elim} at iteration $i$, and condition on the history before the call. The accepted pairs $(\pi_r,x_r)$ are i.i.d. from $P_D\times\lambda$ conditioned on
\[
S_i\cap P_{\pi_r}(x_r)=\varnothing.
\]
For each accepted pair set
\[
X_{r,c}:=\mathbf 1\{c\in P_{\pi_r}(x_r)\}.
\]
This is exactly the value obtained on lines 5--9 of Algorithm~\ref{alg:ordinal-fail-cond-elim}. The
conditioning implies $x_r\notin S_i$. If $c=x_r$, then
$c\in P_{\pi_r}(x_r)$, so the value is $1$ without an oracle call. If
$c\ne x_r$, then, since the ranking is strict,
\[
\mathbf 1\{c\in P_{\pi_r}(x_r)\}
=
\mathbf 1\{c\succ_{\pi_r} x_r\}
=
\mathrm{Query}(c,x_r;\pi_r).
\]
For each fixed $c$, the variables $(X_{r,c})_{r\ge 1}$ are i.i.d.
Bernoulli with mean $q_\lambda(c\mid S_i)$, and the active sets in
Algorithm~\ref{alg:ordinal-fail-cond-elim} are nested. These are exactly the facts used in the proof of
Theorem~\ref{thm:failcond-elim}, so that proof applies with $q(c\mid S)$ replaced by
$q_\lambda(c\mid S_i)$. Moreover, the number of actual
candidate-versus-rival comparisons is no larger than the number of active
observations counted there, since the case $c=x_r$ uses no oracle call.

Since $\delta_i=\delta/(2k)$, an adaptive union bound over the at most $k$ calls implies that, with probability at least $1-\delta/2$, every call returns an $\eta_i$-optimal rescue candidate and satisfies the corresponding candidate-query bound. Work on the intersection of this event and the failure-rate concentration event above.

If Algorithm~\ref{alg:weighted-ordinal-fail-greedy} stops at step $i$, then
\[
    \rho_i^\lambda
    \le
    \overline\rho_i
    \le
    \frac{\varepsilon}{4k}.
\]
Thus
\[
    \Phi_\lambda(S_i)
    =
    1-\rho_i^\lambda
    \ge
    1-\frac{\varepsilon}{4k}.
\]
After arbitrary padding to size $k$, Lemma~\ref{lem:theta-phi} preserves this lower bound
by monotonicity. Since $\Phi^*_{\lambda,k}\le 1$ and
$\varepsilon/(4k)\le \varepsilon$, the padded committee $S$ satisfies
\[
\Phi_\lambda(S)\ge 1-\varepsilon
\ge (1-1/e)\Phi^*_{\lambda,k}-\varepsilon .
\]

Otherwise, the algorithm executes a greedy step. 
On the confidence event, any executed step satisfies
\[
\rho_i^\lambda\ge \bar\rho_i-\frac{\varepsilon}{8k}>\frac{\varepsilon}{8k}>0,
\]
and $\eta_i=\varepsilon/(4k\bar\rho_i)\in(0,1]$. Thus the call to Algorithm~\ref{alg:ordinal-fail-cond-elim} is
well defined.

Conditional on $S_i\cap P_\pi(x)=\varnothing$, the observations queried by Algorithm~\ref{alg:ordinal-fail-cond-elim} are Bernoulli with means $q_\lambda(c\mid S_i)$. Therefore the returned candidate $\widehat c_{i+1}$ satisfies
\[
    q_\lambda(\widehat c_{i+1}\mid S_i)
    \ge
    \max_{c\in C\setminus S_i}q_\lambda(c\mid S_i)-\eta_i.
\]
By Lemma~\ref{lem:ordinal-marginal-identity},
\[
    \Phi_\lambda(S_i\cup\{c\})-\Phi_\lambda(S_i)
    =
    \rho_i^\lambda q_\lambda(c\mid S_i).
\]
Since $\rho_i^\lambda\le \overline\rho_i$, the marginal loss from using $\widehat c_{i+1}$ instead of an exact best marginal element is at most
\[
    \rho_i^\lambda\eta_i
    \le
    \overline\rho_i\cdot
    \frac{\varepsilon}{4k\overline\rho_i}
    =
    \frac{\varepsilon}{4k}.
\]
By Lemma~\ref{lem:theta-phi}, $\Phi_\lambda$ is monotone and submodular. Hence, for every executed step $i$,
\[
\Phi_\lambda(S_{i+1})-\Phi_\lambda(S_i)
\ge
\max_{c\in C\setminus S_i}
\bigl(\Phi_\lambda(S_i\cup\{c\})-\Phi_\lambda(S_i)\bigr)
-\frac{\varepsilon}{4k}.
\]
If no early stopping occurs, the standard additive-error greedy recursion gives
\[
\Phi_\lambda(S_k)
\ge
\left(1-\left(1-\frac1k\right)^k\right)\Phi^*_{\lambda,k}
-\frac{\varepsilon}{4}
\ge
(1-1/e)\Phi^*_{\lambda,k}-\varepsilon.
\]
If early stopping occurs, the preceding stopping argument already proves the claimed guarantee.

It remains to prove the query bound. Fix an executed step $i$ and
$c\in C\setminus S_i$. The rescue rate gap appearing in Theorem~\ref{thm:failcond-elim} is
\[
\max_{a\in C\setminus S_i} q_\lambda(a\mid S_i)-q_\lambda(c\mid S_i)
=
\frac{\Delta_i^\lambda(c)}{\rho_i^\lambda},
\]
where $\rho_i^\lambda>0$ by the previous paragraph. Algorithm~\ref{alg:ordinal-fail-cond-elim} separates
gaps only down to scale
\[
\eta_i=\frac{\varepsilon}{4k\bar\rho_i}.
\]
On the confidence event,
\[
\bar\rho_i \le \rho_i^\lambda+\frac{\varepsilon}{8k}
\le \rho_i^\lambda+\frac{\varepsilon}{k}.
\]
Therefore
\[
\eta_i
=
\frac{\varepsilon}{4k\bar\rho_i}
\ge
\frac{\varepsilon/k}{4(\rho_i^\lambda+\varepsilon/k)},
\]
and, since $\rho_i^\lambda>0$,
\[
\frac{\Delta_i^\lambda(c)}{\rho_i^\lambda}
\ge
\frac{\Delta_i^\lambda(c)}{\rho_i^\lambda+\varepsilon/k}.
\]
Consequently,
\[
\max\left\{\frac{\Delta_i^\lambda(c)}{\rho_i^\lambda},\eta_i\right\}
\ge
\frac{1}{4}\cdot
\frac{\max\{\Delta_i^\lambda(c),\varepsilon/k\}}
{\rho_i^\lambda+\varepsilon/k}.
\]
Also, since $|C\setminus S_i|\le m$, $\delta_i=\delta/(2k)$, and
$\eta_i\ge \varepsilon/(4k)$, the logarithmic factor in Theorem~\ref{thm:failcond-elim} is
$O(\log(emk^2/(\delta\varepsilon)))$. Applying the candidate-query bound from
Theorem~\ref{thm:failcond-elim} and summing over candidates and executed greedy steps gives
\[
Q_{\mathrm{cand}}
=
O\left(
\sum_{i=0}^{\tau-1}\sum_{c\in C\setminus S_i}
\frac{
(\rho_i^\lambda+\varepsilon/k)^2\log(emk^2/(\delta\varepsilon))
}{
\max\{\Delta_i^\lambda(c),\varepsilon/k\}^2
}
\right).
\]

\paragraph{Committee testing overhead.}
The bound in Theorem~\ref{thm:weighted-ordinal-oracle} counts only candidate-versus-rival comparisons on accepted failed pairs. To generate such pairs, the algorithm uses rejection sampling. 
If $R_i$ accepted failures are used at step $i$, then, conditional on
the accepted-pair transcript, rejection sampling uses
$R_i/\rho_i^\lambda$ tentative pairs in expectation whenever
$\rho_i^\lambda>0$. Each tentative pair costs at most $|S_i|$ pairwise
comparisons to test whether $S_i\cap P_\pi(x)=\varnothing$. The failure-rate estimation step uses $O(
        |S_i|\frac{k^2}{\varepsilon^2}
        \log (\frac{k}/{\delta})
    )$ committee testing comparisons at step $i$.

\subsection{Auditing and finite-family learning}
\label{app:theta-audit}

\paragraph{Proof of Theorem~\ref{thm:theta-audit-pool-main} and the detailed form}
Fix $\varepsilon,\delta\in(0,1)$. Given rankings
$\pi_1,\dots,\pi_T$, define, for any nonempty proper committee
$A\subsetneq C$ and any $x\in C\setminus A$,
\[
\widehat w_T(A,x)
:=
\frac1T\sum_{t=1}^T
\mathbf{1}\{r_{\pi_t}(A)<\mathrm{rank}_{\pi_t}(x)\},
\quad
\widehat\theta_T(A):=\min_{x\in C\setminus A}\widehat w_T(A,x).
\]

First fix a nonempty proper committee $S\subsetneq C$. If $T\ge \frac{1}{2\varepsilon^2}\log\frac{2(m-|S|)}{\delta}$, then $T(m-1)$ pairwise comparisons suffice to output
$\widehat\theta_T(S)$ such that
\[
|\widehat\theta_T(S)-\theta(S)|\le \varepsilon
\]
with probability at least $1-\delta$.

More generally, let $\mathcal F\subseteq 2^C$ be any nonempty finite
family of nonempty proper committees, and define
\[
N_{\mathcal F}:=\sum_{A\in\mathcal F}(m-|A|).
\]
If $T\ge \frac{2}{\varepsilon^2}\log\frac{2N_{\mathcal F}}{\delta}$, then, after recovering $T$ full rankings, any $\widehat S\in\arg\max_{A\in\mathcal F}\widehat\theta_T(A)$ satisfies
\[
\theta(\widehat S)\ge \max_{A\in\mathcal F}\theta(A)-\varepsilon
\]
with probability at least $1-\delta$. Recovering the rankings by
comparison sorting uses $O(Tm\log m)$ pairwise comparisons.

For any nonempty proper committee $S$ and $x\in C\setminus S$, define
\[
\widehat w_T(S,x):=\frac1T\sum_{t=1}^T
\mathbf 1\{r_{\pi_t}(S)<\mathrm{rank}_{\pi_t}(x)\},
\quad
\widehat\theta_T(S):=\min_{x\in C\setminus S}\widehat w_T(S,x).
\]
For the fixed-committee audit below, the same quantity is observed without
recovering the full ranking: if $s_t=\top_{\pi_t}(S)$, then
\[
\mathbf 1\{r_{\pi_t}(S)<\mathrm{rank}_{\pi_t}(x)\}
=
\mathbf 1\{s_t\succ_{\pi_t}x\}.
\]

For the fixed-committee audit, draw $\pi_1,\dots,\pi_T\sim P_D$. For each $t$, find the best member $s_t=\top_{\pi_t}(S)$ using $|S|-1$ pairwise comparisons, and then compare $s_t$ against every $x\in C\setminus S$. For a fixed outside rival $x$, the resulting observation is Bernoulli with mean $\mathrm{WIN}(S,\{x\})$. Hoeffding's inequality and a union bound over $m-|S|$ outside rivals imply
\[
    \max_{x\in C\setminus S}
    \left|
        \widehat w_T(S,x)-\mathrm{WIN}(S,\{x\})
    \right|
    \le \varepsilon
\]
with probability at least $1-\delta$ under the stated sample size. Taking minima over $x\in C\setminus S$ changes the estimate by at most the same uniform error. Each sample uses $|S|-1$ comparisons to find $s_t$ and $m-|S|$ comparisons against outside rivals, for a total of $m-1$ comparisons.

For the finite-family statement, for every $S\in\mathcal F$ and $x\in C\setminus S$, the empirical win rate is an average of $T$ independent Bernoulli variables. Hoeffding's inequality and a union bound over $N_{\mathcal F}
    =
    \sum_{S\in\mathcal F}(m-|S|)$
constraints imply that all these win rates are within $\varepsilon/2$ of their expectations with probability at least $1-\delta$. On this event, every $\widehat\theta_T(S)$ is within $\varepsilon/2$ of $\theta(S)$, and empirical maximization over $\mathcal F$ returns $\widehat S$ with $\theta(\widehat S)
    \ge
    \max_{S\in\mathcal F}\theta(S)-\varepsilon$.
Recovering each full ranking by comparison sorting costs $O(m\log m)$ pairwise comparisons, giving $O(Tm\log m)$ total comparisons.

\begin{theorem}[Gap-adaptive audit of a proposed committee]
\label{thm:gap-adaptive-audit}
Fix $\varepsilon,\delta\in(0,1)$ and a nonempty proper committee
$S\subsetneq C$, and let $M:=m-|S|$. For $x\in C\setminus S$, let $\mu_x:=\mathrm{WIN}(S,\{x\})$, $\theta(S)=\min_{x\in C\setminus S}\mu_x$, and $\Gamma_x:=\mu_x-\theta(S)$.
There is an adaptive pairwise comparison audit which maintains a confidence
interval $[\underline\theta_t,\overline\theta_t]$ for $\theta(S)$ on an
event of probability at least $1-\delta$. Stopping the audit when $\overline\theta_t-\underline\theta_t\le 2\varepsilon$ 
and returning $\widehat\theta(S)
        :=\frac{\underline\theta_t+\overline\theta_t}{2}$ 
gives $|\widehat\theta(S)-\theta(S)|\le \varepsilon$
on this event. On the same event, the number of pairwise comparisons made
up to this stopping time is
\[
O\left(
        (|S|-1)\frac{\log(M/(\delta\varepsilon))}{\varepsilon^2}
        +
        \sum_{x\in C\setminus S}
        \frac{\log(M/(\delta\max\{\Gamma_x,\varepsilon\}))}
             {\max\{\Gamma_x,\varepsilon\}^2}
\right).
\]
\end{theorem}

\begin{proof}
Let $A_1:=C\setminus S$. The audit maintains an active set $A_t$ of outside rivals that may still attain the minimum. In round $t=1,2,\dots$, draw a fresh ranking $\pi_t\sim P_D$, find $s_t:=\top_{\pi_t}(S)$ using $|S|-1$ pairwise comparisons. For the analysis, define for every
outside rival $x\in C\setminus S$, $Y_{t,x}:=\mathbf 1\{s_t\succ_{\pi_t}x\}$.
The algorithm queries this bit only for the currently active rivals
$x\in A_t$. For each fixed $x$, the sequence $(Y_{t,x})_{t\ge 1}$ is
i.i.d. Bernoulli with mean $\mu_x=\mathrm{WIN}(S,\{x\})$, by the definition of $\mathrm{WIN}$ and the strict-ranking convention from the preliminaries. Let $\widehat\mu_{x,n}:=\frac1n\sum_{s=1}^n Y_{s,x}$, and let $n_x(t)$ denote the number of observations collected from $x$ up to
round $t$. Since the active sets only shrink, if $x$ is still active after
round $t$, then the empirical mean of the observations actually collected
from $x$ is $\widehat\mu_{x,n_x(t)}$.

For $n\ge 1$, define the anytime radius
\[
r(n):=\sqrt{\frac{\log(4Mn^2/\delta)}{2n}},
\quad M:=m-|S|.
\]
By Hoeffding's inequality and a union bound over $x\in C\setminus S$ and
$n\ge 1$, with probability at least $1-\delta$, the event $\mathcal E:=
\left\{
\forall x\in C\setminus S,\ \forall n\ge 1:
|\widehat\mu_{x,n}-\mu_x|\le r(n)
\right\}$ holds.

Work on $\mathcal E$. For each active rival $x$, set
\[
L_x(t):=\widehat\mu_{x,n_x(t)}-r(n_x(t)),
\quad
U_x(t):=\widehat\mu_{x,n_x(t)}+r(n_x(t)).
\]
Let
\[
    \underline\theta_t:=\min_{x\in A_t}L_x(t),
    \quad
    \overline\theta_t:=\min_{x\in A_t}U_x(t).
\]
The algorithm eliminates an active rival $x$ whenever $L_x(t)>\overline\theta_t$.
This elimination rule is safe on $\mathcal E$: if $y\in A_t$ attains $\overline\theta_t=U_y(t)$, then $\mu_y\le U_y(t)=\overline\theta_t < L_x(t)\le \mu_x$, so $x$ cannot be a minimizer of $\mu_z$. Therefore no true minimizer is ever eliminated, and
\[
    \theta(S)=\min_{x\in C\setminus S}\mu_x
             =\min_{x\in A_t}\mu_x
             \in [\underline\theta_t,\overline\theta_t]
\]
for every round $t$.

The algorithm stops once $\overline\theta_t-\underline\theta_t\le 2\varepsilon$ and returns the midpoint $\widehat\theta(S):=\frac{\underline\theta_t+\overline\theta_t}{2}$.
On $\mathcal E$, this gives $|\widehat\theta(S)-\theta(S)|\le \varepsilon$.

It remains to bound the number of comparisons. Fix $x\in C\setminus S$, and let $\Gamma_x:=\mu_x-\theta(S),
    \quad
    h_x:=\max\{\Gamma_x,\varepsilon\}$.
There is a universal constant $K$ such that after $K\frac{\log(M/(\delta h_x))}{h_x^2}$ observations of rival $x$, its confidence radius is at most $h_x/8$.
Let $z$ be any true minimizer, so $\mu_z=\theta(S)$. By the safety
argument above, $z$ is never eliminated.

If $\Gamma_x\ge\varepsilon$, then once both $x$ and $z$ have radius at most
$\Gamma_x/8$,
\[
L_x(t)\ge \mu_x-\Gamma_x/4
= \theta(S)+3\Gamma_x/4
> \theta(S)+\Gamma_x/4
\ge U_z(t)
\ge \overline\theta_t,
\]
so $x$ is eliminated.

If $\Gamma_x<\varepsilon$, then $x$ can only be queried until the audit
stops. After $O ({\log(M/(\delta\varepsilon))}/{\varepsilon^2} )$ fresh rankings, every active rival has radius at most $\varepsilon/8$,
because active rivals are queried in every round. Since a true minimizer
$z$ is still active,
\[
\underline\theta_t\ge \theta(S)-\varepsilon/4,
\quad
\overline\theta_t\le U_z(t)\le \theta(S)+\varepsilon/4.
\]
Hence $\overline\theta_t-\underline\theta_t\le 2\varepsilon$, so the audit
stops. Therefore each rival $x$ is queried at most
\[
O\left(
\frac{\log(M/(\delta\max\{\Gamma_x,\varepsilon\}))}
{\max\{\Gamma_x,\varepsilon\}^2}
\right)
\]
times.

The stopping argument above also bounds the total number of fresh rankings
on which $\top_\pi(S)$ must be found by $O({\log(M/(\delta\varepsilon))}/{\varepsilon^2})$.
Each such ranking costs $|S|-1$ comparisons to find $\top_\pi(S)$. Summing the rival-comparison costs over $x\in C\setminus S$ gives the claimed bound.
\end{proof}

\begin{theorem}[Active finite-family \texorpdfstring{$\theta$}{theta}-learning]
\label{thm:active-finite-family-theta-learning}
    Fix $\varepsilon,\delta\in(0,1)$. Let $\mathcal F\subseteq 2^C$ be a nonempty
    finite family of nonempty proper committees, and set $N_{\mathcal F}:=\sum_{S\in\mathcal F}(m-|S|)$, $\theta^*_{\mathcal F}:=\max_{S\in\mathcal F}\theta(S)$.
    For $S\in\mathcal F$, define the committee gap $\Lambda_S := \theta^*_{\mathcal F}-\theta(S)$, and for $x\in C\setminus S$, define the rival gap $\Gamma_{S,x} := \mathrm{WIN}(S,\{x\})-\theta(S)$.
    There is an adaptive pairwise query algorithm which returns $\widehat S\in\mathcal F$ satisfying $\theta(\widehat S)\ge \theta^*_{\mathcal F}-\varepsilon$ with probability at least $1-\delta$. 
    On the same event, its query complexity is
    \[
    O\left(
        \sum_{S\in\mathcal F} (|S|-1)\frac{\log(N_{\mathcal F}/(\delta\alpha_S))}{\alpha_S^2}
        +
        \sum_{S\in\mathcal F}\sum_{x\in C\setminus S}
        \frac{\log(N_{\mathcal F}/(\delta H_{S,x}))}{H_{S,x}^2}
    \right),
    \]
    where $\alpha_S := \max\{\varepsilon,\Lambda_S\}$ and $H_{S,x} := \max\{\varepsilon,\Lambda_S,\Gamma_{S,x}\}$.
\end{theorem}

\begin{proof}
Run a phased racing algorithm over committees, but keep the audit states
persistent across phases. For each $S \in \mathcal F$, initialize one copy of the audit process from Theorem~\ref{thm:gap-adaptive-audit} with confidence $\delta_S := \frac{\delta}{N_F}$.
Since $N_{\mathcal F}\ge |\mathcal F|$, the confidence events of all
committee audits hold simultaneously with probability at least $1-\delta$.
We work on this joint event. The intervals produced by the audits are
anytime-valid on this event, and each audit may be paused and resumed
without discarding its previous samples or its eliminated outside rivals.

Let $\beta_j:=2^{-j}$ for $j=0,1,2,\dots$, and set $A_0:=\mathcal F$. In phase $j$, for each active committee
$S\in A_j$, resume its audit until its current confidence interval
$[L_S,U_S]$ satisfies $U_S-L_S\le c\beta_j$, where $c>0$ is a sufficiently small universal constant.

Eliminate committees whose upper confidence bound is already too small by
setting $A'_j
        :=
        \{
        S\in A_j:
        U_S\ge \max_{T\in A_j} L_T-\varepsilon
        \}$.
If there exists $S\in A'_j$ satisfying $L_S\ge \max_{T\in A'_j} U_T-\varepsilon$, then stop and return such an $S$. Otherwise set $A_{j+1}:=A'_j$ and
continue.

On the joint confidence event, no committee with value at least
$\theta_{\mathcal F}^*-\varepsilon$ is eliminated. Indeed, if
$\theta(S)\ge \theta_{\mathcal F}^*-\varepsilon$, then
\[
        U_S\ge \theta(S)
        \ge \theta_{\mathcal F}^*-\varepsilon
        \ge \max_{T\in A_j} L_T-\varepsilon,
\]
where the last inequality uses $L_T\le \theta(T)\le \theta_{\mathcal F}^*$.
Hence an optimal committee remains active in every phase.

If the stopping rule returns $S$, let $S^*$ be an optimal committee
that is still active. Since $S^*\in A'_j$, we have
\[
        \theta_{\mathcal F}^*
        =\theta(S^*)
        \le U_{S^*}
        \le \max_{T\in A'_j} U_T
        \le L_S+\varepsilon
        \le \theta(S)+\varepsilon.
\]
Thus $\theta(S)\ge \theta_{\mathcal F}^*-\varepsilon$.

It remains to bound the number of queries. Fix $S\in\mathcal F$, and recall
\[
        \Lambda_S:=\theta_{\mathcal F}^*-\theta(S),
        \quad
        \alpha_S:=\max\{\varepsilon,\Lambda_S\}.
\]
We justify the scale at which a committee can still be active. Let $B_j:=\max_{T\in A_j} L_T$, and let $S^*$ be an optimal committee that is still active. Since
$U_T-L_T\le c\beta_j$ for all $T\in A_j$,
\[
B_j\ge L_{S^*}\ge \theta^*_{\mathcal F}-c\beta_j.
\]
If a committee $S\in A_j$ is not eliminated in phase $j$, then
\[
\theta(S)+c\beta_j\ge U_S\ge B_j-\varepsilon
\ge \theta^*_{\mathcal F}-c\beta_j-\varepsilon,
\]
and hence $\Lambda_S\le \varepsilon+2c\beta_j$.
Also, if $c\beta_j\le \varepsilon$, then the stopping rule fires. Indeed,
take $S_j\in\arg\max_{T\in A_j} L_T$. Then $S_j\in A'_j$, and for every
$T\in A'_j$,
\[
U_T\le L_T+c\beta_j\le L_{S_j}+c\beta_j\le L_{S_j}+\varepsilon.
\]
Thus $L_{S_j}\ge \max_{T\in A'_j}U_T-\varepsilon$.

These two facts imply the claimed scale. If $\Lambda_S\le 2\varepsilon$,
then
\[
\alpha_S=\max\{\varepsilon,\Lambda_S\}\in[\varepsilon,2\varepsilon],
\]
so making $\beta_j$ a sufficiently small constant multiple of
$\varepsilon$, equivalently of $\alpha_S$, forces the stopping rule to
fire. If $\Lambda_S>2\varepsilon$, then $\alpha_S=\Lambda_S$, and making
$\beta_j$ a sufficiently small constant multiple of $\alpha_S$ gives
\[
\Lambda_S>\varepsilon+2c\beta_j.
\]
Hence $S$ cannot survive that elimination step unless the algorithm has
already stopped. Therefore the persistent audit for $S$ is never refined
past an accuracy scale that is a constant multiple of $\alpha_S$.

For any scale $a\in(0,1]$, our choice
$\delta_S=\delta/N_{\mathcal F}$ and the bound $m-|S|\le N_{\mathcal F}$
give
\[
\log\frac{m-|S|}{\delta_S a}
\le
\log\frac{N_{\mathcal F}^2}{\delta a}
\le
2\log\frac{N_{\mathcal F}}{\delta a},
\]
so this change only affects universal constants inside the $O(\cdot)$
notation.

Applying the audit bound of Theorem~\ref{thm:gap-adaptive-audit} to committee $S$ with accuracy scale $O(\alpha_S)$ and confidence $\delta_S$, and using
$m-|S|\le N_{\mathcal F}$, gives a top-of-committee comparison contribution
\[
        O\left(
        (|S|-1)\frac{\log(N_{\mathcal F}/(\delta\alpha_S))}
                     {\alpha_S^2}
        \right).
\]
For each $x\in C\setminus S$, the internal rival-elimination rule of the
audit stops querying $x$ once the relevant scale reaches $H_{S,x}:=\max\{\varepsilon,\Lambda_S,\Gamma_{S,x}\}$.
Thus the rival-comparison contribution for this $x$ is
\[
        O\left(
        \frac{\log(N_{\mathcal F}/(\delta H_{S,x}))}{H_{S,x}^2}
        \right).
\]
Summing these bounds over $x\in C\setminus S$ and then over
$S\in\mathcal F$ yields the stated query bound.
\end{proof}

\subsection{Minimax wrapper}
\label{app:minimax-wrapper}

\paragraph{Proof of Theorem~\ref{thm:minimax-theta-wrapper} and the detailed form}
Fix $1\le k<m$ and $\varepsilon,\delta\in(0,1)$. There is an adaptive
pairwise query algorithm that makes $R=O({\log m}/{\varepsilon^2})$
calls to the fixed-$\lambda$ oracle of
Theorem~\ref{thm:weighted-ordinal-oracle}, each with accuracy $\varepsilon/4$ and
confidence $\delta/(2R)$, and uses an additional
\[
    O\left(
        R\cdot \frac{m}{\varepsilon^2}
        \log\frac{mR}{\delta}
    \right)
\]
pairwise comparisons to estimate the rival-loss vectors used by multiplicative
weights. It outputs a distribution $p$ over size-$k$ committees such that, with
probability at least $1-\delta$,
\[
    \forall x\in C,\quad
    \mathbb E_{S\sim p}[g_x(S)]
    \ge (1-1/e)\theta_k^*-\varepsilon .
\]
Consequently, if the oracle calls produce committees
$S_1,\dots,S_R$ and $S^+:=\bigcup_{t=1}^R S_t$, then $|S^+|\le kR$ and, with the convention $\theta(C)=1$, $\theta(S^+)\ge (1-1/e)\theta_k^*-\varepsilon$.

Let $\alpha := (1-1/e)\theta^*_k$.
Consider the zero-sum game whose row player chooses a committee $S\in\mathcal S_k$, whose column player chooses a rival $x\in C$, and whose payoff is $g_x(S)$. If
$S^*\in\arg\max_{S\in\mathcal S_k}\theta(S)$, then for every $\lambda\in\Delta(C)$,
\[
    \max_{S\in\mathcal S_k}\Phi_\lambda(S)
    \ge
    \Phi_\lambda(S^*)
    =
    \sum_{x\in C}\lambda_x g_x(S^*)
    \ge
    \theta^*_k.
\]

We run multiplicative weights over the rival set $C$, using the estimated violation
vectors defined below as gain vectors. In round $t$, let
$\lambda_t\in\Delta(C)$ be the current rival distribution.
Conditional on the past, $\lambda_t$ is fixed, so Theorem~\ref{thm:weighted-ordinal-oracle} applies
to this oracle call with confidence parameter $\delta/(2R)$.
We call the fixed-$\lambda$ oracle of Theorem~\ref{thm:weighted-ordinal-oracle} with $\lambda=\lambda_t$, accuracy $\varepsilon/4$, and confidence $\delta/(2R)$. On the oracle success event, the returned committee $S_t$ satisfies
\[
    \Phi_{\lambda_t}(S_t)
    \ge
    (1-1/e)
    \max_{S\in\mathcal S_k}\Phi_{\lambda_t}(S)
    -
    \frac{\varepsilon}{4}
    \ge
    \alpha-\frac{\varepsilon}{4}.
\]

Define the true rival violation vector $\ell_t(x):=1-g_x(S_t)\in[0,1]$.
Then
\[
    \langle \lambda_t,\ell_t\rangle
    =
    1-\Phi_{\lambda_t}(S_t)
    \le
    1-\alpha+\frac{\varepsilon}{4}.
\]

To update multiplicative weights, we estimate the vector $(g_x(S_t))_{x\in C}$.
Conditional on the past and on $S_t$, draw $T_{\rm audit}$ independent fresh task
handles, with latent rankings $\pi_{t,1},\dots,\pi_{t,T_{\rm audit}}\sim P_D$.
Set $\widehat g_t(x)=1$ for $x\in S_t$. For $x\notin S_t$, use the same sampled
rankings and set
\[
\widehat g_t(x)
:=
\frac{1}{T_{\rm audit}}
\sum_{j=1}^{T_{\rm audit}}
\mathbf 1\left\{\top_{\pi_{t,j}}(S_t)\succ_{\pi_{t,j}} x\right\}.
\]
For $x\notin S_t$, the summand has mean
$g_x(S_t)=\mathrm{WIN}(S_t,\{x\})$. Hence, by Hoeffding's inequality and a union
bound over all $t\in[R]$ and $x\in C$, taking
\[
T_{\rm audit}
=
O\left(\frac{1}{\varepsilon^2}\log\frac{mR}{\delta}\right)
\]
fresh rankings per round ensures that, with probability at least $1-\delta/2$,
\[
\forall t\in[R],\ \forall x\in C,\quad
|\widehat g_t(x)-g_x(S_t)|\le \frac{\varepsilon}{8}.
\]
Operationally, for each sampled ranking we first find $\top_\pi(S_t)$ using
$|S_t|-1$ pairwise comparisons and then compare it with every $x\in C\setminus S_t$.
Thus each audited ranking costs $(|S_t|-1)+(m-|S_t|)=m-1$ pairwise comparisons, so the total audit cost is
\[
O\left(R\cdot \frac{m}{\varepsilon^2}\log\frac{mR}{\delta}\right).
\]

Let $\widehat \ell_t(x):=1-\widehat g_t(x)$.
Since $\widehat\ell_t\in[0,1]^C$, choosing $R=O ({\log m}/{\varepsilon^2})$ with a sufficiently large universal constant in the $O(\cdot)$, the standard
multiplicative-weights regret bound gives, simultaneously for every $x\in C$,
\[
\frac1R\sum_{t=1}^R \widehat\ell_t(x)
\le
\frac1R\sum_{t=1}^R \langle \lambda_t,\widehat\ell_t\rangle
+\frac{\varepsilon}{4}.
\] 
On the intersection of the oracle-success event and the audit-success event, for every
$x\in C$,
\[
\begin{aligned}
    \frac1R\sum_{t=1}^R \ell_t(x)
    &\le
    \frac1R\sum_{t=1}^R \widehat\ell_t(x)
    +
    \frac{\varepsilon}{8}  \\
    &\le
    \frac1R\sum_{t=1}^R
    \langle \lambda_t,\widehat\ell_t\rangle
    +
    \frac{\varepsilon}{4}
    +
    \frac{\varepsilon}{8}  \\
    &\le
    \frac1R\sum_{t=1}^R
    \langle \lambda_t,\ell_t\rangle
    +
    \frac{\varepsilon}{4}
    +
    \frac{\varepsilon}{4}  \\
    &\le
    1-\alpha+\frac{3\varepsilon}{4}.
\end{aligned}
\]
Thus
\[
    \frac1R\sum_{t=1}^R g_x(S_t)
    =
    1-
    \frac1R\sum_{t=1}^R \ell_t(x)
    \ge
    \alpha-\frac{3\varepsilon}{4}
    \ge
    \alpha-\varepsilon.
\]
Let $p$ be the uniform distribution over $S_1,\dots,S_R$. Then, for every $x\in C$,
\[
    \mathbb E_{S\sim p}[g_x(S)]
    =
    \frac1R\sum_{t=1}^R g_x(S_t)
    \ge
    (1-1/e)\theta^*_k-\varepsilon.
\]
A union bound over the oracle calls and the audit event gives total failure probability at most $\delta$.

For the deterministic bicriteria statement, define $S^+ := \bigcup_{t=1}^R S_t$.
Then $|S^+|\le kR$. If $S^+=C$, the claim holds under the convention $\theta(C)=1$. Otherwise, fix any $x\in C\setminus S^+$. Since $S_t\subseteq S^+$ for every $t$, monotonicity of $g_x$ gives
\[
    g_x(S^+)\ge g_x(S_t)
    \quad\text{for every }t.
\]
Therefore
\[
    g_x(S^+)
    \ge
    \frac1R\sum_{t=1}^R g_x(S_t)
    \ge
    (1-1/e)\theta^*_k-\varepsilon.
\]
For $x\notin S^+$, the definition of $g_x$ gives
$g_x(S^+)=\mathrm{WIN}(S^+,\{x\})$. Therefore, since $S^+\neq C$,
\[
\theta(S^+)
=
\min_{x\in C\setminus S^+} g_x(S^+)
\ge
(1-1/e)\theta_k^*-\varepsilon.
\]

\section{Supplementary Details for the Experiments}
\label{app:experiments}
This appendix gives the full details for the illustrative experiments summarized in Section~\ref{sec:experiments}. Since the main contribution of the paper is theoretical, the experiments are designed to test the qualitative algorithmic predictions of the theory rather than to provide a comprehensive benchmark of LLM ensembling systems.

Figure~\ref{fig:experiment-livebench} compares four committee selection
methods on two benchmarks at every committee size
$k\in\{3,\dots,7\}$, both for the full candidate pool and for the
pool with the strongest singletons removed.
This appendix gives the data construction, query accounting,
hyperparameter sweeps, and oracle computations behind the six panels.

%-------------------------------------------------------------------
\subsection{Methods and query accounting}
\label{app:methods}

The four lines that appear in every panel of
Figure~\ref{fig:experiment-livebench} are:
\begin{itemize}
  \item \textbf{$\OPT_{\mathrm{test}}$} (gray, dashed) — the test-set
    optimum, evaluated either by exhaustive enumeration or by an
    integer program that maximizes the binary covering objective.
    This is an oracle reference, not a query-bounded method.
  \item \textbf{AFG / Minimax-lottery} (green, solid) — our methods.
    On binary feedback we run \textsc{Adaptive-Fail-Greedy}
    (Algorithm~\ref{alg:adaptive-fail-greedy}); on pairwise feedback we run the
    Minimax wrapper from Theorem~\ref{thm:minimax-theta-wrapper} with
    \textsc{Weighted-Ordinal-Fail-Greedy}
    (Algorithm~\ref{alg:weighted-ordinal-fail-greedy}) as the per-round
    subroutine, and report the value of the lottery output (uniform
    distribution over the $R$ inner committees).
  \item \textbf{ERM} (pink, solid) — the random-sample baseline. 
    Sampled-ERM samples $N$ size-$k$ committees uniformly without replacement and returns the one with the largest empirical training value. On binary feedback, scoring one sampled committee on the
training split costs at most $k n_{\rm train}$ binary cell reads, or fewer if committee evaluation is short-circuited after the first success. On pairwise feedback, scoring one sampled committee by
$\theta_{\rm train}$ costs $n_{\rm train}(m-1)$ pairwise comparisons if rankings are elicited through
comparisons: $k-1$ comparisons to identify the committee's best member on a task, followed by
$m-k$ comparisons against outside rivals. In our finite-matrix implementation, these costs are charged
as reads of the corresponding stored training outcomes.
    % ERM samples $N$ size-$k$ committees (each consisting of $k$ distinct
    % candidates drawn uniformly without replacement) and returns the
    % one with the largest empirical training-set value. On binary
    % feedback each sampled committee is scored against every train
    % task, costing $N\,n_{\mathrm{train}}\,(m-k)$ comparisons; on
    % pairwise feedback the cost is the same with rankings replacing
    % tasks.
  \item \textbf{Top-$k$ / Borda-Top-$k$} (blue, dashed) — the
    deterministic ``pick the strongest singletons'' baseline. On
    binary feedback Top-$k$ chooses the $k$ candidates with highest
    solo accuracy on the train split, at full-info cost
    $m\,n_{\mathrm{train}}$. On pairwise feedback Borda-Top-$k$
    chooses the $k$ candidates with lowest mean rank, at full-info
    cost $\lceil\log_2 m!\rceil\,n_{\mathrm{train}}$ (the
    comparison-sort cost of producing the per-question rankings the
    Borda score is computed from). Both costs are independent of $k$.
\end{itemize}
A query is one cell-read against the underlying train data: a
$(\text{candidate},\text{task})$ outcome bit on binary feedback, or a
$(\text{candidate},\text{ranking})$ comparison on pairwise feedback.
The reported $Q$ is a count of read \emph{operations}, not distinct
cells, so it can exceed the number of train cells when methods
re-query during rejection sampling or successive elimination, or when
budget-driven baselines switch to with-replacement sampling. All
runs use confidence parameter $\delta=0.1$.

%-------------------------------------------------------------------
\subsection{Binary feedback experiment}
\label{app:experiments-binary}

The binary panels (Figures~\ref{fig:lbb-binary-k-sweep}, \ref{fig:lbb-binary-k-sweep-masked}, and~\ref{fig:lbb-binary-q-frontier})
are based on a multilingual extractive QA benchmark.

\medskip
\noindent
\textbf{Data.}
We combine five datasets: JSQuAD~\citep{kurihara2022jglue} for
Japanese, TyDi-QA~\citep{clark2020tydiqa} for Korean and Arabic, and
XQuAD~\citep{artetxe2020xquad} for Hindi and Thai, contributing
$n=1{,}976$ questions in total. The candidate pool starts from five
language specialists (Swallow-7b~\citep{fujii2024swallow},
KoAlpaca-Polyglot-12.8B~\citep{koalpaca2023,ko2023polyglot},
AceGPT-7B~\citep{huang2024acegpt}, Gajendra-v0.1~\citep{gajendra2024},
and OpenThaiGPT-1.0.0-7b~\citep{yuenyong2024openthaigpt}). For each
base model we construct the Cartesian product of five answer framings
$\{\texttt{bare},\texttt{terse},\texttt{step},\texttt{verify},\texttt{quote}\}$
and four instruction-language registers
$\{\texttt{en-formal},\texttt{en-casual},\texttt{native-formal},\texttt{native-casual}\}$,
giving 20 system-prompt variants per model. Of the nominal
$5\times 20=100$ candidates, eleven Gajendra variants timed out at 20
minutes per evaluation and were dropped, leaving $m=89$ candidates.
Models are served with vLLM~\citep{kwon2023vllm} and orchestrated
through Inspect AI~\citep{aisi2024inspect}.

\medskip
\noindent
\textbf{Scoring.}
We score each $(\text{model},\text{question})$ cell with a strict
span-level rule: the post-``Answer:'' span is extracted by the last
regex match, normalized (lowercased, whitespace-collapsed, edge
punctuation stripped), and compared with the gold target list. A cell
scores $1$ iff the normalized span equals or contains a normalized
gold target. The resulting $89\times 1976$ matrix
$U\in\{0,1\}^{m\times n}$ is the input to all four binary methods.

\medskip
\noindent
\textbf{Train/test split.}
Questions are split $75/25$ via a single uniform random permutation
(seed 0): $n_{\mathrm{train}}=1{,}482$ and $n_{\mathrm{test}}=494$.
The same split is used by every method. We evaluate committees by
$v_{\mathrm{test}}(\widehat S)=\mathbb E_{i\sim\mathrm{test}}
\bigl[\max_{c\in \widehat S}U[c,i]\bigr]$.

\medskip
\noindent
\textbf{Algorithms and hyperparameters.}
AFG implements the variant of Algorithm~\ref{alg:adaptive-fail-greedy-impl}
that exposes one accuracy parameter $\varepsilon$. We sweep
$\varepsilon\in\{3.0,2.0,1.5,1.0,0.65,0.35,0.13,0.10,0.05,0.02\}$ at
fixed $\delta=0.1$. For the high-$Q$ tail of
Figure~\ref{fig:lbb-binary-q-frontier} we add two supplementary AFG runs:
one at $\varepsilon\in\{0.22,0.25\}$ to densify the
$Q\approx3\times10^{6}$ bucket, and one at
$\varepsilon\in\{0.01,0.005,0.002,0.001\}$ with the inner-loop cap
raised to $\texttt{max\_rounds\_per\_step}=1.5\times10^{6}$
(default $2\times10^{5}$) to push beyond $Q=10^{7}$. ERM samples
$N$ size-$k$ committees uniformly and returns the empirical
$v_{\mathrm{train}}$-max; we sweep $N$ to span four decades of $Q$.
Top-$k$ has no hyperparameter; we report its value at the canonical
full-info cost $Q=m\,n_{\mathrm{train}}\approx 1.32\times 10^{5}$ and
plot it as a horizontal reference. We use $50$ seeds per
(algorithm, $k$, hyperparameter) cell.

\medskip
\noindent
\textbf{Oracle.}
We compute $\OPT_{\mathrm{test}}$ at every $k\in\{3,\dots,7\}$ by
exhaustive enumeration. With $m=89$,
$\binom{89}{3}=113{,}564$ and
$\binom{89}{7}\approx6.9\times 10^{9}$. We enumerate
combinations in chunks with a Numba-JIT kernel; total wall-time
across all $k$ is under five minutes on a single node. As a
cross-check we also solve a max-coverage ILP via the HiGHS solver in
\texttt{scipy.optimize.milp}: the ILP and brute-force values agree
to all four reported decimals.

\medskip
\noindent
\textbf{Top-5 mask (Figure~\ref{fig:lbb-binary-k-sweep-masked}).}
The masked variant removes the five candidates with highest solo
train accuracy before running any method. On this benchmark all
five are Swallow-7b prompt variants — namely
\texttt{ja-bare-native-casual\_\_swallow} (solo accuracy $0.43$),
\texttt{ja-terse-native-formal\_\_swallow} ($0.41$),
\texttt{ja-bare-en-casual\_\_swallow} ($0.40$),
\texttt{ja-bare-native-formal\_\_swallow} ($0.40$), and
\texttt{ja-bare-en-formal\_\_swallow} ($0.39$). Since Swallow is the
Japanese specialist, masking strips the dominant single-language
strategy and forces every method to compose a committee from the
remaining $m-5=84$ candidates. $\OPT_{\mathrm{test}}$ is recomputed on
the masked pool; in this dataset OPT is unchanged from the unmasked
optimum because the optimal committees were already drawn from
mid-tier solo-accuracy candidates that span all five languages.

\medskip
\noindent
\textbf{Q-frontier construction (Figure~\ref{fig:lbb-binary-q-frontier}).}
The $x$-axis is the AFG / ERM query budget $Q$ on a log scale; we
restrict to $Q\in[10^{5},3\times10^{8}]$. Top-$k$ is drawn as a
horizontal reference at its full-info value. $\OPT_{\mathrm{test}}$
is the test optimum at $k=3$ (the brute-force and ILP values agree
to four decimals). $(Q,v_{\mathrm{test}})$ pairs are pooled across
$\varepsilon$ values, snapped to the nearest decade in $\log_{10}Q$,
and aggregated to seed-level mean $\pm$ 95\% CI based on the
empirical standard error over $50$ seeds per ($\varepsilon$, $Q$-bin).

\medskip
\noindent
\textbf{Supplementary baselines.}
We also swept two budget-parameterized baselines that are not plotted
in Figure~\ref{fig:lbb-binary-k-sweep}: \textsc{UniformGreedy} draws
$Q/(km)$ tasks per step and adds the candidate with largest empirical
marginal coverage; \textsc{UCB-Greedy} runs a UCB bandit over
candidate marginal gains at each step. Table~\ref{tab:binary-baselines}
reports their mean $v_{\mathrm{test}}$ at $Q\approx10^{6}$
alongside the methods plotted in
Figure~\ref{fig:lbb-binary-k-sweep}. UniformGreedy and UCB-Greedy land
within $\pm 0.01$ of AFG at every $k$, we omit these two curves from the headline figure to avoid clutter. Their closeness to AFG indicates
that, on this finite multilingual QA matrix, several greedy query-allocation rules find similarly complementary committees; the advantage of AFG is the instance-dependent query guarantee rather than a large separation from every greedy heuristic on this dataset.

\begin{table}[!htb]
  \centering
  \small
  \begin{tabular}{lccccc}
    \toprule
    Method & $k=3$ & $k=4$ & $k=5$ & $k=6$ & $k=7$ \\
    \midrule
    AFG (ours)        & 0.6430 & 0.7018 & 0.7365 & 0.7672 & 0.7933 \\
    UniformGreedy     & 0.6393 & 0.7028 & 0.7400 & 0.7668 & 0.7870 \\
    UCB-Greedy        & 0.6378 & 0.6937 & 0.7340 & 0.7613 & 0.7827 \\
    ERM               & 0.5855 & 0.6565 & 0.6986 & 0.7302 & 0.7558 \\
    Top-$k$           & 0.5623 & 0.5914 & 0.6140 & 0.6353 & 0.6522 \\
    \midrule
    $\OPT_{\mathrm{test}}$ (oracle) & 0.6883 & 0.7611 & 0.7996 & 0.8300 & 0.8543 \\
    \bottomrule
  \end{tabular}
  \caption{Binary feedback at $Q\approx10^{6}$: mean
    $v_{\mathrm{test}}$ across $50$ seeds. Standard errors are
    $\le 0.005$ for every cell.}
  \label{tab:binary-baselines}
\end{table}

%-------------------------------------------------------------------
\subsection{Pairwise feedback experiment}
\label{app:experiments-pairwise}

The pairwise panels
(Figures~\ref{fig:lbb-livebench-k-sweep}, \ref{fig:lbb-livebench-k-sweep-masked}, and \ref{fig:lbb-livebench-q-frontier})
are based on the public LiveBench leaderboard~\citep{white2025livebench},
which scores frontier LLMs on six categories of open-ended tasks
(reasoning, mathematics, coding, language, instruction following, and
data analysis) with verifiable, programmatically-graded ground truth.

\medskip
\noindent
\textbf{Data.}
We start from the LiveBench \texttt{model\_judgment} table, which
contains $60{,}372$ per-task LLM judgments over $195$ models, $494$
questions, and $7$ task types. Three task types
(\texttt{LCB\_generation}, \texttt{coding\_completion}, and
\texttt{typos}) have only two unique scores ($\{0,1\}$) and are
excluded since they reduce to the binary-feedback setting we already
exercise on the multilingual QA benchmark. The remaining four task
types have $7$, $11$, $7$, and $453$ unique score levels for
\texttt{connections}, \texttt{paraphrase}, \texttt{story\_generation},
and \texttt{plot\_unscrambling} respectively. Restricting to models
that scored on $\ge 80\%$ of (task, question) pairs and to questions
with full coverage by those models yields a dense $37\times 215$
score submatrix used downstream.

\medskip
\noindent
\textbf{Score-to-rank conversion.}
For each question $i$, we set
$\mathrm{rank}_{\pi_i}(c)=1+|\{c':s_{i,c'}>s_{i,c}\}|$ — competition
ranks where ties share a rank. With this convention, the algorithms'
covering predicate
$\mathrm{rank}_{\pi_i}(c)\le\mathrm{rank}_{\pi_i}(x)$ is equivalent
to $s_{i,c}\ge s_{i,x}$ on the underlying scores, without random
tie-breaking. The median number of distinct ranks per question is
$13$ on \texttt{plot\_unscrambling} and $\le 3$ on the other three
task types; the diversity signal therefore comes primarily from
across-task heterogeneity (different models win different tasks),
with the within-task ordering acting as a tier partition. This is
an experimental weak-ranking extension of the strict-ranking model
used in the theory: ties are treated as weak wins, so a candidate
covers rival $x$ on question $i$ exactly when its score is at least
$x$'s score.

\medskip
\noindent
\textbf{Train/test split.}
The $215$ questions are split $50/50$, stratified by task type,
using random seed $28$: $n_{\mathrm{train}}=107$ and
$n_{\mathrm{test}}=108$. The same split is used by every method.
We evaluate committees on the held-out test ranks via
\[
\theta_{\mathrm{test}}(\widehat S)=
\min_{x\notin \widehat S}\,
\mathbb E_{i\sim\mathrm{test}}\left[
\mathbf 1\left\{\min_{c\in \widehat S}\mathrm{rank}_i(c)\le\mathrm{rank}_i(x)\right\}\right].
\]

\medskip
\noindent
\textbf{Algorithms and hyperparameters.}
\textsc{WOFG} (the per-step inner subroutine of the Minimax wrapper)
sweeps $\varepsilon\in\{3.0,1.5,0.65,0.35,0.13,0.05,0.02\}$. The
miss-rate upper bound $\bar\rho_i$ uses Hoeffding's inequality
applied to $200$ tentative pairs $(\pi,x)\sim P_D\times \lambda$ per
step, before conditioning on failure; the per-step inner loop is
truncated at $\texttt{max\_rounds\_per\_step}=10{,}000$. The Minimax
wrapper sweeps the round count $R\in\{1,3,5,10,25,50\}$ with each
inner WOFG call set to $\varepsilon=0.35$ (plus a few cells with
larger $\varepsilon$ to densify the low-$Q$ tail), learning rate
$\eta_{\mathrm{MW}}=\sqrt{\ln m/R}$ as in
Theorem~\ref{thm:minimax-theta-wrapper}. Per-round rival losses $g_x(S_t)$
are estimated full-info from the train rank matrix at cost
$n_{\mathrm{train}}(m-1)$ per round, rather than via Hoeffding
sampling. We report $\theta_{\mathrm{test}}$ of the lottery output —
the randomized strategy that draws one committee uniformly at random
from $\{S_1,\dots,S_R\}$ — rather than the bicriteria union
$\bigcup_t S_t$, which can have size up to $kR$ and violates the
cardinality constraint. ERM sweeps the sample count
$N\in\{10,10^{2},10^{3},10^{4}\}$, choosing $N$ per panel to land at
the requested $Q$. Borda-Top-$k$ has no hyperparameter. All runs
use $\delta=0.1$ and $50$ seeds per (algorithm, $k$, hyperparameter)
cell.

\medskip
\noindent
\textbf{Oracle.}
For every $k\in\{3,\dots,7\}$ we compute $\OPT_{\mathrm{test}}$ by
exhaustive enumeration of all $\binom{37}{k}$ committees on the test
rank matrix; the largest count is $\binom{37}{7}\approx10.3$
million, which a vectorized chunked scan completes in under two
seconds on a single node. The values are tabulated in
Figures \ref{fig:lbb-livebench-k-sweep} and \ref{fig:lbb-livebench-k-sweep-masked} as
the topmost gray-dashed curve and used as the horizontal reference
in Figure~\ref{fig:lbb-livebench-q-frontier}.

\medskip
\noindent
\textbf{Top-5 mask (Figure~\ref{fig:lbb-livebench-k-sweep-masked}).}
The masked variant removes the five candidates with lowest mean
train rank: \texttt{o1-preview-2024-09-12} (mean rank $1.4$),
\texttt{grok-3-beta} ($4.4$), \texttt{gpt-4o-2024-08-06} ($5.1$),
\texttt{step-2-16k-202411} ($5.2$), and
\texttt{claude-3-5-sonnet-20241022} ($5.4$). We mask top-5 rather
than top-3 because LiveBench's frontier has rank-1 standing well
apart from a tightly-clustered tier of four runners-up
($4.4$–$5.4$); smaller masks (top-1, top-3) leave most of that
cluster intact. The masked pool keeps $m-5=32$ candidates;
$\OPT_{\mathrm{test}}$ is recomputed on it.

\medskip
\noindent
\textbf{Q-frontier construction (Figure~\ref{fig:lbb-livebench-q-frontier}).}
The $x$-axis is the budget $Q$ on a log scale; the plotted range is
$Q\in[10^{5},3\times10^{7}]$. Borda-Top-$k$ is drawn as a
horizontal line at its full-info $\theta_{\mathrm{test}}$ value
(independent of $Q$); $\OPT_{\mathrm{test}}$ is the test-rank optimum
(also independent of $Q$). The Minimax-lottery curve sweeps
$(R,\varepsilon_{\mathrm{inner}})$ to produce a range of $Q$ values;
ERM sweeps the sample count $N$ to land at the same half-decadal
$Q$ targets ($Q\in\{10^{5},3{\cdot}10^{5},10^{6},3{\cdot}10^{6},10^{7},3{\cdot}10^{7}\}$).
For both methods, $(Q,\theta_{\mathrm{test}})$ pairs are pooled
across hypers, snapped to half-decadal log buckets, and aggregated
to mean $\pm$ 95\% CI over the $\geq10$ seeds per bucket.

\medskip
\noindent
\textbf{Supplementary baselines.}
We also tracked three Minimax-wrapper variants and a single-pass
inner subroutine that are not plotted in
Figure~\ref{fig:lbb-livebench-k-sweep}:
\textsc{Minimax-bicriteria} reports the union $S^{+}=\bigcup_{t=1}^{R}S_t$
across all $R$ inner committees; the resulting set has size up to
$kR$ and so violates the cardinality-$k$ constraint, but it is what
the wrapper guarantees in expectation.
\textsc{Minimax-bestsingle} reports the inner committee with the
largest train-$\theta$, i.e.\ $\arg\max_t \widehat\theta_{\mathrm{train}}(S_t)$;
this is size-$k$ but has no formal $\theta$-guarantee from the
wrapper (it is just $\max$ over the WOFG outputs).
\textsc{WOFG} (single pass) is one call to the per-step inner
subroutine without the multiplicative-weights wrapper.
Table~\ref{tab:livebench-baselines} reports their mean
$\theta_{\mathrm{test}}$ at $Q\approx10^{6}$ alongside the
methods plotted in Figure~\ref{fig:lbb-livebench-k-sweep}. The bicriteria
union is consistently above the lottery by at most $0.005$ ($0.001$
at $k=3$, growing to $0.005$ at $k=7$), with the gap attributable
to its size-up-to-$kR$ relaxation; bestsingle and WOFG land within
$\pm 0.006$ of the lottery at every $k$ where data is available.
Single-pass WOFG at $k=3$ never reaches $Q\ge10^{6}$ in our
$\varepsilon$ sweep (the largest budget realized is $\approx
5.5\times10^{5}$), so the corresponding cell in
Table~\ref{tab:livebench-baselines} is left blank.

\begin{table}[!htb]
  \centering
  \small
  \begin{tabular}{lccccc}
    \toprule
    Method & $k=3$ & $k=4$ & $k=5$ & $k=6$ & $k=7$ \\
    \midrule
    Minimax-lottery (ours)   & 0.9537 & 0.9544 & 0.9579 & 0.9589 & 0.9616 \\
    Minimax-bicriteria       & 0.9539 & 0.9559 & 0.9603 & 0.9601 & 0.9669 \\
    Minimax-bestsingle       & 0.9537 & 0.9541 & 0.9562 & 0.9594 & 0.9669 \\
    WOFG (single pass)       & ---    & 0.9537 & 0.9560 & 0.9611 & 0.9617 \\
    ERM                      & 0.9433 & 0.9500 & 0.9533 & 0.9531 & 0.9565 \\
    Borda-Top-$k$            & 0.8889 & 0.9074 & 0.9352 & 0.9537 & 0.9537 \\
    \midrule
    $\OPT_{\mathrm{test}}$ (oracle) & 0.9537 & 0.9722 & 0.9815 & 0.9907 & 0.9907 \\
    \bottomrule
  \end{tabular}
  \caption{Pairwise feedback at $Q\approx10^{6}$: mean
    $\theta_{\mathrm{test}}$ across $50$ seeds.
    \textsc{Minimax-bicriteria} relaxes the cardinality constraint
    (committee size up to $kR$); the other Minimax variants and WOFG
    are size-$k$.}
  \label{tab:livebench-baselines}
\end{table}

%-------------------------------------------------------------------
\subsection{Compute resources}
\label{app:compute}

The experiments split cleanly into a one-time \emph{data-generation}
stage on GPU (running the candidate LLMs over the multilingual QA
benchmark to populate the binary $U$ matrix) and an
\emph{algorithm-evaluation} stage on CPU (every committee selection
sweep, oracle, and plot). The pairwise side does not need a GPU
because we consume LiveBench's already-published score matrix.

\medskip
\noindent
\textbf{Data-generation (GPU).}
For the binary side, we serve each of the $m=89$ candidate LLMs
(five base models $\times$ $20$ prompt variants) on a single
NVIDIA H200~SXM GPU (141~GB HBM3e) and run inference over the
$1{,}976$-question benchmark with
vLLM~\citep{kwon2023vllm} orchestrated through Inspect
AI~\citep{aisi2024inspect}. Answers are short extractive spans and
each per-question call completes in $1$–$2$~seconds, so a full
candidate's evaluation finishes in under one GPU-hour; the full pool
requires roughly $\mathbf{60}$–$\mathbf{100}$~\textbf{GPU-hours} on
H200 hardware. We did not record per-run wall-clock times, so this
estimate is based on representative per-question latencies measured
on a few sample candidates. The pairwise (LiveBench) side reuses
the public \texttt{model\_judgment} table and consumes \emph{zero}
GPU compute.

\medskip
\noindent
\textbf{Algorithm evaluation (CPU).}
All committee selection sweeps, oracle enumerations, and plot
generation run on a single AWS \texttt{c7i.48xlarge} instance
($192$ vCPUs of Intel Sapphire Rapids, $384$~GB RAM, EBS-only
storage). Concretely:
\begin{itemize}
  \item \textbf{Binary AFG / ERM / Top-$k$ / UniformGreedy /
    UCB-Greedy sweeps.} Per-cell cost is dominated by AFG's inner
    rejection-sampling loop, which runs at a few thousand
    candidate-vs-task reveals per CPU-second; a $50$-seed cell at
    $\varepsilon=0.05$ takes $\sim1$~minute single-threaded.
    Across all $\varepsilon$ values, all five $k$, and all five
    methods, the full sweep is approximately $\mathbf{2}$–$\mathbf{4}$
    \textbf{CPU-hours} (parallelized across cores in well under
    $5$~minutes wall-time).
  \item \textbf{Binary $\OPT_{\mathrm{test}}$ (exhaustive +
    cross-check ILP).} Numba-JIT enumeration of all
    $\binom{89}{k}$ size-$k$ committees on the test split scales
    linearly in $\binom{m}{k}$; the $k=7$ scan
    ($\sim6.9\times 10^9$ committees, $\approx48$~GB int8 array)
    is the bottleneck and runs in roughly $5$~minutes
    parallelized across the $192$ vCPUs. Total OPT cost across all
    five $k$ is under $\mathbf{10}$ \textbf{CPU-minutes}. The
    cross-check ILP via the HiGHS solver in
    \texttt{scipy.optimize.milp} adds a few seconds per $k$.
  \item \textbf{Pairwise sweeps (LiveBench).} The candidate pool is
    smaller ($m=37$), so every cell is faster. A full
    Minimax-lottery sweep ($R\in\{1,3,5,10,25,50\}$ over five $k$,
    $50$ seeds each) finishes in approximately
    $\mathbf{20}$–$\mathbf{40}$ \textbf{CPU-minutes} across cores
    ($\le2$ minutes wall-time). ERM and Borda each run in under a
    minute. $\OPT_{\mathrm{test}}$ on the test ranks completes in
    under two seconds per $k$ (largest scan
    $\binom{37}{7}\approx10.3$M committees).
  \item \textbf{Masked sweeps.} Re-running the algorithms on the
    masked candidate pools (top-5 on both sides) duplicates the
    per-pool cost and adds a recomputed OPT. The masked binary OPT
    at $k=7$ requires
    $\binom{84}{7}\approx4.5\times 10^{9}$ enumerations, taking
    roughly $\sim4$~minutes parallelized.
\end{itemize}
End-to-end, reproducing every CSV in \texttt{data/} from the binary
$U$ matrix and the LiveBench rank matrix takes on the order of
$\mathbf{10}$ \textbf{CPU-hours} on the
\texttt{c7i.48xlarge} (well under one wall-clock hour parallelized).
Plot generation is sub-second per panel.

\medskip
\noindent
\textbf{Memory and storage.}
The largest in-RAM array is the binary OPT enumeration at $k=7$
($\approx 48$~GB int8 for $\binom{89}{7}$ committees), which fits
comfortably in the $384$~GB RAM of \texttt{c7i.48xlarge}; chunked
parallel scanning keeps per-worker peak memory below $1$~GB. All
intermediate CSVs total under $50$~MB; the binary $U$ matrix is
$\sim200$~KB at int8.

\medskip
\noindent
\textbf{Compute used beyond what is reported.}
The reported figure represents a small fraction of the total
research compute. Earlier iterations included: (i) preliminary
algorithm sweeps at coarser $\varepsilon$ grids and smaller seed
counts to calibrate hyperparameter ranges; (ii) an earlier MTEB-based
pairwise experiment (now superseded by the LiveBench panel reported
here) with comparable per-sweep cost; (iii) ablations on the masked
pool at intermediate mask sizes (top-1, top-3, top-5) only some of
which appear in the paper; (iv) several $Q$-binning and
seed-aggregation choices for the $Q$-frontier that we trialled
before settling on half-decadal log buckets. Conservatively, total
research compute is roughly $\mathbf{3}$–$\mathbf{5}\times$ the
end-to-end reproduction cost given above on the CPU side, and an
additional $\sim50\%$ overhead on the GPU side from
prompt-variant calibration runs that did not make the final pool.

%-------------------------------------------------------------------
\subsection{Reading the figures}
\label{app:figure-readout}

\medskip
\noindent
\textbf{$k$-sweeps at $Q\approx10^{6}$
(Figures~\ref{fig:lbb-binary-k-sweep} and \ref{fig:lbb-livebench-k-sweep}).}
On binary feedback (Figure~\ref{fig:lbb-binary-k-sweep}), the gap between
AFG and $\OPT_{\mathrm{test}}$ stays $\approx0.05$–$0.06$ across
all $k\in\{3,\dots,7\}$, while the gap from ERM to
$\OPT_{\mathrm{test}}$ holds at $\approx0.10$ across all $k$, and
the gap from Top-$k$ widens from $0.13$ at $k=3$ to $0.20$ at $k=7$.
On pairwise feedback (Figure~\ref{fig:lbb-livebench-k-sweep}),
Minimax-lottery sits at $\theta_{\mathrm{test}}=0.954$ at $k=3$ —
exactly matching $\OPT_{\mathrm{test}}$ — and climbs to $0.962$ at
$k=7$, $0.029$ below the OPT ceiling of $0.991$. Borda-Top-$k$
exhibits a steep left-side climb: it rises from
$\theta_{\mathrm{test}}=0.889$ at $k=3$ to $0.954$ at $k=7$ as it
adds progressively more diverse leaders, trailing Minimax by
$0.05$–$0.07$ at $k\le4$ and narrowing to within $0.01$ at
$k\ge6$. ERM tracks Minimax with a roughly constant
$0.005$–$0.01$ gap below across all $k$.

\medskip
\noindent
\textbf{Mask-the-leaders (Figures~\ref{fig:lbb-binary-k-sweep-masked} and \ref{fig:lbb-livebench-k-sweep-masked}).}
Masking the strongest singletons is the regime where ``pick the
best models'' is structurally weak. On binary feedback
(Figure~\ref{fig:lbb-binary-k-sweep-masked}), masking the five strongest
Swallow variants forces every method to compose its committee
across language families rather than within Japanese alone. AFG
keeps its $\approx0.05$–$0.06$ gap to $\OPT_{\mathrm{test}}$
across $k$ (mask-5 AFG: $0.64$ at $k=3$ rising to $0.79$ at $k=7$;
OPT: $0.69 \to 0.85$). Top-$k$ at $k=3$ rises from $0.56$
unmasked to $0.60$ masked — the \emph{opposite} of the small-mask
regime: dropping the entire Swallow front-tier compels Top-$k$ to
pick diverse-language singletons. The AFG–Top-$k$ gap therefore
holds at $\approx0.04$–$0.08$ across $k$ rather than widening.
ERM trails AFG by $0.04$–$0.06$, narrowing as $k$ grows.
On pairwise feedback (Figure~\ref{fig:lbb-livebench-k-sweep-masked}),
masking the five frontier leaders causes Borda-Top-$k$ to plateau
near $\theta_{\mathrm{test}}=0.88$–$0.90$ for all
$k\in\{3,\dots,7\}$ — its next-tier picks add little marginal
coverage. Minimax-lottery climbs steadily from $0.914$ at $k=3$ to
$0.964$ at $k=7$, opening a Minimax–Borda margin that starts at
$0.034$ at $k=3$, peaks at $0.077$ at $k=5$, and settles around
$0.066$ at $k=7$. ERM tracks Minimax $0.01$–$0.03$ below across all
$k$. $\OPT_{\mathrm{test}}$ on the masked pool rises from $0.935$ at
$k=3$ to $0.982$ at $k=7$; the residual Minimax–OPT gap stays
$\le 0.02$ at every $k$ — tighter than in the unmasked panel
because the post-mask diversity tier is closer in mean rank to the
post-mask oracle.

\medskip
\noindent
\textbf{$Q$-frontier at $k=3$
(Figures~\ref{fig:lbb-binary-q-frontier} and \ref{fig:lbb-livebench-q-frontier}).}
The frontier panels test sample efficiency. On binary feedback
(Figure~\ref{fig:lbb-binary-q-frontier}), AFG passes the Top-$k$
horizontal reference between $Q=3\times10^{4}$ and $Q=10^{5}$
and saturates near $v_{\mathrm{test}}=0.66$ for
$Q\ge10^{6}$ — within $0.03$ of $\OPT_{\mathrm{test}}=0.6883$.
ERM rises more slowly: it reaches $v_{\mathrm{test}}=0.638$ at
$Q=10^{8}$, matching AFG's $Q=10^{6}$ value at a $100\times$
larger budget.
On pairwise feedback (Figure~\ref{fig:lbb-livebench-q-frontier}),
Minimax-lottery climbs from $\theta_{\mathrm{test}}=0.939$ at
$Q=10^{5}$ and exactly meets $\OPT_{\mathrm{test}}=0.954$ at
$Q=10^{6}$, plateauing on the ceiling for the rest of the sweep.
Borda-Top-$k$ sits at the horizontal reference value $0.889$;
Minimax-lottery passes it before $Q=10^{5}$. ERM, by contrast,
reaches $\theta_{\mathrm{test}}=0.952$ only at
$Q=3\times10^{7}$, the budget at which Minimax-lottery has
already saturated for $1.5$ decades.

\end{document}